\documentclass[reqno,12pt,letterpaper]{amsart}
\usepackage[foot]{amsaddr}
\usepackage{amsmath,amssymb,amsthm,graphicx,mathrsfs,url}
\usepackage[usenames,dvipsnames]{xcolor}
\usepackage{amsxtra}
\usepackage{tikz}
\usepackage{amscd}
\usetikzlibrary{arrows.meta}
\usepackage[normalem]{ulem}
\usepackage{wasysym} 
\usepackage{graphicx}
\usepackage{subcaption}
\usepackage[normalem]{ulem}
\usepackage{a4wide,color,eucal,enumerate,mathrsfs}
\usepackage[normalem]{ulem}
\usepackage{amsmath,amssymb,epsfig,amsthm} 
\usepackage[latin1,utf8]{inputenc}
\usepackage{csquotes}\MakeOuterQuote{"}
\usepackage{psfrag}
\def\arXiv#1{\href{http://arxiv.org/abs/#1}{arXiv:#1}}
\usepackage{array}
\newcolumntype{P}[1]{>{\centering\arraybackslash}m{#1}}
\usepackage{comment}
\usepackage[colorlinks=true,linkcolor=Red,citecolor=Green]{hyperref}
\usepackage[capitalize]{cleveref}

\newcommand{\floor}[1]{\left\lfloor {#1} \right\rfloor}

\def\?[#1]{\textbf{[#1]}\marginpar{\Large{\textbf{??}}}}

\def\smallsection#1{\smallskip\noindent\textbf{#1}.}
\let\epsilon=\varepsilon 
\newcommand{\R}{{\mathbb R}}
\newcommand{\RR}{{\mathbb R}}
\newcommand{\NN}{{\mathbb N}}
\newcommand{\N}{{\mathbb N}}
\newcommand{\CC}{{\mathbb C}}

\newcommand{\ZZ}{{\mathbb Z}}
\newcommand{\eps}{\varepsilon}
\newcommand{\delete}[1]{}

\newtheorem{theo}{Theorem}[section]	
\newtheorem{prop}[theo]{Proposition}	
\newtheorem{defi}[theo]{Definition}

\newtheorem{lemm}[theo]{Lemma}
\newtheorem{corr}[theo]{Corollary}
\newtheorem{rem}[theo]{Remark}
\newtheorem{ex}[theo]{Example}
\numberwithin{equation}{section}

\DeclareMathOperator{\Spec}{Spec}

\let\Im=\Imag

\newcommand{\norm}[1]{\left\lVert #1 \right\rVert}

\let\Re=\Real

\def\indic{\operatorname{1\hskip-2.75pt\relax l}}
\usepackage{scalerel}

\newcommand\reallywidehat[1]{\arraycolsep=0pt\relax%
\begin{array}{c}
\stretchto{
  \scaleto{
    \scalerel*[\widthof{\ensuremath{#1}}]{\kern-.5pt\bigwedge\kern-.5pt}
    {\rule[-\textheight/2]{1ex}{\textheight}} 
  }{\textheight} %
}{0.5ex}\\           
#1\\                 
\rule{-1ex}{0ex}
\end{array}
}

\def\bint{{\ifinner\rlap{\bf\kern.35em--}
\int\else\rlap{\bf\kern.45em--}\int\fi}\ignorespaces}

\def\bbint{{\ifinner\rlap{\bf\kern.35em--}
\hspace{0.078cm}\int\else\rlap{\bf\kern.45em--}\int\fi}\ignorespaces}

\title[Convergence rates for the Trotter splitting]{Convergence rates for the Trotter splitting for unbounded operators} 

\author{Simon Becker$^*$}
\email{simon.becker@math.ethz.ch}
\thanks{$^*$Corresponding author, ETH Zurich, 
Institute for Mathematical Research, 
Rämistrasse 101, 8092 Zurich, 
Switzerland}

\author{Niklas Galke}

\author{Lauritz van Luijk}

\author{Robert Salzmann}

\date\today
\begin{document}

\begin{abstract}
    We study convergence rates of the Trotter splitting
    $$e^{A+L} = \lim_{n \to \infty} \Big(e^{L/n} e^{A/n}\Big)^n$$
    in the strong operator topology.
    In the first part, we use complex interpolation theory to treat generators $L$ and $A$ of contraction semigroups on Banach spaces, with $L$ relatively $A$-bounded.
    In the second part, we study unitary dynamics on Hilbert spaces and develop a new technique based on the concept of energy constraints.
    Our results provide a complete picture of the convergence rates for the Trotter splitting for all common types of Schr\"odinger and Dirac operators, including singular, confining and magnetic vector potentials, as well as molecular many-body Hamiltonians in dimension $d=3$.
    Using the Brezis-Mironescu inequality, we derive convergence rates for the Schrödinger operator with $V(x)=\pm |x|^{-a}$ potential.
    In each case, our conditions are fully explicit.
\end{abstract}
\subjclass[2020]{35Q40, 47D06, 81Q05} 
\thanks{Communicated by Arieh Iserles}
\keywords{Trotter splitting, convergence rates, unbounded operators, Schrödinger equation}

\maketitle 

\tableofcontents

\section{Introduction}
We derive explicit convergence rates for the Trotter splitting formula, which, under suitable assumptions on $A$ and $L$, states that
\begin{equation}
\label{eq:Trotter}
e^{A+L} = \lim_{n \to \infty} \Big(e^{L/n} e^{A/n}\Big)^n
\end{equation}
in the strong operator topology.  
Originally developed by Lie for matrices, it has been extended to operators by Trotter \cite{T59}, Chernoff \cite{C68}, and Kato \cite{K78}. Nelson \cite{N64} clarified the connection between the splitting formula and Feynman's path integral for Schr\"odinger dynamics.

The formula is also of key importance in numerical applications \cite{Paz83, S94, JL00} and quantum chemistry \cite{K88,LWG+96,BMW+15}. We also refer to the overview article \cite{MQ02}. {In all of these works the underlying key computational observation is that while the exponentials $e^{\tau A}$ and $e^{\tau L}$ can be evaluated efficiently and to high accuracy, computing $e^{\tau(A+L)}$ directly can be significantly more expensive. This motivates splitting methods such as the Trotter product formula.} In particular, in quantum chemistry, it is natural to implement Schr\"odinger's evolution of a complex molecule using  \eqref{eq:Trotter}{ \, with operators $A = -\frac{1}{i} \Delta$ and $L = \frac{1}{i}V$,} separating the kinetic and potential energy terms of the Hamiltonian. Because both components of the dynamics can be efficiently handled via multiplication operators in Fourier and configuration spaces, respectively, this method offers a practical approach to simulate the full molecular dynamics.

{Another important application of the Trotter product formula lies in quantum computing with it being at the heart of many promising quantum algorithms: As already conjectured by Feynman \cite{F82}, simulating the time evolution of quantum systems directly on quantum devices could potentially lead to significant speedups over fully classical simulations. Following this idea and based on the Trotter product formula \eqref{eq:Trotter},  Lloyd  proposed a quantum algorithm to simulate Schrödinger's time evolution of local Hamiltonians  efficiently on a quantum computer \cite{L96}. This result has then been extended to also efficiently simulate the dynamics of sparse Hamiltonians in \cite{BACS07,B+14}. In all of these works, the idea to simulate the desired dynamics is to split its generator into a sum of more tractable terms, whose dynamics can individually be implemented, and then employ \eqref{eq:Trotter}. Simulating quantum dynamics in this way has become a key subroutine for many other quantum algorithms like quantum phase estimation \cite{K95} and the famous HHL algorithm for solving systems of linear equations \cite{HHL09}.}

{

An alternative perspective on the relevance of our quantitative study arises from numerical linear algebra: For $n = 2^m$, define
\[
F_m(A, L) = \left(e^{2^{-m}L} e^{2^{-m}A}\right)^{2^m}.
\]
Once $G_m(A, L) = e^{2^{-m}L} e^{2^{-m}A}$ is computed, $F_m(A, L)$ can be obtained via $m$ successive squaring steps. This corresponds to the scaling-and-squaring method \cite{H05} widely used in practice, providing a quantitative link between splitting methods and classical techniques from numerical linear algebra.
}

{When the operators $A$ and $L$ are bounded, convergence of the Trotter product formula \eqref{eq:Trotter} is straightforward using direct Taylor expansion of the exponentials. However, for unbounded operators, as is typical in partial differential equations on unbounded domains like Schrödinger's evolution of molecular systems mentioned above, the analysis becomes more subtle and has been studied extensively, see, e.g. the overview article \cite{BCM24} and in particular Section 6.3 for applications in quantum mechanics). However, these results are often obtained under strong boundedness assumptions, see \cite[Assumption 2.2]{HO09} or the bounded commutator in Section 3 in \cite{IK24} that are too restrictive to cover examples from quantum theory. The principal contribution of this work lies in its rigorous treatment of unbounded operators, including applications to the Schrödinger equation for physically relevant systems such as the hydrogen atom and harmonic oscillator, where specialized functional analytic frameworks are required (see Section 7 in \cite{LL20}).}

For applications, convergence rates are important as they give a guarantee on the quality of the approximation and furthermore provide runtime guarantees on respective algorithms.
However, uniform convergence rates, i.e., convergence rates in operator norm, are typically not feasible since the Trotter product formula simply does not converge uniformly in general \cite[Sec.~4]{I03}.
This motivates the study of convergence rates in the strong operator topology. That is, for a vector $x$, we are interested in the speed of convergence
\begin{equation}
\label{eq:TrotterConvIntro}
    \|{(e^{L/n}e^{A/n})}^nx-e^{A+L}x\| \to0.
\end{equation}
Although a convergence rate of $\mathcal O(n^{-1})$ is guaranteed in finite dimensions, or more generally for bounded generators, this does not hold in general:
In \cite{BFHJY23} it was shown that the splitting of the Coulomb Hamiltonian $H=-\Delta-|x|^{-1}$, describing the hydrogen atom, converges with $\mathcal O(n^{-1/4})$ on the (eigenstates in) s orbitals, $\mathcal O(n^{-3/4})$ on the p orbitals and with $\mathcal O(n^{-1})$ on the d and higher orbitals. 
Their technique of establishing strong convergence rates is restricted to the Hilbert space setting and to eigenvectors of $A+L$. In this article, we relax these restrictions by proving strong convergence rates on sufficiently regular vectors also in the Banach space setting, recovering, in particular, a form of $\mathcal O(n^{-1/4})$ convergence for the ground state of the hydrogen atom.

Furthermore, for more complicated atomic and molecular configurations, analytic expressions for the eigenstates of the corresponding Hamiltonians are usually not available and hence an analysis similar to \cite{BFHJY23} cannot be employed. Therefore, it is natural to ask the following:  
\begin{center}
\emph{How fast does the Trotter-Kato product converge for general molecular Hamiltonians?}
\end{center}
We answer these questions by providing a complete analysis of the Trotter error for general molecular Hamiltonians with $N$ electrons and $M$ nuclei. Here, again for sufficiently regular vectors, we also find a form of $\mathcal{O}(n^{-1/4})$ speed of convergence similar to the simplest case of the hydrogen atom.

Our article features two directions. In the first part of the article, Section \ref{sec:RelBoundTrotter}, our guiding principle is to develop a theory to obtain convergence rates under the assumption that $L$ is "small" compared to $A$. We gauge the smallness using a relative boundedness condition of the operators. In this case, the appropriate notion of regularity is captured using special interpolation spaces defined through $A$, which are called \emph{Favard spaces} and which we can characterize explicitly in many relevant cases. For Schr\"odinger operators, for instance, in terms of standard Sobolev spaces. We demand that $L$ does not decrease this form of regularity by too much, i.e., when applied to a vector of high regularity, the output should remain regular, though possibly to a lesser extent.
After establishing $\mathcal O(n^{-1})$-convergence rates for \eqref{eq:TrotterConvIntro} for sufficiently regular vectors $x$ and for suitable $L$, we prove $\mathcal O(n^{-\delta})$-convergence rates with $0<\delta< 1$ on less regular vectors and for very general, but still relatively $A$-bounded $L$.
The guiding motive for this first part of the article is to understand the convergence rates of the Trotter-splitting of Schr\"odinger operators with singular potentials, in particular for Hamiltonians modeling atomic and molecular configurations. For instance, we find for the three-dimensional {Coulomb potential}
\begin{equation}
    \big\|\big(e^{-it|x|^{-1}/n}e^{-it\Delta/n}\big)^n\psi - e^{it(-\Delta-|x|^{-1})}\psi\big\|\lesssim n^{-1/4+\varepsilon}\|\psi\|_{H^{2}(\R^3)}
\end{equation}
for any fixed $\eps>0$ (cp.\ Corollary~\ref{cor:Coulomb}).
{Numerical evidence from \cite{BFHJY23} suggests that, on the ground state, the true convergence rate is indeed $n^{-1/4}$. Thus, we expect our estimate to be essentially tight. Proving an analytical lower bound, however, remains an open problem.
}

In the second part of the article, in Section \ref{sec:energy_limited}, the guiding motive is to understand Schr\"odinger operators with confining or trapping potentials. In this case, the kinetic energy and potential energy are not in any sense small with respect to one another. A standard example of such an operator is the quantum harmonic oscillator. To treat the Trotter splitting of such operators, we develop a new approach based on the concept of \emph{energy constraints} \cite{W17,S18,S20,vL24}.
This approach differs from the first one due to the different nature of the perturbation. 
As a prototypical example of the results obtained with this approach, we mention the following convergence rate
\begin{equation}\label{eq:intro_scheiss_oszillator}
    \big\|{\big(e^{-i\frac tn\Delta} e^{i\frac tnx^2}\big)^n\psi -e^{it(-\Delta+x^2)}\psi}\big\| \le \frac{6t^2}{n}\norm{(-\Delta+x^2)\psi}
\end{equation}
for the Trotter splitting of the harmonic oscillator on $L^2(\RR)$ (cp.\ \cref{thm:scheiss_oszillator}).

Together, the two approaches yield a fairly complete understanding of Trotter convergence rates for closed quantum systems.
The general nature of the first approach allows us to apply it to a much larger class of dynamical systems beyond just closed quantum systems. We outline this in the first few sections of the article. The approach outlined in the second part of the article is limited to unitary dynamics on Hilbert spaces.

\subsection{Techniques}
Both approaches for providing convergence rates for the Trotter product formula \eqref{eq:Trotter} in Sections~\ref{sec:RelBoundTrotter} and~\ref{sec:energy_limited} follow similar proof techniques, but diverge at certain points, as we explain in the following:

 \emph{Step 1: Key commutator bound.} 
 In both Sections~\ref{sec:RelBoundTrotter} and \ref{sec:energy_limited}, the first step of bounding the Trotter error is to use a telescoping sum argument, which gives for $x$ of sufficient regularity
 \begin{align}
 \label{eq:KeyCommIntro}
    \left\|{(e^{L/n}e^{A/n})}^nx-e^{A+L}x\right\| \le \sup_{s,\tau\in[0,1]}\left\|[L,e^{sA/n}]x_\tau\right\|.
 \end{align}
 Here, we can either have that $x_\tau$ is the original element $x$ evolved by the joined dynamics, i.e. $x_\tau := e^{\tau(A+L)}x,$ or by the product dynamics, i.e. $x_\tau := (e^{L/n}e^{A/n})^{\floor{\tau n}} x.$ We call this the \emph{key commutator bound}, which can be found in Lemma~\ref{lem:SemiTrottCommBound} or by combining Lemma~\ref{lem:telescope} and the relevant steps in the proof of Theorem~\ref{thm:EL_trotter}.

\smallskip

\emph{Step 2: Stability of regularity}
Next, we need to ensure that the evolved element $x_\tau$ has a regularity comparable to the initial $x.$  The way we quantify regularity depends on the context: In Section~\ref{sec:RelBoundTrotter}, we are in the regime in which $A$ is the dominant generator as $L$ is relatively bounded with $A$-bound $<1$. Hence, in this case, it is sensible to focus on notions of regularity defined through the generator $A.$ In particular, for that, we consider the Favard spaces $F_\gamma(A)$ which for $\gamma\in[0,1]$ are exactly given by elements $x$ such that $\|(e^{tA}-I)x\| =\mathcal{O}(t^{\gamma}).$ Hence, higher values of $\gamma$ correspond to higher regularity of the element $x.$ For $x\in F_\gamma(A)$ the relative boundedness assumptions on $L$ employed give, by the stability result on Favard spaces Lemma~\ref{lemm:stability}, that also $x_\tau\in F_\gamma(A)$ where we chose $x_\tau= e^{\tau(A+L)}x$ in the above.

On the other hand, in Section~\ref{sec:energy_limited}, the role of $A$ and $L$ is symmetric. 
The right notion of regularity is given by the \emph{energy} of a vector $x$ as measured by a positive self-adjoint operator $G$, the `reference Hamiltonian'.
We are free to choose this reference Hamiltonian as long as the dynamics generated { by }$A$ and $L$ are energy-limited with respect to $G$ in the sense of \cite{vL24}.
This ensures that $x_\tau$ above is of comparable regularity as the initial element $x.$

\smallskip

\emph{Step 3: Bounding the key commutator.}
As the final step, we need to estimate the key commutator, which appears on the right-hand side of \eqref{eq:KeyCommIntro}. In Section~\ref{sec:RelBoundTrotter} we write 
\begin{align}
\label{eq:CommIntro}
    [L,e^{sA/n}]x_\tau = L(e^{sA/n}-I)x_\tau - (e^{sA/n}-I)Lx_\tau
\end{align}
and bound each term individually. For the first term, we assume that the singularity of $L$ is controlled when the operator is restricted to elements of sufficient regularity. More precisely, we assume that $L$ is bounded in the Favard space $F_\alpha(A)$ for some $0\le\alpha<\gamma$ from which we see that the first term decays as $\mathcal{O}(n^{-(\gamma-\alpha)}).$ For the second term, we assume that $L$ does not decrease the regularity of $x_\tau \in F_\gamma(A)$ by too much. In particular, if we still have $L x_\tau\in F_\beta(A)$ for some $0<\beta\le \gamma$ the second term in \eqref{eq:CommIntro} decays as $\mathcal{O}(n^{-\beta}).$ Both assumptions on how $L$ interacts with the notion of regularity dictated by $A$ can explicitly be verified for many interesting examples, e.g. Schr\"odinger operators with Coulomb-type potentials and molecular Hamiltonians as shown at the end of Section~\ref{sec:unitary}.

In Section~\ref{sec:energy_limited}, we choose a different approach to bound the right-hand side of \eqref{eq:KeyCommIntro} by using
\begin{align}
\label{eq:FullCommiIntro}
    \|[L,e^{sA/n}]x_\tau\| \le \frac{s}{n}\sup_{u\in[0,s]}\|[L,A]x_{\tau,u}\|,
\end{align}
where $x_{\tau,u} := e^{uA/n}x_\tau.$ 
Under the above assumption of energy limitedness and given that the commutator $[L,A]$ is bounded with respect to $G^{1/2},$ the norm appearing on the right-hand side of \eqref{eq:FullCommiIntro} is indeed finite and we find that the Trotter error decays as $\mathcal{O}(n^{-1}).$ 

\subsection{Related works} 

Our article outlines an approach to obtain initial state-dependent convergence rates for such splittings. Trotter schemes with applications to quantum mechanics have been improved in various directions, for example, for bosonic systems \cite{M24}, lattice systems \cite{BL22}, and in the context of the quantum Zeno effect \cite{EI05,MW19,BDS21,MR23,S24}. Our approach is different in the sense that we try to provide a framework that covers a wide range of singular splittings that are not covered by most results in the literature. Although little about convergence rates beyond standard cases seems to be known in the quantum setting, a more elaborate theory is available for strictly dissipative dynamics; see \cite{NSZ18,NSZ18b,NSZ19} and references therein. A strength of our approach is that it unites several types of dynamics under one umbrella approach. Finally, many works such as \cite{M24} also consider generalizations of the Trotter formula and applications to non-autonomous systems \cite{NSZ20}. An overview of the current state of the results can also be found in \cite{ZNI24}. Although our method also allows for such generalizations, we decided to focus on the most basic setting and will consider such applications in future work. State-dependent error bounds for quantum systems have been discussed in \cite{AFL21,BGHL23,BFHJY23,BFHL24,LGHB24,vL24}.
In particular, \cite{BGHL23,BFHJY23} prove state-dependent convergence for the Trotter product for more restrictive sets of vectors than what we are able to treat in this article. 
The idea of applying energy constraints to obtain state-dependent convergence rates was used before in the context of Lie group representations and the so-called quantum speed limits \cite{BD20,BDLR21,LGHB24}. Extensions of our methods to time-dependent and higher-order splitting schemes as in \cite{AKT14,Th08,Th12} are possible and will be considered elsewhere. 

{Higher-order methods do not necessarily yield better convergence rates because the iterated commutators underlying these schemes can become ill-behaved. For instance, consider the commutator
\[
[-\Delta, V] = \Delta V - 2 \nabla V \cdot \nabla,
\]
which, for the Coulomb potential \( V \), produces a delta distribution—a highly singular object. This illustrates that directly analyzing commutators of unbounded operators is generally not a promising strategy, and one should not expect improved convergence rates from higher-order schemes in such settings. This limitation has been confirmed numerically in Section VI of Burgarth et al.~\cite{BFHJY23}.

To adapt some of these ideas for higher-order methods, one can proceed as in~\cite{IK24}. For example, in the proof of Theorem 2.3, the error \( E_S(A, L; t) \) for Strang splitting with time step \( t \) is expressed as
\[
\begin{split}
E_S(A, L; t) = &\int_0^t e^{(t - \tau)(A + L)} [ e^{\tfrac{1}{2} \tau A}, L ] e^{\tau L} e^{\tfrac{1}{2} \tau A} \, d\tau \\
& - \frac{1}{2} \int_0^t e^{(t - \tau)(A + L)} e^{\tfrac{1}{2} \tau A} [ A, e^{\tau L} ] e^{\tfrac{1}{2} \tau A} \, d\tau.
\end{split}
\]
Error bounds for Strang splitting can then be derived from these commutators, using techniques similar to those in our bound~\eqref{eq:splitting}. However, a more refined analysis must carefully address the interaction between the errors propagated by both terms in this expression. In time-dependent settings, quantitative time-discretizations of the evolution operator are often employed.
}

\subsection{Outline of the article}

In the first part of the article, we study the case of generators $A$ and $L$ where $L$ is relatively $A$ bounded with $A$-bound $<1$.
    \begin{itemize}

        \item[-] In Section \ref{sec:regularL}, in Theorem \ref{thm:TrotterState} we state a vanilla convergence result with $\mathcal O(n^{-1})$ Trotter convergence rate under fairly general assumptions on $L$ and $A$.

    \item[-] In Section \ref{sec:FarvardSpaces}, we introduce Favard spaces and discuss some of their properties. These interpolation spaces capture the propagation of regularity of the data under the time evolution of contraction semigroups.  

    \item[-] In Section \ref{sec:TrotterFavard}, we embed with our Theorem \ref{thm:TrotterState2} the result of the previous section into the broader framework of Favard spaces and convergence rates. 

    \item[-] In Section \ref{sec:contraction_semigroups} we apply the framework of Section \ref{sec:FarvardSpaces} to the Trotter product of generators of contraction semigroups. We start with the so-called positive generators, in which case one can easily define fractional powers and obtain Corollary \ref{theo:positive_op} on the convergence of the Trotter product. 
    
    \item[-] In Section \ref{sec:SA}, we discuss applications for self-adjoint $A$.

    \item[-] In Section \ref{sec:unitary}, we turn to unitary dynamics, which covers the realm of Schr\"odinger dynamics in quantum mechanics. We provide an extensive discussion of the Trotter splitting for molecular Hamiltonians with Coulomb singularity, cf.~Theorem \ref{theo:many-body}. 
    \end{itemize}
    
In Section \ref{sec:energy_limited}, we use a recently developed framework of energy constraints to obtain general convergence rates in \cref{thm:EL_trotter} and \cref{cor:EL_trotter}, which we illustrate in the following subsections.

    \begin{itemize}
    \item[-] In \cref{subsec:Schroedinger}, we apply \cref{thm:EL_trotter} to show a $\mathcal O(n^{-1})$ convergence for all Schrödinger operators whose potentials have bounded second derivative.
    \item[-] In \cref{sec:scheiss_oszillator}, we apply the findings of \cref{subsec:Schroedinger} to the harmonic oscillator to obtain the convergence rate in \eqref{eq:intro_scheiss_oszillator}.
    \item[-] In \cref{subsec:Dirac}, we discuss applications to Dirac operators, such as magnetic Dirac operators.
    \end{itemize}

In Section \ref{sec:numerics}, we illustrate our findings with some numerical experiments. 

\smallsection{Notation}
We denote the identity operator by $I$, write $a \lesssim b$ and $a = \mathcal O(b)$ to indicate that there is a constant $C>0$ such that $a\le Cb$ and furthermore $a\sim b$ if $a\lesssim b$ and $b\lesssim a.$ We denote the resolvent set by $\rho(A)$ and its complement, the spectrum, by $\Spec(A).$ 
The space of bounded operators between Banach spaces $X$ and $Y$ is denoted $B(X,Y)$ and the corresponding operator norm of a linear operator $T:X \to Y$ is denoted by $\Vert T \Vert_{X \to Y} :=\sup_{\Vert x\Vert_X\le 1} \Vert Tx \Vert_Y.$
We use the so-called Japanese bracket notation $\langle\xi\rangle :=(1+|\xi|^2)^{1/2}$ for vectors $\xi\in\RR^n$.

\smallsection{Acknowledgements} We thank Tim M\"obus and Alexander Hahn for many interesting discussions on the topic. Moreover, we want to thank Daniel Burgarth, Nilanjana Datta, and Robin Hillier for organizing a workshop hosted by the ICMS in Edinburgh in May 2023, during which this work was initiated. We would also like to thank the referees for helpful suggestions on the structure of the manuscript. 

\smallsection{Declarations} The authors have no conflict of interest to declare.
In this work, no data were generated or processed.

\smallsection{Funding}
SB acknowledges support from the SNF Grant PZ00P2 216019.
NG has been supported by the Spanish MCIN (project PID2022-141283NB-I00) with the support of FEDER funds, by the Spanish MCIN with funding from European Union NextGenerationEU (grant PRTR-C17.I1) and the Generalitat de Catalunya, as well as the Ministry of Economic Affairs and Digital Transformation of the Spanish Government through the QUANTUM ENIA “Quantum Spain” project with funds from the European Union through the Recovery, Transformation and Resilience Plan - NextGenerationEU within the framework of the ”Digital Spain 2026 Agenda”.
RS acknowledges funding from the European Research Council (ERC Grant AlgoQIP, Agreement No. 851716).
LvL has been funded by the MWK Lower Saxony through the Stay Inspired Program (Grant ID: 15-76251-2-Stay-9/22-16583/2022).

\section{Trotter convergence under relative boundedness}
\label{sec:RelBoundTrotter}
In this section, we obtain convergence rates for the Trotter product formula \eqref{eq:Trotter}
under the assumption that one of the generators, say $L$, is relatively bounded by the other: We say $L$ is \emph{relatively $A$-bounded} if $D(A)\subseteq D(L)$ and there exist $a,b\ge 0$ such that for all $x\in D(A)$ we have
\begin{align}
\label{eq:LRelABoundedFirst}
\|Lx\| \le a\|A x\| + b\|x\|.
\end{align}
We call the infimum over all $a$ for which there is a $b\ge 0$ such that the above inequality holds the \emph{$A$-bound} of $L.$ 

For the remainder of this section, we consider situations in which $L$ is relatively $A$-bounded with $A$-bound $<1.$ The first consequence of this assumption is that for $L$ and $A$ being generators of $C_0$ contraction semigroups\footnote{We call a $C_0$ semigroup $(T_t)_{t\ge 0}$ a contraction semigroup if all maps $T_t$ are contractions, i.e., $\Vert T_t \Vert \le 1.$}, also $A+L$ with domain $D(A+L) = D(A)$ generates a $C_0$ contraction semigroup \cite[Thm.~2.7]{EN00}. The joint dynamics $e^{t(A+L)}$ therefore leaves $D(A)$ invariant. As $A$ is closed \cite[Thm.~1.4]{EN00}, $D(A)$ becomes a Banach space when equipped with the graph norm $\|x\|_{D(A)} = \|x\| +\|Ax\|$ and $e^{t(A+L)}:D(A)\to D(A)$ is a bounded operator by the closed graph theorem. The latter can be seen explicitly by noting that since the $A$-bound of $L$ is strictly smaller than 1, the graph norms of $A$ and $A+L$ are equivalent: In fact, for $x\in D(A)$, we have \[ \begin{split} \left\lVert A  x\right\| \le  \left\lVert (A+L) x\right\rVert + \left\lVert L  x\right\rVert 
 \le  \left\lVert (A+L) x\right\rVert +a\|Ax\| + b\|x\|
\end{split} \]
and hence 
\begin{align*}
    \|Ax\| \le \frac{1}{1-a}\Big(\|(A+L)x\| + b\|x\| \Big)
\end{align*}
and furthermore 
\[ \begin{split} \left\lVert (A +L) x\right\| \le  \left\lVert Ax\right\rVert + \left\lVert L  x\right\rVert 
 \le  (1+a)\left\lVert A x\right\rVert + b\|x\|
\end{split} \]
which gives $\|x\|_{D(A)}\sim \|x\|_{D(A+L)}.$\footnote{To conclude the boundedness of $e^{t(A+L)}:D(A)\to D(A)$ use $\|e^{t(A+L)}x\|_{D(A)} \lesssim \|e^{t(A+L)}x\|_{D(A+L)}$ together with $\|(A+L)e^{t(A+L)}x\| =\|e^{t(A+L)}(A+L)x\| \le \|(A+L)x\| \lesssim \|x\|_{D(A)}$ where we used that $e^{t(A+L)}$ is a contraction.}

In this regime, obtaining convergence rates in the Trotter formula reduces to studying the commutator $[L,e^{sA/n}],$ which from now on we shall refer to as the \emph{key commutator.} This is the content of the following lemma:
\begin{lemm}[Key commutator bound]
 \label{lem:SemiTrottCommBound}
		Let $X$ be a Banach space and $L$ and $A$ be generators of strongly continuous contraction semigroups on $X$ such that $L$ is relatively $A$ bounded, with $A$ bound $<1$. 
		Then for all $n\in\mathbb N$, $t\ge 0$ and $x\in D(A)$ we have
			\[ 
			\left\|\left(\left(e^{tL/n}e^{tA/n}\right)^n - e^{t(A+L)}\right)x\right\| \le t\sup_{s,\tau\in[0,t]}\left\| [L,e^{sA/n} ]e^{\tau(A+L)} x\right\| .\]
	\end{lemm}
\delete{ \begin{rem}
Invariance of the set $\mathcal D$ under $e^{t A}$ and $e^{t (A+L)}$ is not needed for the proof, as it suffices to assume that $e^{t A}e^{s (A+L)}\mathcal D\subseteq D(A)$ for all $t,s\ge 0.$. However, actual invariance can easily be verified in all the applications considered below.
One natural choice in the above is $\mathcal D =D(A).$. In this case, we immediately have $\mathcal D\subseteq D(L)$ and $e^{t A}$ and $e^{t (A+L)}$ leave $\mathcal D$ invariant for all $t\ge 0$ according to the relative boundedness assumption.
The more general version of the result presented in Lemma~\ref{lem:SemiTrottCommBound} is used in the proofs of Theorem~\ref{thm:TrotterState} and Theorem~\ref{thm:TrotterState2} below. Here, we consider $\mathcal D$ to be equal to $D(A^2)$ and the Favard space $F_\gamma$ for some $\gamma\in(0,2)$, respectively. 
 \end{rem}}
 \begin{proof}[Proof of Lemma~\ref{lem:SemiTrottCommBound}]
		By the semigroup property, it is easy to check that we can write the operator difference of interest as a telescoping sum, i.e.
		\begin{align*}
			\left(e^{tL/n}e^{tA/n}\right)^n - e^{t(A+L)} = \sum_{j=0}^{n-1}\left(e^{tL/n}e^{tA/n}\right)^j\left(e^{tL/n}e^{tA/n} - e^{t(A+L)/n}\right)e^{t(n-j-1)(A+L)/n}.
		\end{align*}
		Therefore, we see for $x\in D(A)$ using that $e^{tA/n}$ and $e^{tL/n}$ are contractions
		\begin{align*}
			\left\|\left(\left(e^{tL/n}e^{tA/n}\right)^n - e^{t(A+L)}\right)x\right\|  \le \sum_{j=0}^{n-1}  \left\|\left(e^{tL/n}e^{tA/n} - e^{t(A+L)/n}\right)e^{t(n-j-1)(A+L)/n}x\right\|.
		\end{align*}
		Now, using that by the relative boundedness assumption $D(A+L)=D(A),$ we can write for all $y\in D(A)\subseteq D(L)$\,\footnote{Note that the integral in \eqref{eq:AbleitenDiggi} is well-defined as a Bochner integral as the integrand is continuous, which is true by $L$ being relatively $A$-bounded and the graph norms of $A$ and $A+L$ being equivalent.}
		\begin{align}
  \label{eq:AbleitenDiggi}
			\nonumber \left(e^{tL/n}e^{tA/n} - e^{t(A+L)/n}\right)y &= \int_0^1\frac{d}{ds}\left(e^{stL/n}e^{stA/n}e^{(1-s)t(A+L)/n}\right)y\,ds \\&= \nonumber
			\frac{t}{n}\int_0^1\left(e^{stL/n}Le^{stA/n}e^{(1-s)t(A+L)/n}-e^{stL/n} e^{stA/n}L e^{(1-s)t(A+L)/n}\right)y\,ds \\&= \frac{t}{n}\int_0^1 e^{stL/n} [L, e^{stA/n} ]e^{(1-s)t(A+L)/n}y\, ds.
		\end{align}
        Plugging this into the telescopic sum, we find
        \begin{align*}
            \left\|\left(\left(e^{tL/n}e^{tA/n}\right)^n - e^{t(A+L)}\right)x\right\|  &\le \frac{t}{n}\sum_{j=0}^{n-1}  \left\|\int_0^1 e^{stL/n} [L, e^{stA/n} ]e^{(n-j-s)t(A+L)/n}x\, ds\right\| \\ &\le t\sup_{s,\tau\in[0,t]}\left\| [L, e^{sA/n} ]e^{\tau(A+L)}x\right\|.
        \end{align*}
	\end{proof}

\subsection{Elementary $\mathcal O(n^{-1})$ scaling for regular $L$}
\label{sec:regularL}

In this section, we provide a first simple derivation of convergence rates for the Trotter product formula using the key commutator bound established in Lemma~\ref{lem:SemiTrottCommBound}. We aim for a pointwise estimate with the Trotter product being applied to a fixed Banach space element of sufficient regularity, say $x\in D(A^2).$ Then, for generators $L$, which do not decrease this regularity by too much (with the precise meaning given in the following), we derive a {$\mathcal O(n^{-1})$} convergence rate of the Trotter product.

Denoting the semigroup generated by $A$ by $T_t= e^{tA}$ and furthermore $x_\tau := e^{\tau(A+L)}x,$ we observe that we can decompose the key commutator in Lemma~\ref{lem:SemiTrottCommBound} as
\begin{equation}
\label{eq:splitting}
 [L, T_{s/n}]x_\tau=  L (T_{s/n}-I)x_\tau  - (T_{s/n}-I)Lx_\tau.
 \end{equation}
To estimate both terms on the right-hand side, we use here that for $y\in D(A)$, we have
\begin{align}
\label{eq:NaiveDerivativeEst}
         \left\|(T_{s/n}- I) y\right\| = \frac{s}{n}\left\|\int_0^{1} T_{s/n}A y \,dr\right\| \le \frac{s}{n}\left\|A y\right\|.
     \end{align}
     Hence, the first term in \eqref{eq:splitting} can  easily be estimated under the additional assumption that the joined dynamics $e^{\tau(A+L)}$ leaves $D(A^2)$ invariant and therefore 
\begin{align}
\label{eq:Assumption2Intro}
    x_\tau \in D(A^2).
\end{align}
In this case, using that $L$ is relatively $A$-bounded as \eqref{eq:LRelABoundedFirst}, we see 
\begin{align*}
\Vert L(T_{s/n}-I)x_\tau \Vert \le a \|A(T_{s/n}-I)x_\tau\| + b\|(T_{s/n}-I)x_\tau\| \le \frac{s}{n}\left(a \|A^2x_\tau\| + b\|Ax_\tau\|\right).
\end{align*}
Furthermore, the second term on the right-hand side of \eqref{eq:splitting} can also be estimated, assuming that 
\begin{align}
\label{eq:Assumption1Intro}
    Lx_\tau \in D(A),
\end{align}
in which case we have
\begin{equation*}
\Vert (T_{s/n}-I)Lx_\tau \Vert \le  \frac{s}{n}\|ALx_\tau\|.
\end{equation*}

Both assumptions \eqref{eq:Assumption2Intro} and \eqref{eq:Assumption1Intro} are fulfilled if $L$ satisfies $LD(A^2) \subseteq D(A)$ and furthermore
\begin{align}
\label{eq:RelBoundGraph}
\Vert A L y \Vert \le a' \Vert A^2 y \Vert + b' \Vert  y \Vert
\end{align}
for some $0\le a' <1$ and $b'\ge 0$ and all $y\in D(A^2).$ In fact \eqref{eq:RelBoundGraph} together with \cite [Lem.~IV.3]{P18} gives that $D((A+L)^2) = D(A^2)$ and therefore that the restriction $e^{\tau(A+L)}\vert_{D(A^2)}:D(A^2) \to D(A^2)$ defines a well-defined, bounded operator from which we can conclude that $x_\tau\in D(A^2)$ and $Lx_\tau\in D(A)$ using $x\in D(A^2).$ Combining all of the above, we arrive at the following result, which we can think of as a straightforward operator theoretic generalization of the $\mathcal O(n^{-1})$ convergence rate for matrices.
 \begin{theo}[Perturbative $\mathcal O(n^{-1})$-Trotter] 
\label{thm:TrotterState}Let $X$ be a Banach space and $(L,D(L))$ and $(A,D(A))$ with $D(A)\subseteq D(L)$ generators of contraction semigroups. Assume there exist $0\le a<1$ and $b\ge0$ such that 
\begin{equation}
\label{eq:LRelK}
\Vert L x \Vert \le a\Vert A x\Vert + b\Vert x \Vert
\end{equation} for all $x\in D(A).$ Moreover, assume that $L\,D(A^2)\subseteq D(A)$ and
that there exist $0\le a'<1$ and $b'\ge 0$ such that  
\begin{align}
\label{eq:LrelKGRaphNorm}
\Vert A L x \Vert \le a' \Vert A^2 x \Vert + b' \Vert  x \Vert
\end{align}
for all $x\in D(A^2).$
Then for all $n\in\mathbb N$, $t\ge 0$ and $x\in D(A^2)$ we have
\begin{align}
\label{eq:IntroTheoMain}
			\left\|\left(\left(e^{tL/n}e^{tA/n}\right)^n - e^{t(A+L)}\right)x\right\| \le \frac {t^2}{n} \left(c_2\Vert  A^2x \Vert + c_1\Vert A x \Vert + c_0\|x\|
        \right)
		\end{align}
  for some $c_i\ge0$ depending only on $a,b,a'$ and $b'$.
  If we can choose $b=0$, then $c_1=0$.
  If additionally $b' = 0$ then also $c_0 = 0$.
\end{theo}  

A formal proof of this result in which the constants $c_0,c_1,c_2$ are expressed in terms of $a,b,a',b'$ can be found in Appendix~\ref{app:ProofIntroTheo}.

In many circumstances, the requirements of Theorem~\ref{thm:TrotterState} are too strong for many practically relevant generators $A$ and $L.$ In particular, as we see below, the assumption $L D(A^2) \subseteq D(A)$ is not satisfied in the case of Schr\"odinger operators with singular potentials, e.g., $X=L^2(\R^3),$ $A=-i\Delta$ and multiplication operator $L=iV$ corresponding to the potential $V(x)=\pm |x|^{-1}.$ To also be able to treat these cases, we extend the argument of Theorem~\ref{thm:TrotterState} using milder regularity assumptions.
For that, we use the so-called \emph{Favard spaces} of $A,$ which are interpolation spaces between the underlying Banach space $X$ and the domains $D(A^k)$ for $k\in\N$. It turns out that these provide a more nuanced regularity analysis compared to the setting of Theorem~\ref{thm:TrotterState} that leads to convergence rates of the Trotter product given by fractional powers of $n^{-1}.$

In the next section, we review the theory of Sobolev towers and Favard spaces of generators of $C_0$ contraction semigroups. After that, in Section~\ref{sec:TrotterFavard}, we provide an extension of Theorem~\ref{thm:TrotterState} involving Favard spaces (see Theorem~\ref{thm:TrotterState2}).

\subsection{Sobolev towers and Favard spaces}
\label{sec:FarvardSpaces}

For $A$ being the generator of a strongly continuous contraction semigroup on some Banach space $X$ the corresponding Sobolev tower is defined as follows: For $k\in\N_0$ consider the space $D(A^{k})$, where $D(A^{0})= X,$ equipped with the graph norms $\|x\|_{D(A^k)} = \|x\| + \|A^k x\|.$ The graph norm is equivalent to the sum of all $\|A^lx\|$ for $l\le k$ \cite [Lem.~II.10]{P18}, i.e.
\begin{align}
\label{eq:EquivalenGraphNorm}
 \|x\|_{D(A^k)} \lesssim \sum_{l=0}^k \|A^lx\| \lesssim \|x\|_{D(A^k)}.
\end{align}

Since $A^k$ is a closed operator \cite [Lemma II.8, Lemma II.9]{P18}, $D(A^k)$ is a Banach space. Furthermore, for $(T_t)_{t\ge 0}$ being the strongly continuous contraction semigroup generated by $A,$ we have that the restrictions $T_t|_{D(A^k)}$ define strongly continuous semigroups with generator given by the restriction $A|_{D(A^{k+1})}.$  

Inspired by the fact that $A$ satisfies $\sup_{\lambda>0}\|\lambda(\lambda-A)^{-1}\|<\infty$ \cite[Thm.~3.5]{EN00}, we can define complex interpolation spaces between the different levels of the Sobolev tower: For that, let $\alpha\in(0,\infty)$ and furthermore $r_\alpha\in(0,1]$ and $k_\alpha\in\N_0$ be the unique numbers such that $\alpha = r_\alpha + k_\alpha.$ One then obtains complex interpolation spaces $(D(A^{k_\alpha}),D(A^{k_\alpha+1}))_{r_\alpha,\infty},$ commonly referred to as \emph{Favard spaces}, of the form \cite[Prop.~3.1]{L09}
\begin{align*}
  F_\alpha\equiv F_{\alpha}(A)\equiv F_\alpha(A;X)&:=\Big\{ x \in D(A^{k_\alpha});\, \vert x \vert_{F_{\alpha}}:= \sup_{\lambda>0} \Vert \lambda^{r_\alpha} A(\lambda-A)^{-1} x\Vert_{D(A^{k_\alpha})}<\infty \Big\}. 
\end{align*}
Here, for later convenience, we introduced the notation $F_\alpha(A;X)$ explicitly denoting the Banach space $X$ from which the Favard spaces are constructed.
These spaces naturally become Banach spaces when equipped with an interpolation norm that is equivalent to $\Vert x\Vert_{F_{\alpha}}:=\Vert x \Vert_{D(A^{k_\alpha})} + \vert x \vert_{F_{\alpha}}$. 
In particular, for $\alpha\in(0,1]$, the Favard spaces $F_\alpha$ are interpolation spaces between the Banach space $X$ and the domain $D(A).$ Furthermore, for $\alpha\ge 1$ the Favard space $F_\alpha$ can equivalently be understood as a Favard space of order $r_\alpha\in(0,1]$ but with $D(A^{k_\alpha})$ as the underlying Banach space instead of $X,$ i.e.
\begin{align}
\label{eq:FavardDiscSemigroup}
    F_\alpha(A;X) = F_{r_\alpha}(A;D(A^{k_\alpha})).
\end{align}
\begin{ex}[H\"older spaces]
We consider the heat semigroup on $X = C_b(\RR^d)$, see \cite{BF19}. Its generator $A=\Delta$ has a domain $D(\Delta) = C^2_b(\RR)$ for $d=1$, while for $d \ge 2$ the domain is given by \[D(\Delta) = \{f \in C_b(\RR^d) \cap W^{2,p}(\RR^d), \text{ for all } p \in [1,\infty) \text{ and }\Delta f \in C_b(\RR^d)\}. \]
Then for $\alpha \in (0,1) \setminus \{1/2\}$ the Favard spaces are the H\"older spaces
\[ F_{\alpha}=C_b^{2\alpha}(\RR^d)\]
of bounded continuous functions that are $2\alpha$ H\"older continuous and for $\alpha=1/2$
\[ F_{1/2}= \left\{ f\in C_b(\RR^d); \sup_{x,y} \frac{\vert f(x)+f(y)-2 f((x+y)/2)\vert}{\vert x-y\vert} <\infty\right\}.\]
\end{ex}

\begin{ex}[Besov spaces]
To link Favard spaces to Besov spaces, we recall that the $B_{p,q}^{\alpha}$ Besov spaces also have a well-known interpolation characterization \cite[Theo. 14.4.31]{HNVW23}
\[
(W^{s_0, p}(\mathbb{R}^d), W^{s_1, p}(\mathbb{R}^d))_{\theta, q} = B^s_{p,q}(\mathbb{R}^d)
\]
where $W^{s_0, p}(\mathbb{R}^d)$ are the Sobolev spaces for $s_0 \neq s_1 \in \mathbb N_0$, $s= (1-\theta)s_0+\theta s_1,$ and $p,q \in [1,\infty]$ 
for \( 0 < \theta < 1 \)

Thus, for the special case of $q=\infty,$ the Besov spaces can be described as Favard spaces for a generator $A$ on $X:=W^{s_0, p}(\mathbb{R}^d)$ with domain $D(A)=W^{s_1, p}(\mathbb{R}^d)$ as 
\[ F_{\alpha}=(X,D(A))_{\theta,\infty}.\] 
This includes for $s_0=0$ and $p<\infty$ standard examples such as the translation, the heat, and Poisson semigroup. For the heat semigroup on $L^p(\mathbb R^n)$, with $p \in [1,\infty)$ with domain $D(\Delta)=W^{2,p}(\mathbb R^d)$ we have 
\[F_{\alpha} = B^{2\alpha}_{p,\infty}(\mathbb R^d). \]
\end{ex}

\begin{rem}
One can define Sobolev towers and Favard spaces more generally for operators $A$ only assuming $(0,\infty)\subseteq \rho(A)$ and $\sup_{\lambda>0}\|\lambda(\lambda-A)^{-1}\|\le C$ for some $C>0.$ Under this assumption, all statements in this section remain true apart from \eqref{eq:EquivalencDiff} below and the ones explicitly involving the semigroup $T_t.$ Note that in the special case $C=1$ and $A$ being densely defined, $A$ is immediately closed since $\rho(A)\neq \emptyset$ and hence this assumption is equivalent to $A$ being the generator of a strongly continuous contraction semigroup \cite[Thm.~3.5]{EN00}.
\end{rem}

Note that for all $k\in\N_0$ and $k <\alpha\le\beta\le k+1$ we have $D(A^{k+1})\subseteq F_{k+1}\subseteq F_\beta\subseteq F_\alpha\subseteq D(A^{k}).$ Furthermore, for all $0<\alpha\le \beta$ and $x\in F_\beta$ we have the upper bound\footnote{To see this, we use the fact that $\|x\|_{F_\alpha}$ is equivalent to the norm in \eqref{eq:FarvardNormEquiv}. Consider then first the case $r_\beta\ge r_\alpha.$ In that case, we have $\sup_{t>0}\|t^{-r_\alpha}(T_t-I)x\|_{D(A^{k_\alpha})} \le \sup_{0<t\le 1}\|t^{-r_\beta}(T_t-I)x\|_{D(A^{k_\alpha})}  + 2\|x\|_{D(A^{k_\alpha})} \lesssim \|x\|_{F_\beta}$ where we used \eqref{eq:EquivalenGraphNorm} to bound $\|x\|_{D(A^{k_\alpha})}\lesssim\|x\|_{D(A^{k_\beta})}.$ If on the other hand $r_\beta<r_\alpha$ then $k_\beta\ge k_\alpha+1$ and the bound follows by $\sup_{0<t\le 1}\|t^{-r_\alpha}(T_t-I)x\|_{D(A^{k_\alpha})} \lesssim \sup_{0<t\le 1}\|t^{1-r_\alpha}x\|_{D(A^{k_\alpha+1})} \le \|x\|_{D(A^{k_\beta})}.$}
\begin{align}
\label{eq:FavardDom}
    \|x\|_{F_\alpha} \lesssim \|x\|_{F_\beta}.
\end{align}

The importance of these spaces lies in their relationship to the asymptotic behavior of associated semigroups.  Indeed, we have the equivalent characterization \cite[Prop.~5.7]{L09}
\begin{equation}
\label{eq:FarvardSemigroup}
   F_{\alpha}=\Big\{x \in D(A^{k_\alpha});\, \sup_{t>0} \Vert t^{-r_\alpha}(T_tx-x)\Vert_{D(A^{k_\alpha})} <\infty\Big\}.
\end{equation} 
In particular, for $x\in F_\alpha$ we have
\begin{equation}
\label{eq:FavardBoundGraphNorm}
 \|T_tx-x\|_{D(A^{k_{\alpha}})} = {\mathcal{O}_{}(\|x\|_{F_\alpha}}\, t^{r_\alpha}).
 \end{equation}
 In this context, it is useful to note that the Favard norm $\|x\|_{F_\alpha}$ is equivalent to 
\begin{align}
\label{eq:FarvardNormEquiv}
\|x\|_{D(A^{k_\alpha})} + \sup_{t>0}\|t^{-r_\alpha}(T_t-I)x\|_{D(A^{k_{\alpha}})}.
\end{align}
Furthermore, in the special case of analytic semigroups, the Favard spaces are equivalently described by the elements $x \in X$ for which the $\Vert x\Vert_{F_{\alpha}}$ equivalent norm \[\Vert x \Vert_{D(A^{k_\alpha})} + \sup_{t \in (0,1)} t^{1-r_\alpha} \Vert A T_t x \Vert_{D(A^{k_\alpha})}\] is finite.

The following lemma generalizes \eqref{eq:FavardBoundGraphNorm} from the norms on $X$ and $D(A^{k})$ to the Favard norms.
\begin{lemm}
\label{lem:FarvardConvergRates}
Let $A$ be the generator of a strongly continuous contraction semigroup $\left(T_t\right)_{t\ge 0}$ and $\alpha,\gamma\in(0,\infty)$ such that $ \alpha\le\gamma.$ Then, for $\gamma' = \min\{\gamma-\alpha,1\}$
\[    \|T_t -I\|_{F_\gamma \to F_{\alpha}} = \mathcal{O}\left(t^{\gamma'}\right).\]
\end{lemm}
\begin{proof}
For the proof, we use that 
\begin{align}
\label{eq:EquivalencDiff}
    \|T_t -I\|_{F_\gamma \to F_{\alpha}} \lesssim \|T_t-I\|_{F_\gamma \to D(A^{k_\alpha})} + \sup_{s>0}\|s^{-r_\alpha}(T_s-I)(T_t-I)\|_{F_\gamma \to D(A^{k_{\alpha}})}.
\end{align}
First consider the case $\gamma - \alpha < 1.$ In that case $k_\alpha = k_\gamma$ and hence $\gamma -\alpha = r_\gamma -r_\alpha\ge 0.$ To bound the second term in \eqref{eq:EquivalencDiff} we use 
\begin{align*}
    &\sup_{s>0}\|s^{-r_\alpha}(T_s-I)(T_t -I)\|_{F_\gamma\to D(A^{k_\alpha})} \\ & \lesssim \sup_{0<s\le t} s^{r_\gamma -r_\alpha}\|T_t-I\|_{D(A^{k_\alpha})} +  t^{r_\gamma}\sup_{s > t}s^{-r_\alpha}\|T_s-I\|_{D(A^{k_\alpha})} \le 4 t^{r_\gamma -r_\alpha} ,
\end{align*}
where for the last inequality we have used that $\left(T_t\right)_{t\ge 0}$ restricted to $D(A^{k_\alpha})$ defines a contraction semigroup. As for the first term in \eqref{eq:EquivalencDiff} we have $\|T_t-I\|_{F_\gamma \to D(A^{k_\alpha})} \lesssim t^{r_\gamma} =\mathcal{O}\left(t^{\gamma-\alpha}\right)$ for $t\to 0$ immediately from definition, we have shown $\|T_t -I\|_{F_\gamma \to F_{\alpha}} = \mathcal{O}\left(t^{\gamma-\alpha}\right).$

Next we consider the case $\gamma -\alpha\ge 1.$
For that note that for all $x\in D(A^{k_\alpha+1})\subseteq F_\gamma$ we have 
\begin{align}
\label{eq:FullA}
    \|(T_t-I)x\|_{D(A^{k_\alpha})} = \left\|\int^t_0T_sAx\,ds\right\|_{D(A^{k_\alpha})} \le t\|Ax\|_{D(A^{k_\alpha})} \lesssim t\|x\|_{D(A^{k_\alpha+1})},
\end{align} 
where for the second inequality, we have used \eqref{eq:EquivalenGraphNorm}.
Using this and focusing on the second part in \eqref{eq:EquivalencDiff} we see 
\begin{multline*}
     \sup_{s>0}\|s^{-r_\alpha}(T_s-I)(T_t -I)\|_{F_\gamma\to D(A^{k_\alpha})}\nonumber\\
     \lesssim t\sup_{s>0} \|s^{-r_\alpha}(T_s-I)\|_{F_\gamma\to D(A^{k_\alpha+1})} \nonumber\lesssim t \sup_{s>0}\|s^{-r_\alpha}(T_s-I)\|_{F_{\alpha+1}\to D(A^{k_\alpha+1})} \lesssim t,
\end{multline*}
where we used $F_\gamma\subseteq F_{\alpha+1}$ together with \eqref{eq:FavardDom}. For the first term in \eqref{eq:EquivalencDiff} we can use \eqref{eq:FullA} directly, which gives $\|T_t-I\|_{F_\gamma\to D(A^{k_\alpha})}\lesssim t.$ Combining both yields $\|T_t-I\|_{F_\gamma\to F_\alpha} \lesssim t$ and finishes the proof.

\end{proof}

Finally, we recall a basic but useful property of Favard spaces that shows that they are stable under relatively bounded perturbations.   

\begin{lemm}[Stability of Favard spaces]
\label{lemm:stability}
Let $A_0$ and $A$ be generators of strongly continuous contraction semigroups.
Let $B=A-A_0$ be relatively $A_0$ bounded with $A_0$-bound given by $a \ge 0.$ Then $F_{\alpha}(A_0) \subset F_{\alpha}(A)$ for all $\alpha \in (0,1]$ and in particular if $a<1$, then $F_{\alpha}(A_0)=F_{\alpha}(A).$
\end{lemm}
\begin{proof}
It suffices to show the inclusion $F_{\alpha}(A_0) \subset F_{\alpha}(A)$, since for $a<1$, we also have for some $b>0$
\[ \Vert A_0x  \Vert \le \Vert A x \Vert + \Vert Bx \Vert \le \Vert A x\Vert +a \Vert A_0x \Vert + b \Vert x \Vert  \]
and thus 
\[ \Vert Bx  \Vert \le a \Vert A_0 x\Vert + b \Vert x \Vert \le \frac{a}{1-a} \Vert A x \Vert + \frac{b}{1-a} \Vert x \Vert.  \]
This implies $F_{\alpha}(A)\subset F_{\alpha}(A_0)$ by the first part.

Since $A$ and $A_0$ generate strongly continuous contraction semigroups, we have $(0,\infty)\subset \rho(A_0),\rho(A)$ and $
    \sup_{\lambda>0}\|\lambda(\lambda-A_0)^{-1}\|, \sup_{\lambda>0}\|\lambda(\lambda-A)^{-1}\|\le 1.
$
We observe that we have the simple estimate 
\[ \Vert A(\lambda-A)^{-1}x \Vert \le \Vert x \Vert +\Vert \lambda(\lambda-A)^{-1}x \Vert \le 2\Vert x \Vert.  \]
This is enough to conclude that for fixed $\delta>0$
\begin{equation}
\label{eq:bound1}
 \sup_{\lambda \in (0,\delta)} \Vert \lambda^{\alpha} A(\lambda-A)^{-1}x \Vert \le 2\delta^{\alpha}\Vert x \Vert.  
 \end{equation}

On the other hand, using the resolvent identity
\begin{equation}
\label{eq:resolvent_identity}
\Vert A (\lambda-A)^{-1}x \Vert \le \Vert A (\lambda-A_0)^{-1}x \Vert +\Vert A (\lambda-A)^{-1} B(\lambda-A_0)^{-1} x \Vert.   
\end{equation}

For the first term on the right-hand side of \eqref{eq:resolvent_identity}, relative boundedness yields 
\[\Vert A (\lambda-A_0)^{-1}x \Vert \le (1+a)\Vert A_0 (\lambda-A_0)^{-1}x \Vert + b \Vert  (\lambda-A_0)^{-1}x \Vert.  \]

For the second term on the right-hand side of \eqref{eq:resolvent_identity}, we find 
\[\begin{split}\Vert A (\lambda-A)^{-1} B(\lambda-A_0)^{-1}x \Vert 
&\le \Vert A (\lambda-A)^{-1} \Vert\Vert B(\lambda-A_0)^{-1}x\Vert\\  
&\le \Big(1+ \Vert  \lambda(\lambda-A)^{-1} \Vert\Big)\Vert B(\lambda-A_0)^{-1}x\Vert\\
&\le 2\Vert B(\lambda-A_0)^{-1}x\Vert\\
\end{split}\]

Using $A_0$-boundedness of $B$ we find for the last expression in the previous line
\[\Vert B(\lambda-A_0)^{-1}x\Vert \le a \Vert A_0(\lambda-A_0)^{-1}x\Vert + c \Vert (\lambda-A_0)^{-1}x\Vert.  \]
For $\lambda >\delta$, we have
\[\Vert \lambda^{\alpha}(\lambda-A_0)^{-1}x\Vert \le \frac{ \Vert x \Vert}{\delta^{1-\alpha}}.\]
Thus, we have shown \[ \sup_{\lambda >\delta}\Vert \lambda^{\alpha} A (\lambda-A)^{-1} x \Vert \lesssim \sup_{\lambda >\delta}\Vert \lambda^{\alpha} A_0 (\lambda-A_0)^{-1} x \Vert + \Vert x \Vert,\]
which together with $\eqref{eq:bound1}$ implies the claim.
\end{proof}

\subsection{Trotter convergence on Favard spaces}
\label{sec:TrotterFavard}

In this section, we use the theory of Favard spaces outlined in Section~\ref{sec:FarvardSpaces} to extend Theorem~\ref{thm:TrotterState} by interpolating initial data regularity $D(A^2)$ to $X.$. In particular, we consider $x\in F_\gamma$ for some $0<\gamma<2$ and $F_\gamma$ denoting the corresponding Favard space of $A.$

For that, we follow an argument similar to the one outlined in Section~\ref{sec:regularL}: The starting point of the derivation is again the key commutator bound obtained in Lemma~\ref{lem:SemiTrottCommBound} using the fact that we can decompose
\begin{equation}
\label{eq:splitting2}
 [ L, T_{s/n}]x_\tau =  L (T_{s/n}-I) x_\tau - (T_{s/n}-I)Lx_\tau,
 \end{equation}
 with $T_{s/n} = e^{sA/n}$ and $x_\tau = e^{\tau(A+L)}x.$
 Here, instead of relying on the naive estimate \eqref{eq:NaiveDerivativeEst} we employ the refined version provided in Lemma~\ref{lem:FarvardConvergRates}: 
 
 For the first term on the right-hand side of \eqref{eq:splitting2} we assume that $F_\alpha\subseteq D(L)$ for some $\alpha<\min\{1,\gamma\}$ or, in other words, by the closed graph theorem, the restriction $L|_{F_\alpha}: F_\alpha \to X$ is bounded. This assumption, in particular, implies that $L$ is infinitesimally $A$-bounded \cite[Lem.~2.13]{EN00} and therefore functions analogously to the relative boundedness assumption \eqref{eq:LRelK} in Theorem~\ref{thm:TrotterState}. Assuming furthermore that the joined dynamics $e^{\tau(A+L)}$ leaves the Favard space $F_\gamma$ invariant which implies that
\begin{align}
\label{eq:Assump2FavardTrott}
 x_\tau \in F_\gamma.   
\end{align}
This property, together with Lemma~\ref{lem:FarvardConvergRates} implies that
\begin{equation}
\left\Vert L(T_{s/n}-I)x_\tau \right\Vert \le \|L\|_{F_\alpha\to X}\left\Vert (T_{s/n}-I)x_\tau \right\Vert_{F_\alpha} =\mathcal O_{ \|L\|_{F_\alpha\to X}\Vert x_\tau \Vert_{F_\gamma}}\left(\left(\frac{s}{n}\right)^{\gamma'}\right).
\end{equation}
with $\gamma'=\min\{1,\gamma-\alpha\}.$

The  second term on the right-hand side of \eqref{eq:splitting2} can easily be estimated once we additionally assume \begin{align}
\label{eq:Assumption1FavardTrott}
    Lx_\tau \in F_\beta
\end{align} 
for some $\beta\le \gamma$
in which case we have
\begin{equation}
\label{eq:estm1}
\left\Vert (T_{s/n}-I)Lx_\tau \right\Vert = \mathcal O_{\Vert Lx_\tau \Vert_{F_{\beta}}}\left(\left(\frac{s}{n}\right)^{\beta}\right).
\end{equation}

As seen in the proof of Theorem~\ref{thm:TrotterState2} below, the invariance of $F_\gamma$ under $e^{\tau(A+L)}$ and hence \eqref{eq:Assump2FavardTrott} as well as \eqref{eq:Assumption1FavardTrott} both follow when assuming that $LF_\gamma\subseteq F_\beta$ with $k_\gamma<\beta\le \gamma,$ where $k_\gamma\in\{0,1\}$ is such that $\gamma-k_\gamma\in(0,1].$ In the case $k_\gamma=1,$ this assumption implies that $L$ is also infinitesimally $A$-bounded on the graph space $D(A)$ and can hence be seen as the analog of \eqref{eq:LrelKGRaphNorm} in Theorem~\ref{thm:TrotterState}.

We are now ready to state and formally prove our main theorem of this section in the following.

\begin{theo}
 \label{thm:TrotterState2}Let $X$ be a Banach space and $(L,D(L))$ and $(A,D(A))$ be generators of strongly continuous contraction semigroups. 
Moreover, assume that the Favard spaces of $A$ satisfy $F_\alpha\subseteq D(L)$ and $L F_\gamma\subseteq F_\beta$ for some positive numbers $\alpha < \min\{1,\gamma\}$ and $ k_\gamma < \beta\le \gamma < 2,$ where $k_\gamma\in\{0,1\}$ is such that $\gamma -k_\gamma \in(0,1].$ 
Then for all $n\in\mathbb N$ and $t\ge 0$ we have
\begin{align}
			\left\|\left(e^{tL/n}e^{tA/n}\right)^n - e^{t(A+L)}\right\|_{F_\gamma\to X} =\mathcal{O}\left(\frac{t^{1+\delta}}{n^\delta}\right),
		\end{align}
  where $\delta= \min\{1,\beta,\gamma-\alpha\}.$ 
\end{theo} 
\begin{rem}
\label{rem:Gamma=1CaseFavard}
In the special case $\gamma=1,$ using that $D(A)\subseteq F_1,$ the same argument as for Theorem~\ref{thm:TrotterState2} gives under the condition that $F_\alpha\subseteq D(L)$ and $LD(A)\subseteq F_\beta$ for some $0<\alpha,\beta<1$ that 
\begin{align*}
			\left\|\left(e^{tL/n}e^{tA/n}\right)^n - e^{t(A+L)}\right\|_{D(A)\to X} =\mathcal{O}\left(\frac{t^{1+\delta}}{n^\delta}\right),
		\end{align*}
  where $\delta= \min\{\beta,1-\alpha\}.$ 
\end{rem}
\begin{rem}
Theorem~\ref{thm:TrotterState} can be seen as an extreme point of the above Theorem~\ref{thm:TrotterState2} extending it to the case $\gamma =2$ and $\beta=\alpha =1$. In this case, the Favard spaces $F_2$ and $F_1$ are replaced by $D(A^2)$ and $D(A)$ respectively. Furthermore, the relative $A$-boundedness of $L$ with a bound strictly smaller than 1 with respect to the norm on $X$ and the graph norm on $D(A)$, that is, \eqref{eq:LRelK} and \eqref{eq:LrelKGRaphNorm}, needs to be assumed additionally.
\end{rem}
\begin{proof}[Proof of Theorem~\ref{thm:TrotterState2}]
Since the map $L: F_\alpha \to X,$ is well defined and closed (as $L: D(L) \to X$ is closed) we see by the closed graph theorem that it is also bounded and furthermore, using $\alpha< 1,$ that $L$ is infinitesimally $A$-bounded with respect to the norm on $X$ \cite[Lem.~2.13]{EN00}, i.e., for all $\eps>0$ and $x\in D(A)$ we have
\begin{align}
\label{eq:LRelABounded}
    \|Lx\|\le \eps\|A x\| +b_\eps \|x\|
\end{align}
for some $b_\eps\ge 0.$

We want to repeat this argument to show that $L$ is also infinitesimally $A$-bounded with respect to the graph norm on $D(A^{k_\gamma}).$ In the case $k_\gamma =0$ we have already done this \eqref{eq:LRelABounded}, so we focus on $k_\gamma =1$: From $L F_\gamma \subseteq F_\beta$ and the closed graph theorem we get that $L: F_\gamma \to F_\beta$ is bounded. Using $k_\gamma<\beta$ and hence $F_\beta\subseteq D(A^{k_\gamma}),$ this in particular implies boundedness of $L :F_\gamma \to D(A^{k_\gamma})$. Furthermore, using $\gamma < k_\gamma+1$ and \cite[Lem.~2.13]{EN00} on the Banach space $D(A^{k_\gamma}),$ we see that $L$ is infinitesimally $A$-bounded with respect to the graph norm, i.e., for all $\varepsilon>0$ and $x\in D(A^{k_\gamma+1})\subseteq F_\gamma\subseteq D(L)$ we have
\begin{align}
\label{eq:LRelABoundedonGraph}
    \|Lx\|_{D(A^{k_\gamma})}\le \eps\|A x\|_{D(A^{k_\gamma})} +c_\eps \|x\|_{D(A^{k_\gamma})}
\end{align}
for some $c_\eps\ge 0.$ 

Since $(L,D(L))$ generates a strongly continuous contraction semigroup and is hence dissipative, we can use \eqref{eq:LRelABounded} and \cite[Theorem 2.7]{EN00} which gives that $(A+L,D(A))$ generates a strongly continuous contraction semigroup as well. Furthermore, again from \eqref{eq:LRelABounded} we see that the graph norms of $A+L$ and $A$ are equivalent. This shows that the restriction $e^{t(A+L)}|_{D(A)}$ defines a strongly continuous semigroup on the graph space $D(A),$ which satisfies $\|e^{t(A+L)}\|_{D(A)} \lesssim \|e^{t(A+L)}\|_{D(A+L)}  \le 1$ for all $t\ge 0.$ Its generator is given by $A+L$ restricted on $D((A+L)^2).$\footnote{In the case of $k_\gamma =1$ we even have $D((A+L)^2) = D(A^2)$ which follows using \eqref{eq:LRelABounded} and \eqref{eq:LRelABoundedonGraph} and \cite[Lem.~IV.3]{P18}.} By \cite[Theorem 3.8]{EN00}, it has the set of positive numbers in its resolvent set and satisfies $\sup_{\lambda>0}\|\lambda(\lambda - A+L)^{-1}\|_{D(A)} \lesssim 1.$

We can hence employ the stability result of Favard spaces, Lemma~\ref{lemm:stability}\ \footnote{Note, that Lemma~\ref{lemm:stability}, is applicable in this context as $F_\gamma(A)$ and $F_\gamma(A+L)$ can be seen as Favard spaces of order $r_\gamma\in(0,1]$ but with underlying Banach space being $D(A^{k_\gamma})$ instead of $X.$} on $D(A^{k_\gamma})$ as underlying Banach space together with \eqref{eq:LRelABoundedonGraph}, which gives $F_\gamma(A+L) = F_\gamma(A).$
Furthermore, we have that $e^{t(A+L)}$ is bounded uniformly over $t\ge0$ with respect to the Favard norm of $F_\gamma \equiv F_\gamma(A)$ using the closed graph theorem for the operator $e^{t(A+L)}:F_\gamma\to F_\gamma.$

The key commutator bound in Lemma~\ref{lem:SemiTrottCommBound} gives for all $x\in D(A)$
\begin{equation} 	
\label{eq:KeyCommProofTheo}
\left\|\left(\left(e^{tL/n}e^{tA/n}\right)^n - e^{t(A+L)}\right)x\right\| \le t\sup_{s,\tau\in[0,t]}\left\| [L,e^{sA/n} ]e^{\tau(A+L)} x\right\|.
\end{equation}
In fact, since $L:F_\alpha\to X$ is bounded and furthermore  
$F_\gamma\subseteq \overline{D(A)}^{\|\,\cdot\,\|_{F_\alpha}}$, which holds by  $\gamma<\alpha$ and \cite[Prop.~5.14., Theo.~5.15]{EN00}, we get \eqref{eq:KeyCommProofTheo} for all $x\in F_\gamma.$ Using this and the above we see
\begin{align}
\label{eq:FinallyComm}
    \nonumber&\left\|\left(e^{tL/n}e^{tA/n}\right)^n - e^{t(A+L)}\right\|_{F_\gamma\to X} \le t\sup_{s,\tau\in[0,t]}\left\| [L,e^{sA/n} ]e^{\tau(A+L)} \right\|_{F_\gamma\to X}\\& \lesssim t\sup_{s\in[0,t]}\left\| [L,e^{sA/n} ]\right\|_{F_\gamma\to X} \le t\sup_{s\in[0,t]}\left\| L\left(e^{sA/n} -I\right)\right\|_{F_\gamma\to X} +t\sup_{s\in[0,t]}\left\| \left(e^{sA/n} -I\right)L\right\|_{F_\gamma\to X}.
\end{align}
For the first term in \eqref{eq:FinallyComm} we can use Lemma~\ref{lem:FarvardConvergRates} together with the fact that $L:F_\alpha\to X$ is bounded which gives 
\begin{align*}
    \sup_{s\in[0,t]}\left\| L\left(e^{sA/n} -I\right)\right\|_{F_\gamma\to X} \lesssim \sup_{s\in[0,t]}\left\| e^{sA/n} -I\right\|_{F_\gamma\to F_\alpha}  = \mathcal{O}\left(\left(\frac{t}{n}\right)^{\gamma'}\right)
\end{align*}
where $\gamma'=\min\{1,\gamma-\alpha\}.$ To finish the proof we use the boundedness of $L:F_\gamma\to F_\beta$ to bound the second term in \eqref{eq:FinallyComm} as
\begin{align*}
    \sup_{s\in[0,t]}\left\| \left(e^{sA/n} -I\right)L\right\|_{F_\gamma\to X} \lesssim\sup_{s\in[0,t]}\left\|e^{sA/n} -I\right\|_{F_\beta\to X} =\mathcal{O}\left(\left(\frac{t}{n}\right)^{\beta'}\right),
\end{align*}
where $\beta' = \min\{1,\beta\}.$
\end{proof}

In the following subsections, we apply the result of Theorem~\ref{thm:TrotterState2} for special cases of generators for which we have good control over the corresponding Favard spaces. The first case, in Section \ref{sec:contraction_semigroups}, considers $-A$ to be a positive operator on some Banach space $X.$ 
In the second case, in Section~\ref{sec:SA}, we consider $X$ to be a Hilbert space and $A$ either self-adjoint with $-A$ being positive semidefinite or $A$ being anti self-adjoint (i.e., $iA$ self-adjoint) and hence $A$ being the generator of a unitary group. 
In all of these cases, we see that the Favard spaces of $A$ are closely linked to the domains of the fractional powers of $A$ (or, more precisely, of $|A|$). This leads us to formulate the abstract result of Theorem~\ref{thm:TrotterState2} for these special cases in a more concrete way in Corollaries~\ref{theo:positive_op}, \ref{corr:positive_op}, \ref{corr:positive_selfop}, \ref{theo:singulaire} respectively.

\subsection{Strictly dissipative dynamics}
\label{sec:contraction_semigroups}
Positive operators on general Banach spaces are defined as follows:
\begin{defi}[Positive operator] 
A closed operator $B$ is called positive on a Banach space $X$ if $(-\infty,0]\subseteq \rho(B)$ and $\Vert (\lambda-B)^{-1}\Vert =\mathcal O((1-\lambda)^{-1})$ for $\lambda \in (-\infty,0].$
\end{defi}
\begin{rem}
If $B$ is a self-adjoint operator on a Hilbert space, then positivity in the usual sense, that is, $\inf_{x \in D(B);\Vert x \Vert=1}\langle Bx,x\rangle > 0$ is equivalent to positivity as in the above definition. 
\end{rem}
The relevance of this definition for us is based on the following simple observation.
\begin{lemm}
Let $A$ be a generator of a strongly continuous contraction semigroup, then $-A$ is positive if $0 \in \rho(A).$
\end{lemm}
\begin{proof}
 Since $0 \in \rho(A),$ we have for some $\eps>0$ because of the openness of the resolvent set and the continuity of the resolvent in $\lambda$ that $[-\eps,\eps]\subseteq \rho(A)$ and  
\[ \Vert  (\lambda+A)^{-1}\Vert =\mathcal{O}(1) \text{ for } |\lambda|\le \eps. \] Furthermore, by \cite[Thm.~3.5]{EN00}, we have the bound $\Vert \lambda (\lambda-A)^{-1} \Vert\le 1$ for $\lambda>0$ for generators of semigroups of contraction. Combining both, we see
\[ \Vert (\lambda+A)^{-1}\Vert =\mathcal O((1-\lambda)^{-1})\text{ for }\lambda \in (-\infty,0].\]
\end{proof}

If we assume that $-A$ is positive, we can define $A^{-\alpha}$ for $\alpha>0$, and the operator is given by the explicit formula \cite[Corr. 5.28]{EN00}  
\[ A^{-\alpha} = \frac{1-e^{-2\pi i \alpha}}{2\pi i}  \int_0^{\infty} s^{-\alpha}(s-A)^{-1} \ ds \text{ for }\alpha \in (0,1).\]
The operator $A^{\alpha}$ is then defined as the inverse of $A^{-\alpha}$ with $D(A^{\alpha}) = \operatorname{ran}(A^{-\alpha}).$
For $\alpha<\beta$, we also have $D(A^{\beta}) \subseteq F_{\beta}\subseteq D(A^\alpha)$, where the first inclusion can be found in \cite[Prop.~4.7]{L09} and the proof of the second inclusion follows from \cite[Prop.~5.33]{EN00} \footnote{The proof assumes a negative growth bound of the semigroup since the book only defines Favard spaces in this context. However, a careful inspection of the proof shows that it applies to our setting, as well.}
\begin{equation}
\label{eq:bound2}
\Vert x \Vert + \Vert A^{\alpha}x \Vert \lesssim \Vert x \Vert_{F_{\beta}}\lesssim \Vert x \Vert+\Vert A^{\beta}x \Vert.
\end{equation}

Using this, we can immediately deduce the following result from Theorem~\ref{thm:TrotterState2}:
\begin{corr}
\label{theo:positive_op}
Let $X$ be a Banach space and $L$ and $A$ be generators of strongly continuous contraction semigroups on $X$ with $0 \in \rho(A)$, $D(A^{\alpha}) \subseteq D(L)$, and $LD(A)\subseteq D(A^{\beta})$ for some $\alpha,\beta \in (0,1).$
Then, for fixed $\varepsilon>0$, the Trotter product satisfies
\begin{align*} \left\|\left(\left(e^{tL/n}e^{tA/n}\right)^n - e^{t(A+L)}\right)\right\|_{D(A) \to X}=\mathcal O_{\varepsilon}\left(\frac{t^{1+\delta}}{n^\delta}\right),
\end{align*}
  where $\delta= \min\{\beta,1-\alpha-\eps\}.$
\end{corr} 
\begin{proof}
By \eqref{eq:bound2} we have that $LD(A)\subseteq D(A^\beta)\subseteq F_\beta.$ Furthermore, since $D(A^\alpha)\subseteq D(L)$ we have, again by \eqref{eq:bound2} that for any $\eps>0$ fixed also $F_{\alpha+\eps}\subseteq D(L).$ Therefore, the conditions of Theorem~\ref{thm:TrotterState2} or Remark~\ref{rem:Gamma=1CaseFavard} are fulfilled which finishes the proof. 
\end{proof}

Assuming $X$ to be a Hilbert space, see \cite[Corr.\ 4.30]{L09}, we have the equality $D(A^\beta)= F_\beta$. This way, we remove the dependence on $\varepsilon>0$ in Theorem \ref{theo:positive_op} and find
\begin{corr}
\label{corr:positive_op}
Let $X$ be a Hilbert space and $L$ and $A$ be generators of strongly continuous contraction semigroups on $X$ with $0 \in \rho(A)$, $D(A^{\alpha}) \subset D(L)$, and $LD(A)\subset D(A^{\beta})$ for some $\alpha,\beta \in (0,1).$
Then, the Trotter product satisfies
\begin{align*}   
\left\|\left(\left(e^{tL/n}e^{tA/n}\right)^n - e^{t(A+L)}\right)\right\|_{D(A) \to X}=\mathcal O\left(\frac{t^{1+\delta}}{n^\delta}\right)
\end{align*}  where $\delta= \min\{\beta,1-\alpha\}.$ 
\end{corr} 

Our framework in the above theorem allows, for example, for $A$ to be the Dirichlet-Laplacian on a sufficiently nice bounded domain $\Omega \subset \mathbb R^n$, since $0 \notin \Spec( \Delta_{H^2(\Omega) \cap H_0^1(\Omega)}).$ However, it does not allow us to study, for instance, basic operators such as $A=\Delta$ on $\mathbb R^n.$ We will take care of this in the next two subsections.

\subsection{Self-adjoint and anti self-adjoint generators on Hilbert spaces}
\label{sec:SA}
 In the case of self-adjoint or anti self-adjoint $A$, one can use the spectral theorem to easily characterize the Favard spaces. This is the content of the following lemma:
 \begin{lemm}
 \label{lemm:SA}
 Let $-A$ be self-adjoint and positive semi-definite on a Hilbert space $X$, then in the sense of the functional calculus for all $\alpha>0$
 \[ F_{\alpha}(A) = D(|A|^{\alpha})\]
 and the Favard norm is equivalent to the graph norm of $|A|^{\alpha}.$

 Let $K$ be self-adjoint, then for all $\alpha>0$  
  \[ F_{\alpha}(iK) = D(\vert K \vert^{\alpha})\]
  and the corresponding Favard norm is equivalent to the graph norm of $|K|^\alpha.$
\end{lemm}
\begin{rem}
\label{rem:Kgammabloed}
The specific form of the Favard spaces presented in Lemma~\ref{lemm:SA} enables us to provide stronger connections between different levels of the Sobolev tower:
Let $-A$ be self-adjoint and positive semidefinite. Then for all $\alpha>0$ the graph space $D(|A|^\alpha)$ naturally becomes a Hilbert space when equipped with the scalar product $\langle x,y\rangle_\alpha := \langle |A|^\alpha x,|A|^\alpha y\rangle + \langle x,y\rangle.$ It is easy to verify that the restriction $-A':=-A|_{D(|A|^{\alpha+1})}$ defines again a self-adjoint, positive semi-definite operator on $D(|A|^\alpha)$ with fractional powers $|A'|^\beta$ being defined on the domain $D(|A'|^{\beta}) =D(|A|^{\alpha+\beta}).$ Hence, using Lemma~\ref{lemm:SA}, we see that the Favard spaces of $A'$ satisfy 
\begin{align}
\label{eq:FavardSemigroupCont}
F_\beta(A';D(|A|^\alpha))= D(|A|^{\alpha+\beta}) =  F_{\alpha+\beta}(A; X),
\end{align}
which can be seen as a refined version of \eqref{eq:FavardDiscSemigroup}.

Let now $L$ be closed and such that $F_\alpha(A;X) \subseteq D(L)$ and $L F_\gamma(A;X)\subseteq F_\beta(A;X)$ for some positive numbers $\alpha<1$ and $\gamma-1< \beta\le\gamma<2.$  Using \eqref{eq:FavardSemigroupCont} instead of \eqref{eq:FavardDiscSemigroup}, we can argue similarly to the proof of Theorem~\ref{thm:TrotterState2} but replacing $D(A^{k_\gamma})$ with $D(|A|^{\beta})$ in the case $\gamma\ge 1$ that $e^{t(A+L)}: F_\gamma(A;X) \to F_\gamma(A;X)$ is a well-defined bounded operator. This enables us in Corollary~\ref{corr:positive_selfop} to lift the restriction $\beta> k_\gamma$ needed in Theorem~\ref{thm:TrotterState2} and allow for all $\beta\ge\gamma-1.$ The same argument can also be employed for anti self-adjoint generators, i.e., $A=iK$ for $K$ being self-adjoint, giving the same extended range on the parameter $\beta$ in Corollary~\ref{theo:singulaire}.
\end{rem}
\begin{proof}[Proof of Lemma~\ref{lemm:SA}]
Let $r\in(0,1)$ and $m\in \N.$
 Using the functional calculus, we find with $E_{|A|}$ being the spectral measure of $|A|=-A$
\[ \Vert \lambda^{r} A^{m}(\lambda-A)^{-1}x \Vert^2 = \int_{0}^{\infty}  \lambda^{2r} \frac{t^{2m}}{(\lambda+t)^2}d\langle E_{|A|}(t)x,x\rangle.\]

The right-hand side is maximized for $\lambda = \frac{{r}}{1-r} t$ which shows that 
\[\begin{split} \sup_{\lambda>0}\Vert \lambda^{r} A^{m}(\lambda-A)^{-1}x \Vert^2 
&=\left( \frac{r}{1-r}\right)^{2r} \int_{0}^{\infty} \frac{t^{2(m+r)}}{t^2 ( 1+ \frac{r}{1-r})^2}d\langle E_{|A|}(t)x,x\rangle \\
&= r^{2r}(1-r)^{2-2r} \int_{0}^{\infty} t^{2(r+m-1)}d\langle E_{|A|}(t)x,x\rangle \\
&\sim \Vert |A|^{r+m-1} x\Vert^2.
 \end{split}\]
 Noting that the maximising $\lambda$ is independent of $m$, this shows in particular for all $\alpha>0,$ with $r_\alpha\in(0,1)$ and $k_\alpha\in \N_0$, the unique numbers such that $\alpha=r_\alpha+k_\alpha,$ that 
 \begin{align*}
   \sup_{\lambda>0}\Vert \lambda^{r_\alpha} A(\lambda-A)^{-1}x \Vert_{D(A^{k_\alpha})} \sim   \Vert |A|^{r_\alpha} x\Vert + \Vert |A|^{\alpha} x\Vert
 \end{align*}
 Noting that $\Vert |A|^{r_\alpha} x\Vert\le \|x\| + \||A|^{\alpha}x\|,$ this gives by definition of the Favard norm $\|x\|_{F_\alpha(A)}\sim \|x\|_{D(|A|^{\alpha})}.$
 
 For the second part of the Lemma, using again the functional calculus for $r\in(0,1)$ and $m\in\N$
\[\begin{split} \Vert \lambda^{r} K^m(\lambda-iK)^{-1}x \Vert^2 &= \int_{-\infty}^{\infty}  \lambda^{2r} \frac{t^{2m}}{\lambda^2 + t^2}d\langle E_K(t)x,x\rangle
\end{split}\]

The right-hand side is maximized for $\lambda = \frac{\sqrt{r}}{\sqrt{1-r}} t$ which shows that 
\[\begin{split} \sup_{\lambda>0}\Vert \lambda^{r} K^m(\lambda-iK)^{-1}x \Vert^2 
&=\left(\frac{r}{1-r}\right)^{r} \int_{-\infty}^{\infty} \frac{t^{2(m+r)}}{\frac{r}{1-r} t^2+ t^2}d\langle E_{K }(t)x,x\rangle \\
&=r^{r} (1-r)^{1-r} \int_{-\infty}^{\infty} t^{2(r+m-1)}d\langle E_{ K }(t)x,x\rangle \\
&\sim \Vert \vert K\vert^{r+m-1} x\Vert^2.
\end{split}\]
Following the same argument as above, this shows $\|x\|_{F_\alpha(iK)} \sim \|x\|_{D(|K|^\alpha)}.$
 \end{proof}
 
The two results presented in Lemma~\ref{lemm:SA} can be unified in the following way:
For normal operators $N$, defining contraction semigroups, we define commuting self-adjoint operators \cite[Prop.~5.30]{S12} $\Re(N):=\overline{(N+N^*)}/2$ and $\Im(N):=\overline{(N-N^*)}/(2i)$ such that $N=\Re(N)+i\Im(N).$ For $N$ to generate a contraction semigroup, we demand that $\Re(N)\le 0.$ Then, we recall that $T_t=e^{tN} := e^{t\Re(N)} e^{it\Im(N)}= e^{it\Im(N)}e^{t\Re(N)}.$ 
We then have the following generalization of Lemma \ref{lemm:SA}, which we prove in the appendix.
\begin{prop}
\label{prop:normal}
    Let $N$ be a normal operator on a Hilbert space with $\Re(N) \le 0$ as above, then for $\alpha \in (0,1)$
     \[ D(\vert N \vert^{\alpha})=F_{\alpha}(\Re(N))\cap F_{\alpha}(\Im(N)) = F_{\alpha}(N).\]
\end{prop}

Using this characterization of the Favard spaces, we obtain the following Corollary~\ref{corr:positive_selfop} from Theorem~\ref{thm:TrotterState2} and Remark~\ref{rem:Kgammabloed}.
It can be seen as an analogue of Corollary \ref{corr:positive_op} without the requirement that $0 \in \rho(A).$
 
\begin{corr}
\label{corr:positive_selfop}
Let $X$ be a Hilbert space and $L$ and $A$ be generators of strongly continuous contraction semigroups on $X$ with $-A$ self-adjoint and positive semidefinite. Moreover, assume that $D(|A|^{\alpha}) \subseteq D(L)$, and $LD(|A|^{\gamma})\subseteq D(|A|^{\beta})$  for some positive numbers $\alpha < \min\{1,\gamma\}$ and $ \gamma -1 < \beta\le \gamma < 2.$ 
Then for all $n\in\mathbb N$ and $t\ge 0$ the Trotter product satisfies
\begin{align}
     \left\|\left(\left(e^{tL/n}e^{tA/n}\right)^n - e^{t(A+L)}\right)\right\|_{D(|A|^{\gamma}) \to X}=\mathcal{O}\left(\frac{t^{1+\delta}}{n^\delta}\right),
\end{align}
  where $\delta= \min\{1,\beta,\gamma-\alpha\}.$
\end{corr}

\subsubsection{Unitary dynamics and singular potentials}
\label{sec:unitary}

In this section, we use the characterization of the Favard spaces in Lemma~\ref{lemm:SA} to formulate the result of Theorem~\ref{thm:TrotterState2} in the case of unitary dynamics on a Hilbert space.
Hence, we shall now assume that $X$ is a Hilbert space and that $A = iK$, with $K$ being self-adjoint. 

\begin{ex}[Schr\"odinger operators]\label{ex:Laplace} For the case of $K=-\Delta\ge 0$ on $X=L^2(\mathbb R^d)$, which is self-adjoint on the domain $D(K) =H^2(\mathbb R^d),$ Lemma~\ref{lemm:SA} gives that the Favard spaces $F_{\alpha}\equiv F_\alpha(iK)$ are precisely given by the Sobolev spaces  $ H^{2\alpha}(\mathbb{R}^d)$ for all $\alpha>0$ 
\end{ex}
\begin{ex}[Dirac operators]
Let $d=2$ for simplicity, then the two-dimensional Dirac operator reads in terms of Pauli matrices $\sigma_i$
\[ K= \sum_{i=1}^2 \sigma_i D_{x_i} + \sigma_3 m.\]
By the standard Foldy-Wouthuysen transform, the operator $K$ is unitarily equivalent to 
\[ \operatorname{diag}( \sqrt{(-\Delta)+m^2},-\sqrt{(-\Delta)+m^2}).  \]
From this, it follows that $F_{\alpha}\equiv F_\alpha(iK)= H^{\alpha}(\mathbb R^2; \CC^2).$
\end{ex}

\delete{To satisfy the assumptions of Theorem \ref{theo:positive_op} we require that $L:D(K^{\gamma}) \to D(K^{\beta})$ and
\begin{equation}
\label{eq:estm3}
\Vert K^{\beta} Lf \Vert \lesssim \Vert f \Vert  + \Vert K^{\gamma}f \Vert
\end{equation}
and that $L$ is relatively $K^{\alpha}$ bounded.}

Using Lemma~\ref{lemm:SA} we obtain the following corollary from Theorem~\ref{thm:TrotterState2} and Remark~\ref{rem:Kgammabloed} for the case of $A=iK$ and $L=iH$ for some self-adjoint operators $K$ and $H$:
\begin{corr}
\label{theo:singulaire}
Let $X$ be a Hilbert space and $H$ and $K$ be self-adjoint. Moreover, assume that $D(|K|^{\alpha}) \subseteq D(H)$, and $HD(|K|^{\gamma})\subseteq D(|K|^{\beta})$  for some positive numbers $\alpha < \min\{1,\gamma\}$ and $ \gamma-1 < \beta\le \gamma < 2.$    
Then for all $n\in\mathbb N$ and $t\in\R$  the Trotter product satisfies
\begin{align}
     \left\|\left(\left(e^{itH/n}e^{itK/n}\right)^n - e^{it(H+K)}\right)\right\|_{D(|K|^{\gamma}) \to X}=\mathcal{O}\left(\frac{|t|^{1+\delta}}{n^\delta}\right),
\end{align}
  where $\delta= \min\{1,\beta,\gamma-\alpha\}.$ 
\end{corr} 

\smallsection{Application to singular potentials}
 In the following, we apply Corollary~\ref{theo:singulaire} for the case of $X=L^2(\mathbb R^d)$, multiplication operator $H=V$ with singular potential $V(x)=|x|^{-a},$\footnote{Note that the assumptions on $H$ in Corollary~\ref{theo:singulaire} are independent of the sign of $H.$ Hence, all arguments here also work for the attractive singular potentials $V(x) = -|x|^{-a}.$} for some $a>0,$ $K=-\Delta\ge 0.$  
Here, from Lemma~\ref{lemm:SA} and Example~\ref{ex:Laplace}, we have that $F_\alpha(iK) = D(K^\alpha) = H^{2\alpha}(\R^d)$ for all $\alpha>0.$ 

The relative boundedness condition in Corollary~ \ref{theo:singulaire}, i.e., the condition $H^{2\alpha}(\R^d) \subseteq D(V)$ for some suitable $\alpha\in(0,1),$ can be obtained as follows: Assume that the potential satisfies $V =V_1+V_2\in L^p(\RR^d) + L^{\infty}(\RR^d)$ for some $p\ge 1.$ Then we have for $1/2=1/p+1/q$ 
\begin{equation}
\label{eq:estm_coul}
 \Vert V\varphi \Vert_2 \le \Vert V_1\Vert_p \Vert \varphi \Vert_{q} + \Vert V_2\Vert_{\infty}\Vert \varphi\Vert_2.
 \end{equation}
By the Sobolev embedding, we have $H^{k}(\mathbb R^d)\subseteq L^q(\mathbb R^d)$ for the parameters
\[ \frac{1}{2}-\frac{k}{d}=\frac{1}{q}\]
and hence we can choose $k=\frac{d(q-2)}{2q}$ and $\alpha=\frac{d(q-2)}{4q}.$  For the special case of $p=2$ and $q=\infty,$ we find $\alpha= \frac{d}{4}.$ For the potential $V(x)=|x|^{-a}$, we can choose $p=2$ as long as $a<d/2,$ in which case we have shown that $H^{d/2}(\R^d)\subseteq D(V).$ 

To verify the second condition in Corollary~\ref{theo:singulaire} for singular potentials, i.e., $V D((-\Delta)^\gamma)\subseteq D((-\Delta)^\beta)$ for some suitable $k_\gamma<\beta\le \gamma<2$, we can invoke the Brezis-Mironescu inequality \cite[Corr.4]{BM01} which states that for $s \in (1,\infty)$ and $\beta=\frac{s-1}{2}$ 
\begin{equation}
\label{eq:BrezisMiro}
\Vert (-\Delta)^{\beta}(fg) \Vert_2 \lesssim \Vert f \Vert_{\infty} \Vert g \Vert_{H^{s-1}} + \Vert g \Vert_{2s} \Vert  f \Vert_{H^s}^{1-1/s} \Vert f \Vert_{\infty}^{1/s}.
\end{equation}

We use this for $g=V_1$ being the singular part of the potential $V(x) = |x|^{-a}$ because the remainder can be chosen to be smooth and hence trivially satisfies the desired condition. In order for $V(x) = |x|^{-a}$ to be in $L^{2s}_{\text{loc}}(\mathbb R^d)$ (and hence $V_1\in L^{2s}(\mathbb R^d)$) we require that $a<\frac{d}{2s}.$ Thus, as above, $a<\frac{d}{2}$ is necessary for this condition to apply. More generally, for this choice of potential $V \in H^{s-1}_{\text{loc}}(\mathbb R^d) $ if and only if $a < \tfrac{d}{2} - (s-1).$ With this choice and furthermore $f =\varphi \in H^{2
\gamma}(\mathbb R^d)=D((-\Delta)^\gamma),$ we find from \eqref{eq:BrezisMiro} for $\gamma\ge \tfrac{d}{4},$ and hence by Sobolev embedding $\varphi \in L^\infty(\mathbb R^d),$ and $s \le 2\gamma$ 
\begin{equation}
\label{eq:BM}
 \Vert (-\Delta)^{\beta}(V_1 \varphi) \Vert_2 \lesssim \Vert \varphi \Vert_{\infty} \Vert V_1 \Vert_{H^{s-1}} + \Vert V_1 \Vert_{2s} \Vert  \varphi  \Vert_{H^s}^{1-1/s} \Vert \varphi \Vert_{\infty}^{1/s}.
 \end{equation}

Hence, collecting all constraints on $s$ and using $\beta =\tfrac{s-1}{2},$ we see $V D((-\Delta)^\gamma)\subseteq D((-\Delta)^\beta)$ for all $\beta <\min\{\tfrac{d-2a}{4a},\tfrac{d-2a}{4},\gamma-\tfrac{1}{2}\}$ 

For the case $a=1, d=3,$, that is, for the Coulomb potential, and furthermore $\gamma=1$, we can choose any $\beta <1/4.$ Thus, Corollary~\ref{theo:singulaire} gives in the Coulomb case $V(x) = \pm |x|^{-1}$ a Trotter convergence as
\[ \left\|\left(\left(e^{itV/n}e^{-it\Delta/n}\right)^n - e^{it(-\Delta+V)}\right)\right\|_{H^{2}(\mathbb R^3) \to L^2(\mathbb R^3)}= \mathcal O(n^{-1/4+\varepsilon}).\]
for all $t\in\R$ and $\varepsilon>0$ fixed. More generally, the above derivations give the following:
\begin{corr}
\label{cor:Coulomb}
Let $V(x)=\pm |x|^{-a}$ with $0< a<d/2$ and $\gamma\in(0,2).$ Then the Trotter formula converges as
\[ \left\|\left(\left(e^{itV/n}e^{-it\Delta/n}\right)^n - e^{it(-\Delta+V)}\right)\right\|_{H^{2\gamma}(\mathbb R^d) \to L^2(\mathbb R^d)}= \mathcal O\left(n^{d/4-\gamma}+ n^{-\beta} +n^{-1} \right)\]
for all $\gamma-1<\beta<\min\{\tfrac{d-2a}{4a},\tfrac{d-2a}{4},\gamma-\tfrac{1}{2}\}.$
\end{corr}

\smallsection{Molecular Hamiltonians} We consider a molecule with $N$ electrons of mass $m_\mathrm{e}$ and charge $-e$. In addition, we have $M$ nuclei of masses $m_j$ and charges $Z_j e$ for $j=1,...,M$ described by the many-body Schr\"odinger operator
\[ H = -\sum_{j=1}^N \frac{\Delta_{x_j}}{2m_\mathrm{e}} -\sum_{j=1}^M \frac{\Delta_{y_j}}{2m_j} + V(x,y) \]
where 
\begin{equation}\label{eq:potential}
    V(x,y) = \frac{1}{2} \sum_{i \neq j} \frac{e^2}{\vert x_i-x_j \vert} - \sum_{i, j} \frac{e^2Z_j}{\vert x_i-y_j \vert} + \frac{1}{2} \sum_{i \neq j} \frac{e^2}{\vert y_i-y_j\vert}.
\end{equation}

In the many-body setting, we can repeat the steps leading to the result of Corollary~\ref{cor:Coulomb} by working with (Sobolev-)Bochner spaces instead of standard Sobolev spaces.

To streamline the notation we shall no longer discriminate between $x$ and $y$ variables and just consider $x \in \mathbb R^{3(N+M)}.$

We shall allow for Trotter splittings \begin{equation}
\label{eq:Trotter_2}
A:=i\Delta - iV_1\text{ and }L:=-iV_2
\end{equation}
where $V_1$ and $V_2$ is any choice of summands in \eqref{eq:potential} such that $V=V_1+V_2.$ This includes, of course, the Trotter splitting into kinetic and potential energy, but also allows us to consider the kinetic energy together with the nucleus-nucleus and nucleus-electron interaction, separated from the electron-electron interaction. We have this freedom since combining Lemma \ref{lemm:stability} with Lemma \ref{lemm:SA}, we see that the Favard spaces $F_{\alpha}(A)$ are precisely the Sobolev spaces $H^{2\alpha}(\mathbb R^{N+M})$ for all $\alpha\in(0,1)$ and are independent of $V_1,$ since $V_1$ is relatively $-\Delta$ bounded. 

We shall show that
\begin{theo}
\label{theo:many-body}
For the Trotter-splitting \eqref{eq:Trotter_2}, there is for any $\varepsilon>0$ a constant $C_{\varepsilon,N,M}$ such that 
\[ \left\|\left(\left(e^{tL/n}e^{tA/n}\right)^n - e^{t(A+L)}\right)\right\|_{H^2(\mathbb R^{3(N+M)}) \to L^2(\mathbb R^{3(N+M)})}\le C_{\varepsilon,N,M} n^{-1/4+\varepsilon}.\]
\end{theo}
\begin{proof}

The potential $V$ consists of linear combinations of Coulomb potentials $V_{\operatorname{sp}}(x_i-x_j)$ in terms of the single particle Coulomb potential \[V_{\operatorname{sp}}(x):=|x|^{-1} \text{ for } x\in \mathbb R^3 \setminus \{0\}.\] 

We start by proving the $V_{\text{sp}}: H^{3/2}(\mathbb R^{3(N+M)})\to L^2(\mathbb R^{3(N+M)})$ boundedness.
Let $\varphi \in H^2(\mathbb R^{3(N+M)})$, we have the natural identification of the multi-variable $L^2$ as a single particle $L^2$ Bochner space
\[ \Vert V_{\text{sp}}\varphi\Vert_{L^2(\mathbb R^{3(N+M)})} =\Vert V_{\text{sp}}\varphi\Vert_{L^2(\mathbb R^3;X)},\]
where we defined $X:=L^2(\mathbb R^{3(N+M-1)})$. We also notice that $\varphi \in H^2(\mathbb R^3;X)$ which follows easily using the Fourier representation of the Sobolev space and the simple inequality $\langle \xi_1 \rangle \le \langle (\xi_1,..,\xi_{3(N+M)}) \rangle.$
This allows us to repeat the estimate in \eqref{eq:estm_coul} in the Bochner space and we obtain using the same composition of $V_{\text{sp}}=V_{\text{sp},1} + V_{\text{sp},2} \in L^2(\mathbb R^3) + L^{\infty}(\mathbb R^3)$ as in the subsection before
\begin{equation}
\label{eq:bound1a}
\begin{split}\Vert V_{\text{sp}}\varphi \Vert_{L^2(\mathbb R^3;X)} 
&\le \Vert V_{\text{sp},1}\Vert_{L^2(\mathbb R^3)} \Vert \varphi \Vert_{L^{\infty}(\mathbb R^3;X)} + \Vert V_{\text{sp},2}\Vert_{L^{\infty}(\mathbb R^3)}\Vert \varphi\Vert_{L^2(\mathbb R^3;X)}\\
&\lesssim \Vert V_{\text{sp},1}\Vert_{L^2(\mathbb R^3)} \Vert \varphi \Vert_{H^{3/2}(\mathbb R^3;X)} + \Vert V_{\text{sp},2}\Vert_{L^{\infty}(\mathbb R^3)}\Vert \varphi\Vert_{L^2(\mathbb R^3;X)}\\
&\lesssim \Vert V_{\text{sp},1}\Vert_{L^2(\mathbb R^3)} \Vert \varphi \Vert_{H^{3/2}(\mathbb R^{3(N+M)})} + \Vert V_{\text{sp},2}\Vert_{L^{\infty}(\mathbb R^3)}\Vert \varphi\Vert_{L^2(\mathbb R^3;X)}\\
& \lesssim (\Vert V_{\text{sp},1}\Vert_{L^2(\mathbb R^3)} +\Vert V_{\text{sp},2}\Vert_{L^{\infty}(\mathbb R^3)})\Vert \varphi \Vert_{H^{3/2}(\mathbb R^{3(N+M)})}.
\end{split}
\end{equation}

Next, we show that $V_{\text{sp}}: H^2(\mathbb R^{3(N+M)}) \to H^{s-1}(\mathbb R^{3(N+M)})$ where $s<3/2.$ It suffices to consider the singular part of the potential $V_{\text{sp},1}$.

Using Peetre's inequality, which implies that for $t \ge 0$ 
$$ \langle (\xi_1,..,\xi_m,\xi_{m+1},...,\xi_n)\rangle^t \le 2^t \langle (\xi_1,..,\xi_m) \rangle^t \langle (\xi_{m+1},...,\xi_n) \rangle^t, $$
we have for $Y:=H^{s-1}(\mathbb R^{3(N+M-1)})$ and $s<3/2$ together with the Brezis-Mironescu inequality that
\begin{equation}
\label{eq:bound2a}
\begin{split} \Vert V_{\text{sp},1}\varphi \Vert_{H^{s-1}(\mathbb R^{3(N+M)})} 
&\lesssim \Vert V_{\text{sp},1} \varphi \Vert_{H^{s-1}(\mathbb R^3;Y)}\\
&\lesssim \Vert \varphi \Vert_{H^{3/2}(\mathbb R^3;Y)} \Vert V_{\text{sp},1} \Vert_{H^{s-1}(\mathbb R^3)} + \Vert V_{\text{sp},1} \Vert_{L^{2s}(\mathbb R^3)} \Vert  \varphi  \Vert_{H^s(\mathbb R^3;Y)}^{1-1/s} \Vert \varphi \Vert_{H^{3/2}(\mathbb R^3;Y)}^{1/s}\\
&\lesssim \Vert \varphi \Vert_{H^{2}(\mathbb R^{3(N+M)})}(\Vert V_{\text{sp},1} \Vert_{H^{s-1}(\mathbb R^3)}+\Vert V_{\text{sp},1} \Vert_{L^{2s}(\mathbb R^3)}),
\end{split}
\end{equation}
where we used in the last line that for $\alpha,\beta \ge 0$
\[\langle (\xi_1,..,\xi_m) \rangle^{\alpha} \langle (\xi_{m+1},...,\xi_n) \rangle^{\beta} \le \langle(\xi_1,..,\xi_m,\xi_{m+1},...,\xi_n)\rangle^{\alpha+\beta}. \]
This extends the estimate \eqref{eq:BM} of the previous section to the many-particle Hilbert space. 

It remains to argue that we can lift the bounds \eqref{eq:bound1a} and \eqref{eq:bound2a} to the two-body interactions appearing in the many-body Hamiltonian, instead of the single particle interaction $V_{\operatorname{sp}}.$
To see this, notice that the two-particle interaction $W(x_1,x_2):=\sqrt{2} V_{\text{sp}}(x_1-x_2)$ -- without loss of generality, we consider now the first two particles -- is given by $W(x):=(V\circ U)(x),$
where we introduced the unitary map 
\[ U(x) := 2^{-1/2}(x_1-x_2,x_1+x_2,x_3,...,x_{N+M}).\] 

Using the Fourier representation of the Sobolev spaces, we find that the Sobolev norms are invariant under composition by a unitary map 
\[ \begin{split}
\Vert \langle -\Delta \rangle^{\beta} (s \circ U) \Vert_{L^2(\RR^{3(N+M)})}  &= \Vert \hat{s}\circ U \Vert_{L^2(\RR^{3(N+M)};\langle \xi \rangle^{2\beta} d\xi)} \\
&=\Vert \hat{s} \Vert_{L^2(\RR^{3(N+M)};\langle \xi \rangle^{2\beta} d\xi)} \\
&= \Vert \langle -\Delta \rangle^{\beta} s \Vert_{L^2(\RR^{3(N+M)})}. 
\end{split}\] 
This property allows us to lift bounds \eqref{eq:bound1a} and \eqref{eq:bound2a} to 
\begin{equation}
\label{eq:bound1a2}
\begin{split}\Vert W \varphi \Vert_{L^2(\mathbb R^3;X)} &=\Vert (V_{\text{sp}}\cdot(\varphi\circ U^*))\circ U \Vert_{L^2(\mathbb R^3;X)} =\Vert (V_{\text{sp}}\cdot(\varphi\circ U^*)) \Vert_{L^2(\mathbb R^3;X)} \\
& \lesssim (\Vert V_{\text{sp},1}\Vert_{L^2(\mathbb R^3)} \Vert+\Vert V_{\text{sp},2}\Vert_{L^{\infty}(\mathbb R^3)}\Vert )\Vert \varphi \circ U^* \Vert_{H^{3/2}(\mathbb R^{3(N+M)})} \\
& \lesssim (\Vert V_{\text{sp},1}\Vert_{L^2(\mathbb R^3)} \Vert+\Vert V_{\text{sp},2}\Vert_{L^{\infty}(\mathbb R^3)}\Vert )\Vert \varphi \Vert_{H^{3/2}(\mathbb R^{3(N+M)})} 
\end{split}
\end{equation}
and 
\begin{equation}
\label{eq:bound2a2}
\begin{split} 
\Vert W_{1}\varphi \Vert_{H^{s-1}(\mathbb R^{3(N+M)})}  &= 
\Vert (V_{\text{sp},1}\cdot (\varphi \circ U^*))\circ U \Vert_{H^{s-1}(\mathbb R^{3(N+M)})} =\Vert V_{\text{sp},1}\cdot (\varphi \circ U^*) \Vert_{H^{s-1}(\mathbb R^{3(N+M)})}\\
&\lesssim \Vert \varphi \circ U^* \Vert_{H^{2}(\mathbb R^{3(N+M)})}(\Vert V_{\text{sp},1} \Vert_{H^{s-1}(\mathbb R^3)}+\Vert V_{\text{sp},1} \Vert_{L^{2s}(\mathbb R^3)}) \\
&\lesssim \Vert \varphi  \Vert_{H^{2}(\mathbb R^{3(N+M)})}(\Vert V_{\text{sp},1} \Vert_{H^{s-1}(\mathbb R^3)}+\Vert V_{\text{sp},1} \Vert_{L^{2s}(\mathbb R^3)}).
\end{split}
\end{equation}
The identities \eqref{eq:bound1a2} and \eqref{eq:bound2a2} complete the proof, and the result follows from Theorem~\ref{thm:TrotterState2} and the characterization of Favard spaces in terms of Sobolev spaces as explained before the statement of the result.
\end{proof}

In summary, we have shown that for very general input states, we can get arbitrarily close to the $\mathcal O(n^{-1/4})$ rate that was observed for the ground state of the single-particle Coulomb Hamiltonian in \cite{BFHJY23}, even for many-body Coulomb Hamiltonians.

\section{Convergence rates for energy-limited unitary dynamics}
\label{sec:energy_limited}
\def\H{\mathcal H}
\def\ip#1#2{\langle #1,#2 \rangle}
\def\norm#1{{\|#1\|}}
\def\placeholder{{\,\cdot\,}}
\def\abs#1{{|#1|}}
\def\D{{\mathcal D}}
\def\oo{\infty}
\def\energy#1{\mathbf E(#1)}
\def\energyp#1{\mathbf E'(#1)}

In this section, we present a method for obtaining convergence rates inspired directly by quantum theory. Unlike the previous section, which focused on possibly singular but small perturbations, this approach is motivated by the study of regular but large perturbations tailored towards applications in quantum mechanics. 
Our approach is based on the "operator $E$-norm" introduced by Shirokov \cite{S20} and the notion of "energy-limited dynamics" introduced in \cite{vL24}.

We fix a separable Hilbert space $\H$ and a positive self-adjoint operator with $\inf \Spec G=0$.
We may think of $\H$ as the Hilbert space describing a quantum system and of $G$ as an energy-observable (a reference Hamiltonian) for this system.
For vectors $\psi\in\H$, the \emph{energy-expectation value} $\energy \psi$ is given by
\begin{equation}
    \energy\psi := \norm{G^{1/2}\psi}^2
\end{equation}
if $\psi \in D(G^{1/2})$ and by $\energy\psi :=\oo$ otherwise.

In the following, we use the language of Sobolev towers, which we introduced in Subsection \ref{sec:FarvardSpaces} and will briefly recall for the reader's convenience the setting they are used in this section.
For $n\in\NN$, we set $\H_n = D(G^{n/2})$ which is a Hilbert space in the graph inner product $\ip\psi\phi_n = \ip\psi\phi+\ip{G^{n/2}\psi}{G^{n/2}\phi}$.
Note that $\H_1$ is precisely the space of finite-energy vectors because
\begin{equation}
     \energy\psi =\norm\psi_1^2 -\norm\psi^2.
\end{equation}
For negative $n$, $\H_n$ is defined as the completion of $\H$ with respect to the inner product $\ip\psi\phi_{n} = \ip\psi{(1+G^{n})\phi}$.
We have continuous embeddings
\begin{equation}
    \ldots \supset \H_{-1} \supset \H \supset \H_1\supset \ldots.
\end{equation}
The dual space of $\H_n$ is naturally identified with $\H_{-n}$.
Thus, the adjoint $A^*$ of a bounded operator $A\in B(\H_n, \H_m)$ is naturally a bounded operator $A^*:\H_{-m}\to\H_{-n}$.
An operator $A\in B(\H_n,\H_{-n})$ is called {Hermitian} if $A=A^*$ \cite[p.\ 192]{RS75}.
Note that relatively $G^{1/2}$-bounded operators on $D(G^{1/2})$ are precisely the elements of $B(\H_1,\H)$.
Similarly, {Hermitian} quadratic forms $a$ on $\H$ with form domain $Q(a)=\H_1$ for which $M>0$ exists such that $\pm a \le M(G+1)$, i.e.,
\begin{equation}
    \pm a(\psi,\psi) \le M \norm\psi_1^2,\qquad \psi\in\H_1,
\end{equation}
may be regarded as {Hermitian} operators $A\in B(\H_1,\H_{-1})$ via the equation 
\begin{equation}
    a(\psi,\phi) = \ip\phi{A\psi}_{-1}, \qquad \psi,\phi\in \H_1.
\end{equation}

Following \cite{S20}, the operator $E$-norm is defined as:

\begin{defi}[$E$-norm, Shirokov \cite{S20}]
    The \emph{operator $E$-norm} $\norm\placeholder_E^G$ is defined on the space $B(\H_1,\H)$ of $G^{1/2}$-bounded operators for $E>0$ by
\begin{equation}
    \norm{A}_E^G = \sup \big\{ \norm{A\psi} ;  \norm\psi =1,\ \energy\psi \le E\big\}.
\end{equation}
\end{defi}

We remark that the optimization may be restricted to a core $\D$ of $G^{1/2}$, i.e., a $\|\cdot\|_1$-dense subspace $\D\subset\H_1$, \cite{vL24}.
It is shown in \cite{S20} that the operator $E$-norm is equivalent to the operator norm in $B(\H_1,\H)$.
In particular, operator $E$-norms at different energy constraints $E>0$ are equivalent. In fact, Shirokov showed \cite{S20}:
\begin{equation}\label{eq:equiv_norms}
    \norm\placeholder_E^G \le\norm\placeholder_{E'}^G \le \sqrt{\frac{E'}E}\,\norm\placeholder_E^G, \qquad E'\ge E>0.
\end{equation}
Since we are typically interested in upper bounds on the operator $E$-norm, the following equivalent definition via a minimization (see \cite[Prop.~2.19]{vL24}) will be useful:
\begin{equation}\label{eq:variational_principle_opE}
    \norm{A}_E^G = \min \big\{ \sqrt{\lambda E + E_0} : \lambda,E_0\ge0,\ A^*A \le \lambda G +E_0 \big\}.
\end{equation}
The operator inequality in \eqref{eq:variational_principle_opE} is in the sense of quadratic forms on $\H_1=D(G^{1/2})$.
The reason that we are interested in the operator $E$-norm is that convergence in $E$-norm implies strong convergence.\footnote{In fact, Shirokov showed that the operator $E$-norm induces the strong operator topology on bounded subsets of $B(\H)$ if the reference Hamiltonian has compact resolvent \cite{S20}.}
For instance, consider self-adjoint operators $A$ and $B$ on $\H$ such that $A+B$ is essentially self-adjoint, then the Trotter product converges by \cite{T59, K78}, and convergence rates in $E$-norm
\begin{equation}
    \norm{{(e^{itA/n}e^{itB/n})}^n-e^{it\overline{(A+B)}}}_E^G \le \eps(t,n,E),
\end{equation}
{ where $\overline{A+B}$ is the closure of $A+B$}, immediately yields state-dependent convergence rates: If $\psi\in\H_1$ is a unit vector, then
\begin{equation}
    \norm{{(e^{itA/n}e^{itB/n})}^n\psi-e^{it\overline{(A+B)}}\psi} \le \eps(t,n,\energy\psi).
\end{equation}
To prove such convergence rates for Trotter products in the operator $E$-norm, we need the dynamics to respect the energy scale determined by $G$.
For this, we follow \cite{vL24}.
The maximal output energy given an energy-constraint $E>0$ is described by \cite{W17}
\begin{equation}
    f_U(E) := \sup \big\{ \energy{U\psi} ; \norm\psi\le1,\ \energy\psi\le E\big\}, \qquad E>0. 
\end{equation}
The function $f_U$ is a non-decreasing concave function $\RR^+\to [0,\oo]$. 
It is connected to the operator $E$-norm via
\begin{equation}
    f_U(E) = \Big(\norm{G^{1/2}U}_E^G\Big)^2.
\end{equation}
Unitary operators for which $f_U$ is finite are called energy-limited \cite{W17,vL24,S19}.
Thus, a unitary $U$ is energy-limited if and only if it restricts to an operator $B(\H_1,\H_1)$.

\begin{defi}[Energy-limited dynamics {\cite{vL24}}]
    A unitary one-parameter group $U(t)$ is \emph{energy-limited} if there exist `stability' constants $\omega,E_0\ge0$ such that
    \begin{equation}\label{eq:fU_bound}
        f_{U(t)}(E) \le e^{\omega \abs t}(E+E_0)-E_0, \qquad  t\in\RR,\ E>0.
    \end{equation}
    A collection of unitary one-parameter groups $U_j(t)$, $j\in I$, is \emph{jointly energy-limited} if there exist joint stability constants $\omega,E_0$ such that \eqref{eq:fU_bound} holds for all $U_j(t)$.
\end{defi}

Note that the right-hand side of \eqref{eq:fU_bound} is a semigroup of affine functions of $E$ for times $t>0$.
Energy-limitedness guarantees that the energy expectation value $E(t):=\energy{U(t)\psi}$ depends continuously on $t$ for vectors $\psi\in \H_1$. In fact, it guarantees the continuity bound \cite{vL24}
\begin{equation}
    \abs{E(t+s)-E(s)} \le \omega \abs{t} (E(s)+E_0 \norm\psi^2) + \mathcal O(\abs{t}^2).
\end{equation}
Energy-limited unitary one-parameter groups are fully characterized by the following structure theorem:

\begin{theo}[{\cite[Thm.~3.7]{vL24}}]
    Let $H$ be a self-adjoint operator and let $U(t)=e^{-itH}$.
    Let $\omega,E_0\ge0$.
    Then $U(t)$ is energy-limited with stability constants $\omega,E_0$ if and only if 
    \begin{enumerate}[(i)]
        \item $\H_1$ is invariant under $U(t)$ and the restriction $U_0(t):=U(t)\upharpoonright \H_1$ is a strongly continuous one-parameter group of bounded operators on $\H_1$.
        \item
        The inequality $\pm i [G,H]\le \omega (G+E_0)$ holds in the sense that
        \begin{equation}
            |\ip{G^{1/2}\psi}{G^{1/2}H\psi} - \ip{G^{1/2}H\psi}{G^{1/2}\psi}| \le \omega(\norm{G^{1/2}\psi}^2+E_0\norm\psi^2)
        \end{equation}
        for all $\psi \in \H_1\cap D(H)$ such that $ H\psi\in \H_1$.
    \end{enumerate}
\end{theo}

These necessary and sufficient conditions for energy-limitedness are rather hard to check in practice.
Recall that an operator $H$ is $G$-bounded if and only if $H$ is in $B(\H_2,\H)$.
{In the following we mention two results that make energy-limitedness checkable through explicit relative boundedness criteria:

\begin{prop}[{\cite[Thm.~3.9]{vL24}}]\label{thm:EL_thm}
    Let $H$ be a $G$-bounded self-adjoint operator and let $\omega,E_0\ge0$ be such that the operator inequality $\pm i[H,G]\le \omega(G+E_0)$ holds in the sense of quadratic forms on $\D$, i.e.,
    \begin{equation} \label{eq:commutator_bound}
        \abs{\ip{H\psi}{G\psi}-\ip{G\psi}{H\psi}} \le \ip\psi{\omega(G+E_0)\psi}, \qquad \psi\in\D,
    \end{equation}
    then $U(t)$ is energy-limited with stability constants $\omega,E_0$.
\end{prop}

}

Readers familiar with Nelson's commutator theorem \cite[Thm.~X.37]{RS75} may recognize that the assumption of self-adjointness is redundant:
If $\D$ is a core for $G$, every symmetric $G$-bounded operator $H$ satisfying \eqref{eq:commutator_bound} on $\D$ is essentially self-adjoint and satisfies the same inequality on $\D=\H_2$.

The following result is due to Fröhlich \cite{F77}.

{
\begin{prop}[Fröhlich {\cite{F77}}]\label{thm:frohlich}
    Let $H$ be a self-adjoint operator satisfying the assumptions of \cref{thm:EL_thm} with respect to a reference operator $G_0$ with $\inf\Spec G=0$ with stability constants $\omega,E_0>0$.
    If $-[H,G_0]^2\le \nu^2 (G_0+E_0)^2$ holds in the sense that\footnote{The left-hand side of \eqref{eq:commutator_rb} is defined as the supremum over $|\ip{H\phi}{G_0\psi}-\ip{G_0\phi}{H\psi}|$ where $\phi$ ranges over all unit vectors in $\H_2$.}
    \begin{equation}\label{eq:commutator_rb}
        \norm{[H,G_0]\psi} \le \nu \norm{(G_0+E_0)\psi}, \qquad \psi \in \H_2,
    \end{equation}
     then $U(t)$ is energy-limited with stability constants $2\nu,E_0^2$ with respect to $G=(G_0+E_0)^2-E_0^2$.
\end{prop}


\begin{proof}
    We want to apply \cite[Lem.~2]{F77} with $N=G_0+E_0$ to obtain $\norm{Ne^{itH}\psi} \le e^{\nu \abs t}\norm{N\psi}$, which is equivalent to the claim.
    However, \cite[Lem.~2]{F77} contains two additional assumptions.
    The assumption that $N\ge1$ is never used in Fröhlich's paper and may be replaced by $N\ge c$ for some scalar $c>0$, which holds in our case (since we assume $E_0>0$).
    The other assumption in \cite[Lem.~2]{F77} is that $[H,G_0]$ satisfies the hypothesis of Nelson's commutator theorem. 
    As Fröhlich remarks in \cite[Rem.~1]{F77}, this can be replaced by the relative boundedness assumption $\norm{[H,G_0]\psi}\le \nu \norm{N\psi}$ for all $\psi\in\H_2$ that we assumed.
\end{proof}
}

We now come to the main theorem of this section, which proves convergence rates for the Trotter product formula of energy-limited unitary groups in the operator $E$-norm.

\begin{theo}\label{thm:EL_trotter}
    Let $A$ and $B$ be self-adjoint operators that generate energy-limited unitary one-parameter groups and let $\omega, E_0$ be joint stability constants.
    Let $\D$ be a core for $G^{1/2}$ such that 
    \begin{equation}
        (t,s)\mapsto e^{itA}e^{isB} \psi,\quad \text{and}\quad (t,s)\mapsto e^{itB}e^{isA}\psi 
    \end{equation}
    are both in $C^2(\RR^2,\H)$ for all $\psi\in\D$.
    If the commutator $[A,B]:\D\to\H$ is $G^{1/2}$-bounded, then $A+B$ is essentially self-adjoint on $\D$ and
        \begin{equation} \label{eq:trotter_estimate}
        \|\big(e^{itA/n}e^{itB/n}\big)^n - e^{i\overline{(A+B)}}\|_E^G
        \le \frac{t^2}{2n}\norm{[A,B]}_{f_{2t}(E)}^G 
    \end{equation}
    for all $n\in\mathbb N$, where $f_t(E) = E + (e^{\omega |t|}-1)(E+E_0)$.
\end{theo}

We remark that the $ C^2$ assumption in \cref{thm:EL_trotter} ensures that the commutator $[A,B]$ makes sense as an operator on the core $\D$.
Furthermore, as seen in Step 2 in the proof of Theorem~\ref{thm:EL_trotter}, essential self-adjointness of $A+B$ on $\D$ is a direct consequence of the convergence of the Trotter product. This is in sharp contrast to the setup in Section~\ref{sec:RelBoundTrotter} in which the fact that the sum of the generators is again a generator of a strongly continuous contraction semigroup or even a unitary group could be deduced directly from the employed relative boundedness assumptions.

Recall that an operator $K$ in a Hilbert space generates a contraction semigroup if and only if it is maximally dissipative, i.e., a dissipative operator that admits no proper dissipative extensions \cite{ACE23}.

\begin{lemm}\label{lem:telescope}
    Let $A,B$ be self-adjoint operators generating energy-limited unitary dynamics with joint stability constants $\omega,E_0$. Let $iC$ be a maximally dissipative operator. 
    Then
    \begin{equation}\label{eq:pub}
        \|(e^{itA/n}e^{itB/n})^n - e^{itC}\|_E^G
        \le n\norm{e^{itA/n}e^{itB/n}-e^{itC/n}}_{f_{2t-2t/n}(E)}^G,
    \end{equation}
    where $f_t(E)=e^{\omega t}E+(e^{\omega t}-1)E_0$ and $t>0$. 
    If $C$ is self-adjoint and $e^{itC}$ is energy-preserving, i.e., commutes with $G$ for all $t\in\RR$, we have
    \begin{equation}\label{eq:pub2}
        \|(e^{itA/n}e^{itB/n})^n - e^{itC}\|_E^G
        \le n\norm{e^{itA/n}e^{itB/n}-e^{itC/n}}_{E}^G.
    \end{equation}
\end{lemm}
\begin{proof}
    Set $V(t)=e^{itA}e^{itB}$. Then $f_{V(t)}(E)\le f_{2t}(E)$ and
\begin{align*}
    \norm{V(t/n)^n - e^{itC}}_E^G
    &= \bigg\|{\sum_{j=1}^n e^{i(t(j-1)/n)C} (V(t/n)-e^{itC/n}) V(t/n)^{n-j}}\bigg\|_E^G\nonumber\\
    &\le \sum_{j=1}^n \norm{(V(t/n)-e^{itC/n}) V(t/n)^{n-j}}_E^G\nonumber\\
    &\le\sum_{j=1}^n \norm{V(t/n)-e^{itC/n}}_{f_{2t\frac{n-j}n}(E)}^G\nonumber\\
    &\le n\norm{V(t/n)-e^{itC/n}}_{f_{2t-2t/n}(E)}^G, 
\end{align*}
where we used the semigroup property of $f_t$ and that $f_t(E)$ is monotone in $t$.
The second claim follows similarly by sorting the dynamics generated by $C$ in the first equation to the right instead of the left.
\end{proof}

\begin{lemm}\label{lem:key_comm}
    Under the assumptions of \cref{thm:EL_trotter}, we have
    \begin{equation}
        \norm{e^{-itA}Be^{itA}-B}_E^G \le t \norm{[A,B]}_{f_t(E)}^G.
    \end{equation}
\end{lemm}
\begin{proof}
    {Assume $t>0$.}
    Let $\psi,\phi$ be unit vectors in $\D$ and let $\ip\psi{G\psi}\le E$.
    Then
    \begin{align*}
        \ip\phi{(e^{-itA}Be^{itA}-B)\psi} 
        &= \int_0^t \frac{d}{ds} \ip{e^{isA}\phi}{Be^{isA}\psi}\,ds\\
        &= \int_0^t \ip{Ae^{isA}\phi}{Be^{isA}\psi}-\ip{e^{isA}\phi}{BAe^{isA}\psi}\,ds\\
        &= \int_0^t \ip{e^{isA}\phi}{[A,B]e^{isA}\psi}\,ds
    \end{align*}
    Optimizing over unit vectors $\phi,\psi\in \D$ with $\ip\psi{G\psi}\le E$, we get 
    \begin{equation*}
        \norm{e^{-itA}Be^{itA}-B}_E^G \le \int_0^t \norm{[A,B]}_{f_s(E)}^G \,ds\le  t\norm{[A,B]}_{f_t(E)}^G.\qedhere
    \end{equation*}
\end{proof}
We need the following Lemma on dissipative extensions of symmetric operators:
\begin{lemm}[{\cite[Thm.~2.5]{ACE23}}]\label{lem:dissipative_ext}
    Let $T$ be a symmetric operator on a separable Hilbert space $\H$.
    Let $S\supset T$ be an extension such that $iS$ is dissipative, i.e., an extension such that $\Im\ip\psi{S\psi} \le 0$, $\psi\in D(S)$,
    then it follows that
    \begin{equation}
        T\subset S \subset T^*\quad\text{and}\quad T\subset S^* \subset T^*.
    \end{equation}
\end{lemm}
{
\begin{lemm}\label{lem:step1}
    Under the assumptions of \cref{thm:EL_trotter}, we have
    \begin{equation}
        \norm{e^{itA}e^{itB}-e^{itC}} \le \frac{t^2}{2} \norm{[A,B]}_{f_{2t}(E)}^G, \qquad t>0,
    \end{equation}
    for every maximally dissipative extension $iC \supset i(A+B)$.
\end{lemm}

Note that, once we prove \cref{thm:EL_trotter}, the \cref{lem:step1} implies that 
\begin{equation}
    \norm{e^{itA}e^{itB}-e^{it\overline{(A+B)}}} \le \frac{t^2}{2} \norm{[A,B]}_{f_{2t}(E)}^G, \qquad t\in\RR.
\end{equation}
}
\begin{proof}
    Note that $f_{e^{itA}}(E),\,f_{e^{itB}}(E)\le f_{\abs t}(E)$.
    {The operator $-iC^*$ is also maximally dissipative}, and, by \cref{lem:dissipative_ext}, $C^*$ is also an extension of $A+B$.
    {Consider unit vectors $\psi,\phi\in \D$ with $\energy\psi\le E$.}
    We have 
    \begin{align*}
        \ip\phi{(e^{itA}e^{itB}-e^{itC})\psi}
        &=\ip{e^{-itC^*}\phi}\psi-\ip\phi{e^{itA}e^{itB}\psi} \\
        &= \int_0^t \frac{d}{ds} \ip{e^{-isC^*}\phi}{e^{isA}e^{isB}\psi}\,ds\\
        &= \int_0^t \Big(\ip{-iC^*e^{-isC^*}\phi}{e^{isA}e^{isB}\psi} - \ip{e^{-isC^*}\phi}{e^{isA}(A+B)e^{-isB}\psi}\Big)ds \\
        &= \int_0^t \Big(\ip{-ie^{-isC^*}\phi}{Ce^{isA}e^{isB}\psi} - \ip{e^{-isC^*}\phi}{e^{isA}(A+B)e^{-isB}\psi}\Big)ds \\
        &= i\int_0^t \ip{e^{-isC^*}\phi}{((A+B)e^{isA}-e^{isA}(A+B) )e^{isB}\psi}\,ds \\
        &= i\int_0^t \ip{e^{-isC^*}\phi}{[B,e^{isA}]e^{isB}\psi}\,ds, 
    \end{align*}
    where we used that $\phi\in D(C^*)$ and that $C$ equals $A+B$ on $e^{isA}e^{isB}\D\subset D(A+B)$. 
    Using that $e^{-isC^*}$ is a contraction, we get
    \begin{equation}
        |\ip\phi{(e^{itA}e^{itB}-e^{itC})\psi}| \le \int_0^t \norm{[B,e^{isA}]e^{isB}\psi}\,ds.
    \end{equation}
    Optimizing over $\phi$ gives
    \begin{align*}
        \norm{(e^{itA}e^{itB}-e^{itC})\psi}
        \le \int_0^{t} \norm{[B,e^{isA}]e^{isB}\psi}\,ds 
        \le \int_0^{t} \norm{[B,e^{isA}]}_{f_{s}(E)}^G ds.
    \end{align*}
    Since $\norm{[B,e^{itA}]}_E^G = \norm{e^{-itA}Be^{itA}-B}_E^G$, we can apply \cref{lem:key_comm} to obtain
    \begin{align}
        \norm{(e^{itA}e^{itB}-e^{itC})\psi} 
        &= \int_0^{t} s\norm{[A,B]}_{f_{2s}(E)}^G \,ds \le \frac{t^2}{2} \norm{[A,B]}_{f_{2t}(E)}^G.\label{eq:purple}
    \end{align}
    {
    Optimizing this over unit vectors $\psi\in\D$ with $\energy\psi\le E$ shows the claim.
    }
\end{proof}

{
Before we move on, we note the following consequence of \cref{lem:step1} and \cref{thm:EL_trotter}:
\begin{corr}
    Under the assumptions of \cref{thm:EL_trotter}, it holds that
    \begin{equation}
        \norm{e^{itA}e^{itB}-e^{itB}e^{itA}}_E^G \le t^2 \norm{[A,B]}_{f_{2t}(E)}^G,\qquad E>0,\ t\in\RR.
    \end{equation}
\end{corr}

}

Using Lemmas \ref{lem:telescope} -- \ref{lem:dissipative_ext}, we can prove \cref{thm:EL_trotter}:

\begin{proof}[Proof of \cref{thm:EL_trotter}.]
    As stated in the theorem, we set $f_t(E) = e^{\omega t}(E+E_0)-E_0$. Note that $f_{e^{itA}}(E),\,f_{e^{itB}}(E)\le f_{\abs t}(E)$.

   \emph{Step 1.}
    In the first step, we establish the desired bound \eqref{eq:trotter_estimate} for $t>0$ and a maximally dissipative extension $iC$ of the skew-symmetric operator $i(A+B)$.
    \Cref{lem:telescope}, gives
    \begin{equation}\label{eq:roma}
        \norm{(e^{itA/n}e^{itB/n})^n-e^{itC}}_E^G \le n \norm{e^{itA/n}e^{itB/n}-e^{itC/n}}_{f_{2t-2t/n}(E)}^G.
    \end{equation}
    {Applying \cref{lem:step1}, we obtain}
    \begin{equation}\label{eq:almond}
        \norm{(e^{itA/n}e^{itB/n})^n-e^{itC}}_E^G \le \frac{t^2}{2n} \norm{[A,B]}_{f_{2t}(E)}^G,
    \end{equation}
    {where we used the semigroup property $f_{t}\circ f_{s}= f_{t+s}$.}

    \emph{Step 2.}
    Since the commutator is $G^{1/2}$-bounded by assumption, the right-hand side of \eqref{eq:almond} is finite and goes to zero as $n\to\infty$. 
    This implies that the Trotter product converges strongly for $t>0$ to the dynamics generated by a maximally dissipative extension $i(A+B)|_\D$.
    Since this holds for all maximally dissipative extensions of $i(A+B)|_{\D}$, it follows that $i(A+B)|_\D$ has a unique maximally dissipative extension. This extension generates an isometry semigroup because the strong limits of unitaries are isometries, and therefore $C$ is a symmetric operator.
    Thus, one of the defect indices of $(A+B)|_\D$ is zero.
    If we apply the same arguments (including step one) with $A$ and $B$ replaced by their negatives, we learn that the same holds for $-i(A+B)|_\D$.
    Thus, $(A+B)|_\D$ is essentially self-adjoint since both defect indices are zero.
    We conclude that \eqref{eq:almond} holds with $C=\overline{(A+B)}$ for $t\in\RR$.
\end{proof}

Using the variational principle \eqref{eq:variational_principle_opE} or the estimate \eqref{eq:equiv_norms}, we can further bound the right-hand side of \eqref{eq:trotter_estimate}

\begin{corr}\label{cor:EL_trotter}
    Under the assumptions of \cref{thm:EL_trotter}, the convergence rate of the Trotter product is bounded by
    \begin{equation}\label{eq:trotter_estimate1}
        \|\big(e^{itA/n}e^{itB/n}\big)^n - e^{i\overline{(A+B)}}\|_E^G 
        \le \frac{t^2}{2n} \norm{[A,B]}_E^G\cdot \sqrt{1+ (e^{2\omega \abs t}-1)(1+\tfrac{E_0}E)}
    \end{equation}
    Furthermore, if $M,E_1\ge0$ are such that $-[A,B]^2\le M(G+E_0)$ holds in the sense
    \begin{equation}\label{eq:abstract_commutator_rb}
        \norm{[A,B]\psi}^2 \le M^2 (\|G^{1/2}\psi\|^2 + E_1 \|\psi\|^2),\qquad \psi\in \D,
    \end{equation}
    then 
    \begin{equation} \label{eq:trotter_estimate_rb}
        \|\big(e^{itA/n}e^{itB/n}\big)^n - e^{i\overline{(A+B)}}\|_E^G
        \le \frac{t^2}{2n}  M \sqrt{e^{2\omega \abs t}(E+E_0) + E_1-E_0}.
    \end{equation}
\end{corr}

In particular, when we choose $E_1=E_0$ the bound \eqref{eq:trotter_estimate_rb} takes the simple form
\begin{equation}\label{eq:trotter_estimate_rb2}
    \|\big(e^{itA/n}e^{itB/n}\big)^n - e^{i\overline{(A+B)}}\|_E^G
    \le \frac{t^2}{2n}  Me^{\omega \abs t} \sqrt{E+E_0}.
\end{equation}

In cases where the joint dynamics is energy-preserving, i.e., commutes with the reference Hamiltonian, the exponential behavior in the time parameter can be tamed:

\begin{corr}\label{cor:EL_trotter_energy_pres}
    Let $A,B$ be self-adjoint operators satisfying the assumptions of \cref{thm:EL_trotter} and let $\omega,E_0, f_t(E)$ be as in that theorem.
    Assume that $\overline{A+B}$ has energy-preserving dynamics with respect to $G$.
    Then 
    \begin{equation}\label{eq:EL_Trotter_energy_pres}
        \|\big(e^{itA/n}e^{itB/n}\big)^n - e^{i\overline{(A+B)}}\|_E^G
        \le \frac{t^2}{2n} \norm{[A,B]}_{f_{2t/n}(E)}^G
    \end{equation}
    In particular, if $M>0$ is such that $-[A,B]^2\le M(G+E_0)$ in the sense of \eqref{eq:abstract_commutator_rb}, then
    \begin{equation}
        \|\big(e^{itA/n}e^{itB/n}\big)^n - e^{i\overline{(A+B)}}\|_E^G
        \le \frac{t^2}{2n}  M e^{\omega t/n}\sqrt{E+E_0}.
    \end{equation}
\end{corr}
\begin{proof}
    This follows from \eqref{eq:pub2} in \cref{lem:telescope} and \eqref{eq:purple}:
    \begin{align*}
        \|\big(e^{itA/n}e^{itB/n}\big)^n - e^{i\overline{(A+B)}}\|_E^G
        &\le n \norm{e^{itA/n}e^{itB/n} -e^{it\overline{(A+B)}/n}}_E^G\\
        &\le \frac{t^2}{2n} \norm{[A,B]}_{f_{2t/n}(E)}^G.\qedhere
    \end{align*}
\end{proof}

In \cref{sec:numerics}, we will see numerically that these bounds correctly predict the actual convergence rates of the infinite-dimensional problem.

\subsection{Application to Schrödinger operators}\label{subsec:Schroedinger}

In this section, we consider Schrödinger operators $-\Delta + V(x)$ with potentials $V$ whose second derivatives are bounded. 
Additionally, we need polynomial boundedness:

We write $P$ and $Q$ for the momentum and position operators $P\psi(x) = -i\psi'(x)$ and $Q\psi(x) =x\psi(x)$ defined on their natural domains.
Note that the Schrödinger operator $-\Delta+V(x)$ equals $P^2+V(Q)$.
Throughout this section, we write $N = P^2+Q^2$ for the harmonic oscillator. 

\begin{lemm}
\label{lemm:derivatives}
    Let $V\in C^4(\RR,\RR)$ with $V$ and its first 4 derivatives being polynomially bounded.
    Then for Schwartz function $\psi \in \mathcal S(\mathbb R)$ we have that the maps
    \begin{equation}
        (t,s)\mapsto e^{itP^2}e^{isV(Q)}\psi 
        \quad\text{and}\quad
        (t,s)\mapsto e^{itV(Q)}e^{isP^2}\psi
    \end{equation}
    are in $C^2(\RR^2,\H)$. In particular, the commutator $[V(Q),P^2]$ is well-defined on $\mathcal S(\RR)$.
\end{lemm}

\begin{proof}
    
We shall prove the claim for $G(t,s):=e^{itP^2}e^{isV(Q)}\psi$, the other ordering follows from analogous arguments.
We first recall that, since $e^{isV(Q)}\psi  \in D(P^2)=H^2(\mathbb R)$, we can differentiate in $t$ to find
\[ \partial_t G(t,s) = ie^{itP^2} P^2(e^{isV(Q)}\psi).\]
Since by assumption $P^{2}(e^{isV(Q)}\psi) \in D(P^2)$ also, we can apply the same argument again to find that
\[ 
\partial^2_t G(t,s) = -e^{itP^2} P^{4}(e^{isV(Q)}\psi).
\]
We can explicitly compute the derivative $ (iP)^{2}(e^{isV(Q)}\psi)$ to see that $\partial_t^2 G(t,s)$ is continuous in $(t,s).$
Finally, since
\[ F(s):=P^{2}(e^{isV(x)}\psi(x)) =e^{isV(x)}(2is \psi'(x)V'(x)+ \psi''(x) + s\psi(x)(iV''(x)-sV'(x)^2))  \]
is differentiable in $s$ with 
\[ \begin{split} F'(s)&=iV(x) e^{isV(x)}(2is \psi'(x)V'(x)+ \psi''(x) + s\psi(x)(iV''(x)-sV'(x)^2)) \\
& \quad + ie^{isV(x)}(2V'(x)\psi'(x) + \psi(x)(V''(x)+2isV'(x)^2)) \end{split} \]
by Lebesgue's differentiation theorem. This implies that $\lim_{h \rightarrow 0} \frac{\Vert F(s+h)-F(s)-hF'(s)\Vert_{L^2}}{h}=0.$
Continuity of $e^{itP^2}$ thus allows us to interchange limits and conclude that
\[\partial_s\partial_t G(t,s) =\lim_{h \rightarrow 0} e^{itP^2} \frac{F(h+s)-F(s)}{h} = e^{itP^2}F'(s)= -e^{itP^2}P^2 (e^{isV(x)}V(x) \psi)(x)\]
exists. 
The continuity of the derivative $\partial_s\partial_t G(t,s)$ then just follows from the (strong) continuity of the semigroup $e^{itP^2}$ and the continuity of the expression of $F'(s)$ in $s$. 

The continuity of the semigroups $(e^{itP^2})$, as above, and since $\psi \in D(V(Q)^k)$ for $k \ge 0,$ implies that 
\[ \partial_s^k G(t,s) = e^{itP^2} e^{isV(x)}(iV(x))^k \psi.\]
By strong continuity of the two groups, this also implies that $(t,s) \mapsto \partial_s^k G(t,s)$ is continuous.

Finally, since $e^{isV(x)}(iV(x)) \psi \in D(P^2)$, we have 
\[ \partial_t \partial_s G(t,s) =-e^{itP^2} P^2 e^{isV(x)} V(x)\psi(x)  \]
and by working out the derivative $P^2 e^{isV(x)} {V(x)}\psi(x)$ and using the strong continuity of the semigroups, we find that $(t,s) \mapsto \partial_t\partial_s G(t,s)$ is continuous which implies the claim. 
\end{proof}

Since we assume $V''$ is uniformly bounded, we can bound\footnote{
Let $f$ be Lipschitz continuous (e.g.\ $V'$ above). Then it has a weak derivative $f'\in L^\infty$. Let $a = |f(0)|$ and $M = \|f'\|_{L^\infty}$ be the optimal Lipschitz constant. Then noting that $|x|\le (1+x^2)/2$ we have
\begin{align*}
    f(x)^2 = (f(x) - f(0) + f(0))^2\le M^2x^2 + 2aM|x| + a^2\le M(M+a)x^2 + a(M+a)\le w^2(1+x^2)
\end{align*}
where $w=M+a$.}  
\begin{equation}\label{eq:Lipschitzbound}
    |V'(x)|\le (|V'(0)| + \norm{V''}_{L^{\infty}})(1+x^2)^{1/2}
\end{equation}
and, hence, $|V(x)| \lesssim 1+x^2$.
In particular, this implies that $V(Q)$ is relatively bounded with respect to $N=P^2+Q^2$.

\begin{lemm}\label{lem:commutator}
    Let $V$ satisfy the assumption of \cref{lemm:derivatives}.
    Then the commutator $[V(Q), P^2]$ is relatively bounded with respect to $N$.
    The energy-constrained operator norm with respect to the reference Hamiltonian $G=N^2-1$ is bounded by
    \begin{equation}\label{eq:opE_commutator}
        \|[V(Q),P^2]\|_E^G \le \nu \sqrt{E+1}
    \end{equation}
    where $\nu = 5\|V''\|_{L^\infty} + 4|V'(0)|$.
\end{lemm}

\begin{proof}
    First note that
    \begin{equation*}
        -i[V(Q),P^2] = PV'(Q) + V'(Q)P = 2V'(Q)P - iV''(Q).
    \end{equation*}
    Since $V''\in L^\oo$, we only have to show that $V'(Q)P$ is $N$-bounded.
    Let $w = \norm{V''}_{L^\infty} + |V'(0)|$ such that $|V'(x)|^2\le w^2(1+x^2)$, cf.\ \eqref{eq:Lipschitzbound}. Then $V'(Q)^2 \le w^2(1+Q^2)$ and, hence,
    \begin{equation*}
        \|V'(Q)P\psi\|^2 \le w^2 \| (1+Q^2)^{1/2}P\psi\|^2 = w^2(\|QP\psi\|^2 + \| P\psi\|^2)
    \end{equation*}
    for $\psi\in D(N)$.
    This reduces the problem to showing that $QP$ is $N$-bounded, which follows from the canonical commutation relations ($[Q,P]=i$):
    \begin{align*}
        2PQ^2P&=  PQPQ + PQ[Q,P] + QPQP + [P,Q]QP\\
        &= PQPQ+QPQP + i (PQ-QP) \\
        &= P^2Q^2 + P[Q,P]Q + Q^2P^2 + Q[P,Q]P + i(PQ-QP) \\
        &= P^2 Q^2 + Q^2 P^2 +2i(PQ-QP)\\
        &\le (P^2+Q^2)^2 + 2 = N^2+2 \le 3N^2,
    \end{align*}
    where we used that $N\ge1$.
    Putting everything together, we get
    \begin{align}
        \|[V(Q),P^2] \psi\| &\le \|V''\|_{L^\infty}\norm\psi +2 \|V'(Q)P\psi\|\nonumber \\
        &\le \|V''\|_{L^\infty}\norm\psi +2 w\big( \norm{QP\psi}^2 + \norm{P\psi}^2 \big)^{1/2}\nonumber \\
        &\le \|V''\|_{L^\infty}\norm\psi +2 w\big( \tfrac32 \norm{N\psi}^2 + \norm{N\psi}^2 \big)^{1/2}\nonumber \\
        &= \|V''\|_{L^\infty}\norm\psi +\sqrt{10}w\|N\psi\|\nonumber\\
        &{\le} \nu \norm{N\psi} 
        = \nu\norm{(G_0+1)\psi} = \nu\norm{ (G+1)^{1/2}\psi} \label{eq:babboe}
    \end{align}
    where {we used $\sqrt{10}\le4$ and set}
    $\nu = 4w + \|V''\|_{L^\infty} = 4|V'(0)| + 5\|V''\|_{L^\infty}$.
    Finally, eq.~\eqref{eq:opE_commutator} follows from eq.~\eqref{eq:variational_principle_opE}.
\end{proof}

\begin{lemm}\label{lem:energy_limitedness}
    Let $V$ satisfy the assumption of \cref{lemm:derivatives}.
    Take $G=N^2-1$ as the reference Hamiltonian.
    \begin{enumerate}
        \item
        The unitary one-parameter group $e^{itP^2}$ is energy-limited with stability constants $20,1$
        \item 
        The unitary one-parameter group $e^{itV(Q)}$ is energy-limited with stability constants $2\nu,1$ where $\nu$ is as in \cref{lem:commutator}.
    \end{enumerate}
    In particular, $\gamma,1$ are joint stability constants with $\gamma=\max\{2\nu,20\}$.
\end{lemm}
\begin{proof}
    We prove both items applying {\cref{thm:frohlich} with $G_0=N-1$, $E_0=1$}.
    We start with Statement (2).
    Clearly, both $V(Q)$ and $P^2$ are $N$-bounded.
    By eq.~\eqref{eq:babboe}, the inequality
    \begin{equation}
         -[V(Q),N]^2 = -[V(Q),P^2]^2 \le \nu (G+1)=\nu(G_0+1)^2
    \end{equation}
    holds in the sense of quadratic forms.
    Thus, \cref{thm:frohlich} implies that $e^{itV(Q)}$ is energy-limited with stability constants $2\nu,1$ with respect to $(G_0+1)^2 - 1 =G=N^2+1$.

    Statement (1) follows from the argument showing statement (2):
    $[P^2,N] = [P^2, Q^2]$ is of the above form with $V(x) = -x^2$. We have $|V'(0)|= 0$, $\norm{V''}_{L^{\infty}} = 2$ and thus, cf.\ \eqref{eq:opE_commutator},
    \begin{equation}
        -[P^2,N]^2 \le ((4+1)\cdot 2) N^2 = 10 (G+1).
    \end{equation}
    Therefore, $P^2$ is energy-limited with stability constants $20,1$ with respect to $G$, and the claim follows.
\end{proof}
\begin{theo}\label{thm:trotter_EL_schrodinger}
    Let $V\in C^4(\RR,\RR)$ have bounded second derivative and assume that $V^{(3)}$ and $V^{(4)}$ are polynomially bounded.
    Then $H=P^2 + V(Q)$ is essentially self-adjoint on $\mathcal S(\RR)$.
    Set $\nu=5\norm{V''}_{L^\oo}+ 4|V'(0)|$ and $\gamma= \max\{2\nu,20\}$.
    Then
    \begin{equation} \label{eq:trotter_el_schrodinger}
        \norm{\big(e^{-itP^2/n}e^{tV/n}\big)^n \psi - e^{-itH} \psi}
        \le \frac{t^2}{2n} \nu e^{\gamma |t|} \norm{N\psi}
    \end{equation}
    for all $\psi\in D(N)$.
\end{theo}
\begin{proof}
    We apply \cref{cor:EL_trotter} with $A=P^2$, $B=V(Q)$ and $\D = \mathcal S(\RR)$.
    Energy-limitedness of $e^{itP^2}$ and $e^{itV(Q)}$ with the claimed stability constants is proved in \cref{lem:energy_limitedness}.
    $G^{1/2}$-boundedness of the commutator $[A,B] = [V(Q),P^2]$ is proved in \cref{lem:commutator}, and the $C^2$ condition is proved in \cref{lemm:derivatives}.
    Then, \eqref{eq:abstract_commutator_rb} holds with $M = \nu$ and $E_0=1$.
    Therefore, we get (cp.\ \eqref{eq:trotter_estimate_rb2})
    \begin{equation}\label{eq:help42}
        \norm{\big(e^{-itP^2/n}e^{tV/n}\big)^n - e^{-itH}}_E^G \le \frac{t^2}{2n} \nu e^{\gamma \abs t} \sqrt{E+1}.
    \end{equation}
    If $\psi\in D(N)$ is a unit vector, the left-hand side of \eqref{eq:trotter_el_schrodinger} is bounded by the operator $E$-norm $\norm{\big(e^{-itP^2/n}e^{tV/n}\big)^n - e^{-itH}}_E^G$ with $E = \energy\psi = \norm{N\psi}^2-1$.
    Thus, \eqref{eq:trotter_el_schrodinger} follows from \eqref{eq:help42}.
\end{proof}

{
\subsection{Explicit example: Trotter splitting for the Harmonic oscillator}\label{sec:scheiss_oszillator}

Let us specialize to the Harmonic oscillator. 
In this case, the joint dynamics is energy-preserving relative to the reference Hamiltonian $G=N^2-1$.
Using \cref{cor:EL_trotter_energy_pres}, we get:

\begin{theo}\label{thm:scheiss_oszillator}
    Let $\psi\in D(N)$. Then
    \begin{equation}
        \norm{\big(e^{itP^2/n}e^{tQ^2/n}\big)^n \psi - e^{itN}\psi} \le 6\,\frac{t^2\norm{N\psi}}{n}  
    \end{equation}
    for all $n \ge 55t$.
\end{theo}

We will see numerically in \cref{sec:numerics} that this bound correctly predicts the actual convergence rates.

\begin{proof}
    By \cref{lem:energy_limitedness}, the dynamics of $P^2$ and $Q^2$ have joint stability constants $20,1$.
    \cref{cor:EL_trotter_energy_pres} gives 
    \begin{equation*}
        \norm{\big(e^{itP^2/n}e^{itQ^2/n}\big)^n - e^{itN}}_E^G \le \frac{t^2}{2n} \norm{[P^2,Q^2]}^{N^2-1}_{e^{20\abs t/n}(E+1)-1}.
    \end{equation*}
    By \cref{lem:commutator}, we have 
    \begin{equation*}
        \norm{[P^2,Q^2]}_E^{N^2-1}
        \le 10\sqrt{E+1}.
    \end{equation*}
    Taken together, these imply
    \begin{equation*}
        \norm{\big(e^{itP^2/n}e^{itQ^2/n}\big)^n - e^{itN}}_E^G \le \frac{t^2}{2n} 10\, e^{10t/n} \sqrt{E+1}
    \end{equation*}
    Since $10t/n \le \log(6/5)$ guarantees $5 e^{10t/n} \le 6$ and since and $10/\log(6/5)\le 55$ the claim follows. 
\end{proof}

}

\subsection{Application to magnetic Dirac operators}
\label{subsec:Dirac}

We consider the free Dirac operator in 2 and 3 dimensions, which takes the abstract form, for $D_{x_i}:=\frac{1}{i} \frac{\partial}{\partial x_i},$ 
\[ H_{\text{free}} = \begin{pmatrix} mc^2 & cD^* \\ cD & -mc^2 \end{pmatrix} \in \mathbb C^{2(n-1)\times 2(n-1)},\]
with Pauli matrices $\sigma_i$
\[ D=\begin{cases}
 \sum_{i=1}^3 \sigma_i D_{x_i} & \text{ if }n=3,\\
D_{x_1}  + i D_{x_2} & \text{ if }n=2,
\end{cases} \text{ with }D \in \mathbb C^{(n-1)\times (n-1)},\]
see \cite{Th92} for a general introduction to Dirac operators. Then we have the following:
\begin{lemm}
\label{lemm:Dirac_der}
   Let $V\in C^2(\RR^n,\CC^{2(n-1) \times 2(n-1)})$ {Hermitian} for $n \in \{2,3\}$ with $V$ and its first $2$ derivatives being polynomially bounded.
    If $\psi \in \mathcal S(\mathbb R^n; \CC^{2(n-1)}),$  then the maps
    \begin{equation}
        (t,s)\mapsto e^{itH_{\operatorname{free}}}e^{isV(Q)}\psi 
        \quad\text{and}\quad
        (t,s)\mapsto e^{itV(Q)}e^{isH_{\operatorname{free}}}\psi
    \end{equation}
    are in $C^2(\RR^2,\H)$, where $\H = L^2(\mathbb R^{n};\mathbb C^{2(n-1)\times 2(n-1)}).$
\end{lemm}
We omit its proof as it is very similar to Lemma \ref{lemm:derivatives}. It follows from \cite[Thm.~4.3]{Th92} that the Dirac operator $H_{\operatorname{free}}+V$ with $V$ as in Lemma \ref{lemm:Dirac_der} is essentially self-adjoint on $\mathcal S(\mathbb R^n; \CC^{2(n-1)}).$

\smallsection{Magnetic Dirac operators}

Unlike potential perturbations that we considered for Schr\"odinger operators, we may also consider quantum systems incorporating a magnetic field. 

For Dirac operators, magnetic fields are incorporated as follows:  In three dimensions, the magnetic vector potential $A \in C^{\infty}(\RR^3;\RR^3)$ gives a magnetic field $\mathbf B = \operatorname{curl}(A),$ whereas in two dimensions, the magnetic field strength just reduces to a scalar field $B = \partial_{x_1} A_2-\partial_{x_2} A_1$. 

The magnetic Dirac operator then takes the abstract form 
\[ H = \begin{pmatrix} mc^2 & cD^* \\ cD & -mc^2 \end{pmatrix},\]
with 
\[ D=\begin{cases}
 \sum_{i=1}^3 \sigma_i (D_{x_i} -A_i) & \text{ if }n=3,\\
(D_{x_1}-A_1) + i (D_{x_2}-A_2) & \text{ if }n=2
\end{cases}\]
and $D_{x_1}=-i\partial_{x_1}.$
The Pauli operator is then, for $\mathbf \sigma = (\sigma_1,\sigma_2,\sigma_3)$, defined by
\begin{equation}
\label{eq:Pauli}
H_P= D^*D=(-i\nabla-A)^2 - \begin{cases} \mathbf \sigma \cdot \mathbf B & \text{ if }n=3,\\
B & \text{ if }n=2.
\end{cases}
\end{equation}
\smallsection{Homogeneous magnetic fields}
For illustration purposes, we shall now consider a homogeneous magnetic field in two dimensions, i.e., $\mathbf B=(0,0,B_0)$ and $B_0 >0.$ We can then choose the unbounded(!) magnetic potential $A(x)=\frac{B_0}{2}(-x_2,x_1,0)$
that yields the constant magnetic field and observe that
\begin{equation}
\label{eq:commutator}
 [D,D^*]=2B_0.
 \end{equation}
From this, we can deduce that 
\[ \Spec(H_P ) =2 \mathbb N_0 B_0\]
such that $G:=H_P$ is a positive self-adjoint operator with $\inf \Spec G=0.$ 
We also notice that 
\[ D= 2D_{\bar z} - \frac{iB_0}{2}z \text{ with }z=x_1+ix_2\]
where $\partial_{\bar z} = \partial_{x_1}+i\partial_{x_2}$ with $D_{\bar z}=-i\partial_{\bar z}$ and $\partial_{ z} = \partial_{x_1}-i\partial_{x_2}$ with $D_{z}=-i\partial_{z}$ such that
\[\begin{split} H_P &= \left(D_{x_1} - \frac{B_0 x_2}{2}\right)^2 + \left(D_{x_2} + \frac{B_0 x_1}{2}\right)^2-B_0\\
&= -\Delta +B_0 (\bar z \partial_{\bar z}-z\partial_z) +\frac{B_0^2}{4} \vert z\vert^2 - 2B_0.\end{split}\]
We therefore consider a Trotter-splitting of the Hamiltonian with 
\begin{equation}
\label{eq:TROTTER}
K= \begin{pmatrix} 0 & 2D_{z} \\ 2D_{\bar z} & 0 \end{pmatrix} \text{ and }P=\begin{pmatrix} 0 & \frac{iB_0}{2}\bar z \\ -\frac{iB_0}{2}z & 0\end{pmatrix}.
\end{equation}

\begin{theo}\label{thm:trotter_EL_Dirac}
    For kinetic energy $K$ and magnetic potential $P$ as in \eqref{eq:TROTTER} and $G=(-\Delta + \vert z \vert^2-2)^2$ as in \cref{eq:Pauli}, the Trotter splitting for the magnetic Dirac operator satisfies for $\omega:=2\varepsilon\max \{1,B_0\}$, $E_0:=\frac{4}{\varepsilon \omega}(2 + \tfrac 1 {\varepsilon^2})\max\left\{B_0,1\right\}$ with $\varepsilon>0$ arbitrary
    \begin{equation}
        \norm{\big(e^{-itK/n}e^{-itP/n}\big)^n  - e^{-it(K+P)}}_{E}^{G}
        \le  \frac{t^2 B_0(\sqrt{6}+1)}{2n} \sqrt{e^{2\omega \abs t}(E+E_0)+2-E_0}.
    \end{equation}
 In particular, for $\varepsilon=1/\abs t$, we obtain a global-in time bound of the form 
     \begin{equation}
        \norm{\big(e^{-itK/n}e^{-itP/n}\big)^n  - e^{-it(K+P)}}_{E}^{G}
        =\mathcal O_{B_0}(t^2(1+\abs t)(1+\sqrt{E})/n)
    \end{equation}
\end{theo}

\begin{proof}
We apply \cref{eq:trotter_estimate_rb2} with $A=K$ and $B=P$. The $C^2$ condition follows from \cref{lemm:Dirac_der}. We write $N=-\Delta + \vert z \vert^2$ and $G_0=N-2$ such that $G_0^2 =G.$  

It then suffices, by \cref{thm:frohlich}, to compute and bound the relevant commutators: We compute
\begin{align*}
[ K, N] &=\begin{pmatrix} 0 & -2i\bar z \\ -2i z & 0 \end{pmatrix}& \text{ and }&&[ P, N]&=2B_0\begin{pmatrix} 0 & -D_{\bar z} \\  D_z & 0 \end{pmatrix},
\end{align*}
{so in particular
    $|[K,N]|^2, B_0^{-2} |[P,N]|^2\le 4N = 4(G_0 + 2)$.
Moreover, we recall the identity 
\[ [K,G] = [K,N]G_0 + G_0[K,N].\]
By this, the Cauchy-Schwarz inequality, Young's inequality for products\footnote{
It holds for $x,y\in\RR$, $\lambda > 0$ and $p,q >1$ with $1/p + 1/q = 1$ that
$$|xy|\le \tfrac \lambda p |x|^p + \tfrac {\lambda^{1-q}} q|y|^q$$
which we apply here for $p=q=2$ and once for $\lambda = \varepsilon$, resp. $\lambda = B_0\varepsilon$, and once for $\lambda = \delta$.
} and by above operator inequalities for the commutators we have for $\varepsilon>0$ arbitrary and $\delta = \varepsilon^2/2$ that
\[\begin{split} \vert \langle \psi, [K,G]\psi\rangle \vert &\le 2 \vert \langle G_0 \psi, [K,N]\psi \rangle \vert  \\
&\le \varepsilon\norm{G_0\psi}^2 + \tfrac 1 {\varepsilon}\norm{[K,N]\psi}^2\\
&\le \varepsilon\norm{G_0\psi}^2 + \tfrac{4}{\varepsilon} \langle (G_0+2) \psi,\psi \rangle\\
&= \varepsilon\norm{G_0\psi}^2 + \tfrac{4}{\varepsilon} \langle G_0 \psi,\psi \rangle + \frac{8}{\varepsilon}\Vert \psi \Vert^2\\
&\le \varepsilon\norm{G_0\psi}^2 + \tfrac{2\delta}{\varepsilon} \Vert G_0 \psi \Vert^2 + \frac{2}{\varepsilon \delta} \Vert\psi \Vert^2 + \frac{8}{\varepsilon}\Vert \psi \Vert^2\\
&= 2\varepsilon\norm{G_0\psi}^2 + \frac{4}{\varepsilon^3} \Vert\psi \Vert^2 + \frac{8}{\varepsilon}\Vert \psi \Vert^2\\
&= 2\varepsilon\norm{G_0\psi}^2 + 4\left(\frac{1}{\varepsilon^3} + \frac{2}{\varepsilon}\right) \Vert\psi \Vert^2
\end{split}\]

and similarly
\[\begin{split}  \vert \langle \psi, [P,G]\psi \rangle\vert &\le 
\varepsilon B_0\norm{G_0\psi}^2 + \tfrac 1 {\varepsilon B_0}\norm{[P,N]\psi}^2\\
&\le \varepsilon B_0\norm{G_0\psi}^2 + \tfrac{4B_0}{\varepsilon} \langle (G_0+2) \psi,\psi \rangle\\
&\le \varepsilon B_0\norm{G_0\psi}^2 + \tfrac{2\delta B_0}{\varepsilon} \Vert G_0 \psi \Vert^2 + \frac{2B_0}{\varepsilon \delta} \Vert\psi \Vert^2 + \frac{8B_0}{\varepsilon}\Vert \psi \Vert^2\\
&= 2\varepsilon B_0 \Vert G_0 \psi \Vert^2 + 4B_0\left(\frac{1}{\varepsilon^3} + \frac{2}{\varepsilon}\right) \Vert\psi \Vert^2.\end{split} \]
}
Thus, we find that $e^{itP}$ and $e^{itK}$ are energy-limited with stability constants $\omega:=2\varepsilon\max \{1,B_0\}$, {$E_0:=\frac{4}{\varepsilon \omega}(2 + \tfrac 1 {\varepsilon^2})\max\left\{B_0,1\right\}$} with respect to $G.$
We also compute
\[ \begin{split}
    [K,P] &=iB_0\operatorname{diag}(-(D_z z +\bar z D_{\bar z}),D_{\bar z} \bar z+ z D_{ z}) \\
    &=B_0\operatorname{diag}(-1 - (z\partial_{z} +\bar z \partial_{\bar z}),1 + (\bar z\partial_{\bar z}+ z \partial_{ z}))\\
    &=B_0\operatorname{diag}(-1 - (x_1 \partial_{x_1} + x_2 \partial_{x_2}),1 + x_1 \partial_{x_1} + x_2 \partial_{x_2}).
\end{split}.\]
From \eqref{eq:babboe} we conclude that 

\begin{equation}
\label{eq:no_pauli}
\begin{split}\Vert  [K,P] \psi \Vert &\le B_0( \Vert \psi \Vert + \Vert x_1 \partial_{x_1} \psi \Vert + \Vert x_2 \partial_{x_2}\psi \Vert) \\
&\le B_0 (\Vert \psi \Vert +  \sqrt{6} \Vert N \psi \Vert) \\ 
&\le B_0(\sqrt{6}+1) \Vert (G+2)^{1/2}\psi\Vert. \end{split}
\end{equation}
This shows that in the notation of \eqref{eq:abstract_commutator_rb}
 \[\Vert  [K,P] \psi \Vert^2 \le B_0^2(\sqrt{6}+1)^2 (\Vert G^{1/2} \psi\Vert^2 + 2\Vert \psi \Vert^2),\]
 we have $M=B_0(\sqrt{6}+1)$ and $E_1=2$. Thus, the result follows from \eqref{eq:trotter_estimate_rb}.
\end{proof}

\begin{rem}
    To close the estimate, it seemed necessary to consider again the harmonic oscillator as a reference Hamiltonian in the previous theorem. We were initially considering the Pauli operator, too, but were unable to bound $[K,P]$ in \eqref{eq:no_pauli} with respect to this reference operator.
\end{rem}

{\section{Numerical examples}
\label{sec:numerics}

The optimality of our convergence rates for the Coulomb potential has been numerically confirmed in \cite{BGHL23}. In this section, we will try to numerically support our findings that the true figure of merit that leads to slower than $\mathcal O(1/n)$ convergence rates in the Trotter splitting is not the occurrence of a singularity in the Coulomb potential, but the (low) Sobolev space regularity of the potential. We therefore decided not to study singular potentials, as this has already been investigated in \cite{BGHL23}, but bounded potentials of low regularity in Section \ref{subsec:LRP}, instead.  We then also investigate in Subsection \ref{subsec:conf} the quantum harmonic oscillator, which has been the guiding example of our previous section.

We compare our convergence rates on suitably chosen vectors to the true convergence rates of the Trotter splitting in large finite-dimensional approximations of the full infinite-dimensional Hilbert space. We focus on simple Schr\"odinger operators on $L^2(\mathbb R/\mathbb Z)$ with periodic boundary conditions.
We verified numerically that the dimension of the approximation is sufficiently large in the sense that the Trotter error remains constant when increasing it further. We interpret this as a heuristic confirmation that our numerics capture the convergence rate of the full infinite-dimensional Trotter splitting. A more comprehensive discussion of this latter point can be found in \cite{BGHL23}.

\subsection{Low-regularity potentials}
\label{subsec:LRP}

Here, we numerically evaluate the scaling of the Trotter error for the case $A= i\Delta$ on $X=L^2([0,1])$ and $L=-iV$ for some potential $V(x)$ of low regularity. Our focus lies on bounded but non-smooth potentials, in which case we expect a slow decay of the Trotter error from our results: A potential of low regularity means that in Corollary~\ref{theo:singulaire} the condition $VD((-\Delta)) \subseteq D((-\Delta)^\beta)$ only holds true for small values of $\beta$. 

\begin{figure}[ht!]\centering
\includegraphics[width=7cm]{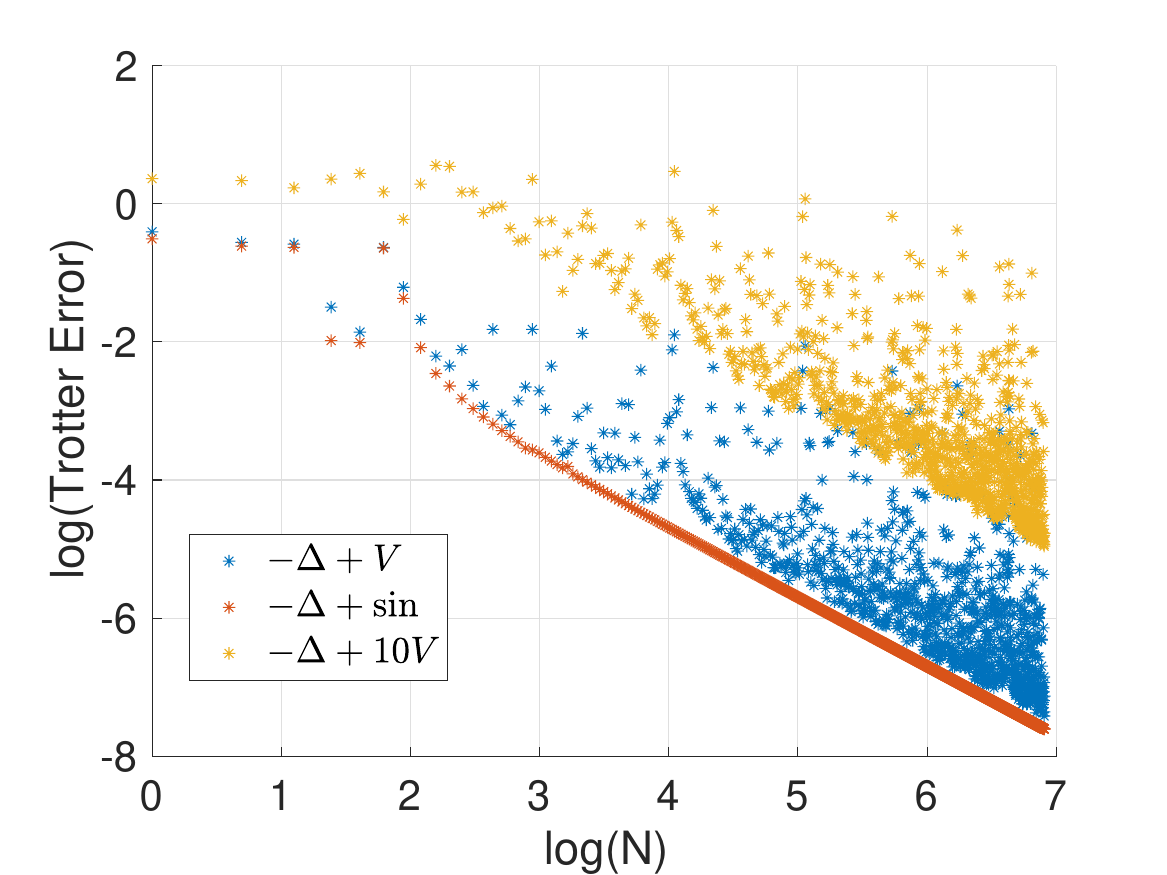}
\raisebox{.62cm}{\includegraphics[width=7cm]{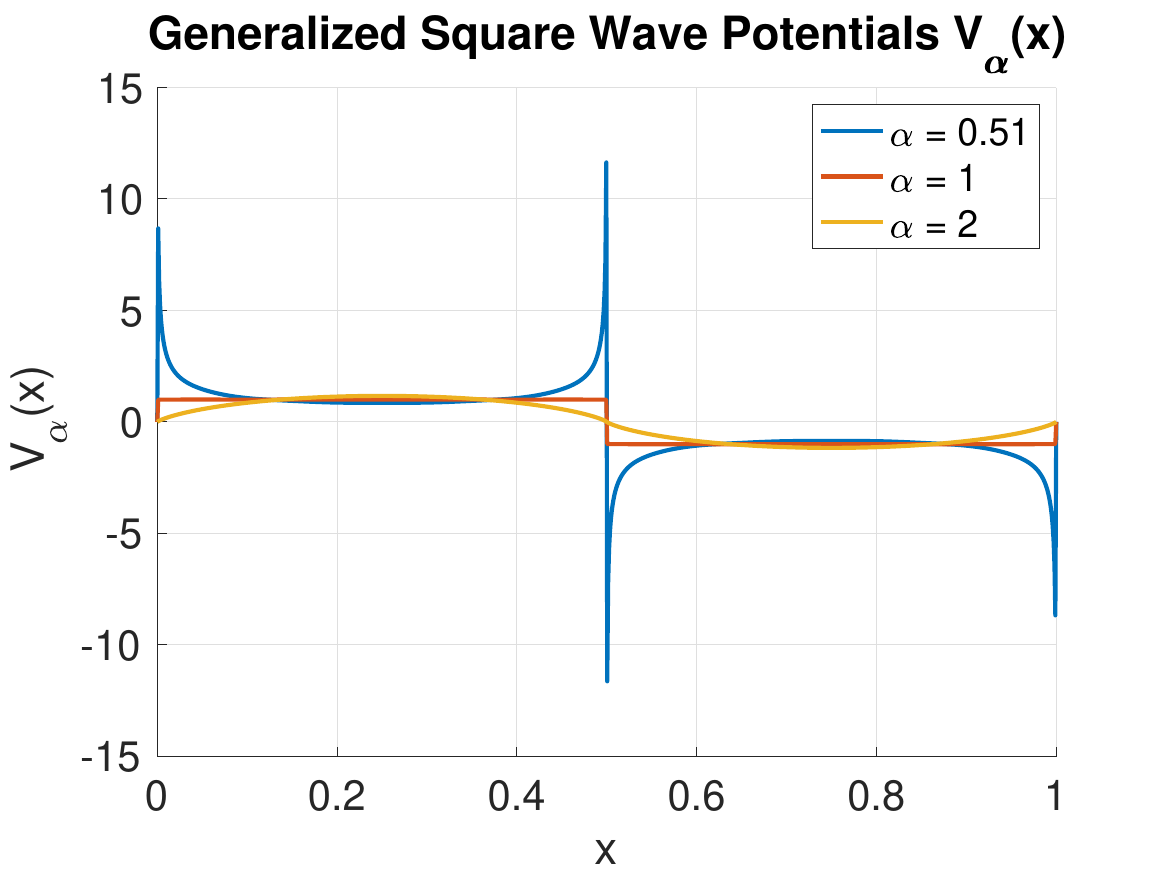}}
\caption{
\label{fig:Trotter}Left: Trotter error for the ground state of respective Schr\"odinger operator (numerically computed with finite matrix truncation in $\mathbb C^{2M+1}$ of size $M=400$) after time $t=1$. On the right, we see the potential $V_{\alpha}$ for $\alpha \in \{0.51,1,2\}.$ }
\end{figure}

\begin{figure}[ht!]
\includegraphics[width=5cm]{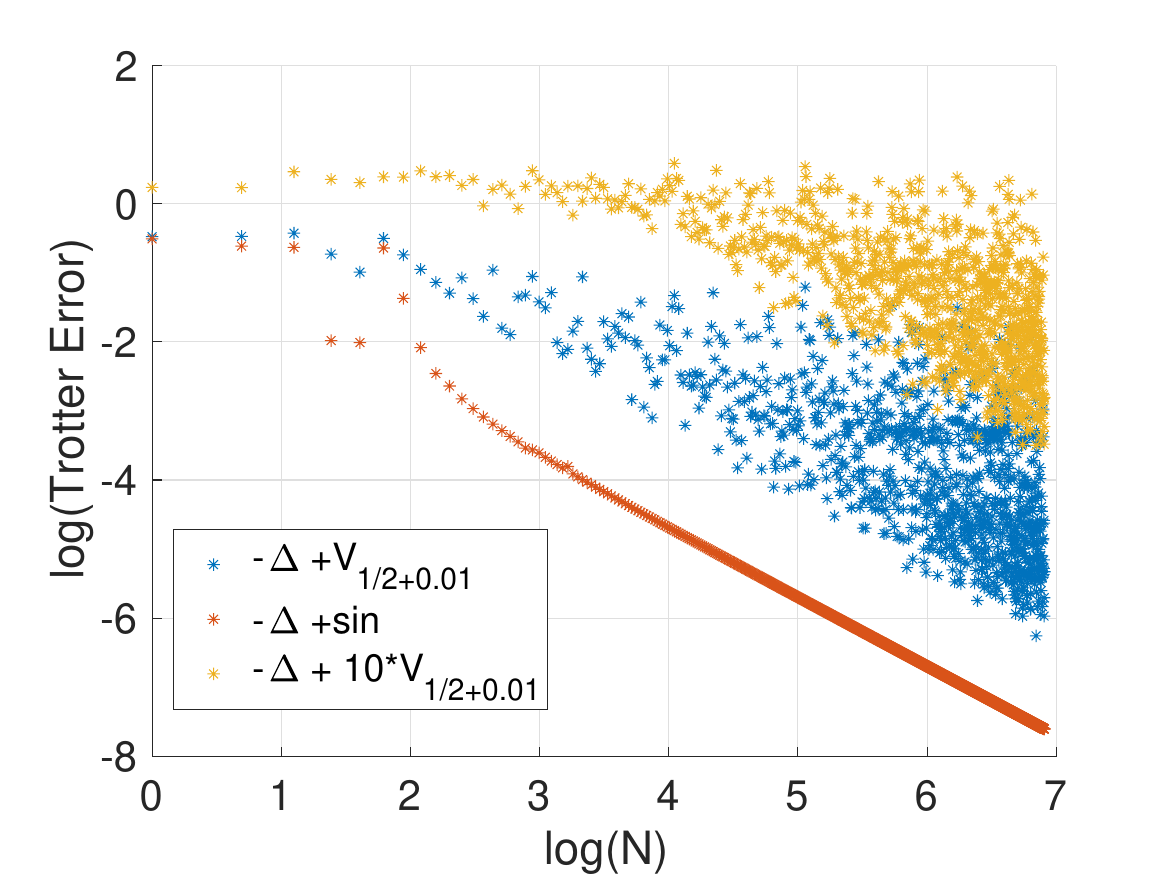}
\includegraphics[width=5cm]{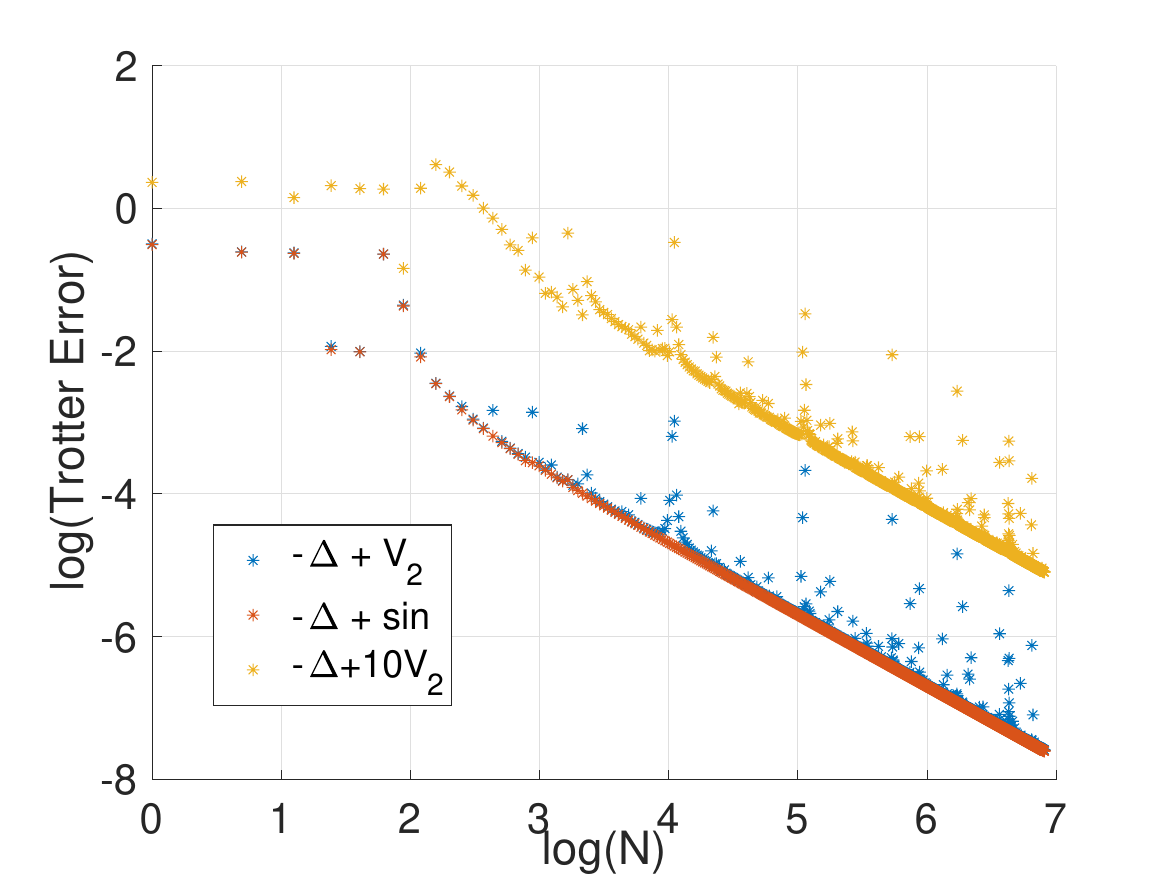}
\includegraphics[width=5cm]{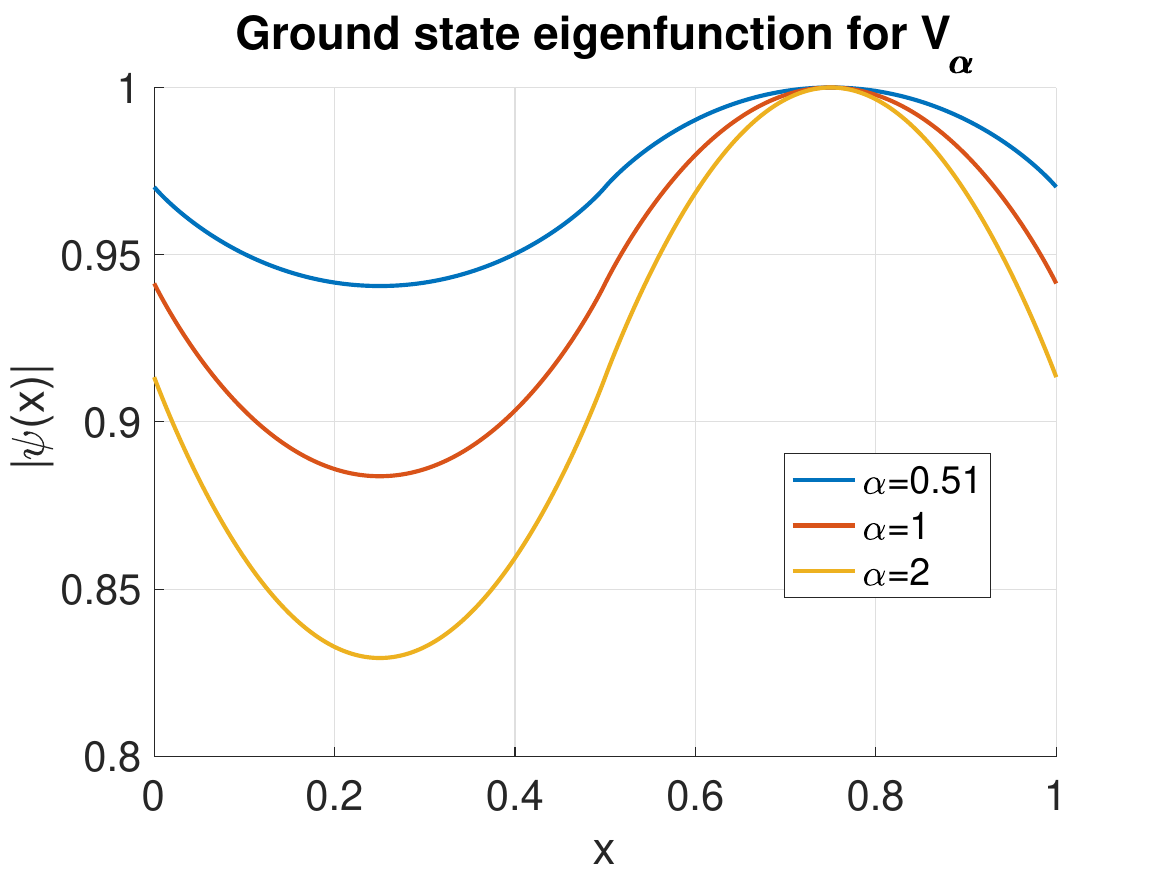}
\caption{
\label{fig:Trotter2} Trotter error for ground state of respective Schr\"odinger operator (numerically computed with finite matrix truncation in $\mathbb C^{2M+1}$ of size $M=400$) after time $t=1$ (very non-smooth potential on the left, more regular one in the center). $N$ indicates the number of iterations in the Trotter product. On the right, we show the ground-state wavefunctions of $-\Delta+V_{\alpha}$ for relevant $\alpha.$}
\end{figure}

We consider the square wave potential over the interval $V(x)=2(H(2x)-H(2x-1))-1$ where $H$ is the Heaviside function $H(x) = \indic_{[0,\infty)}(x).$
When periodically continued, this function has a Fourier series representation 
\[V(x) = \frac{4}{\pi} \sum_{n \in 2\mathbb N_0+1} \frac{e^{2\pi i nx}-e^{-2\pi i nx}}{2in}.\]
We can generalize this by defining the family of potentials $V_{\alpha}$ with $\alpha>1/2$ by
\[V_{\alpha}(x) := \frac{4}{\pi} \sum_{n \in 2\mathbb N_0+1} \frac{e^{2\pi i nx}-e^{-2\pi i nx}}{2in^{\alpha}}.\]
This way $V_{\alpha} \in H^s(\mathbb R/\mathbb Z)$ for $s<\alpha-1/2.$ We see in Figure \ref{fig:Trotter2} that this directly influences the Trotter error.

This implies that with respect to the Fourier basis, $V_{\alpha}$ can be written as 
\[ V_{\alpha} =\frac{2}{\pi} \sum_{n \in 2\mathbb N_0+1} \frac{J^n-(J^*)^n}{in^{\alpha}}, \]
where $J$ is the right shift.
The Fourier series of the Laplacian with respect to the basis $(e^{2\pi i nx})_{n \in \ZZ}$ is 
\[ -\Delta  = \operatorname{diag}(4\pi^2 n^2)_{n \in \ZZ}.\]

The Trotter error $\Vert (e^{i\Delta/n}e^{-iV/n})^n \psi - e^{-i(-\Delta+V)} \psi\Vert$ for $V$ above and $V(x)=\sin(2\pi x)$ on the interval $[0,1]$ with periodic boundary conditions is shown in Figure \ref{fig:Trotter} where the respective initial state $\psi$ is the (numerically computed) ground state of the Schr\"odinger operator $-\Delta+V$, respectively. The ground state of the Schr\"odinger operator is always at least in $H^2(\mathbb R/ \mathbb Z)$ and illustrated for the $V_{\alpha}$ potentials in Figure \ref{fig:Trotter2} (right). 

Since we use matrix approximations for our numerics, the Trotter error always scales as $\mathcal O(n^{-1}),$ but the error decays much slower for the low regularity potentials. We see in Figure \ref{fig:Trotter2} that this directly influences the Trotter error. On the left, we almost see no uniform decay in Figure \ref{fig:Trotter2} (left) for the $V_{1/2+0.01}$ potential, while the $V_2$ potential seems to follow the $\mathcal O(1/n)$ decay rate at large. This is to be expected from our error bounds, as the $V_{1/2+0.01}$ potential maps $D(-\Delta)$ to $D((-\Delta)^{\beta})$ for $\beta>0$ very small.

\subsection{Confining potentials}
\label{subsec:conf}
In Figures \ref{fig:Trotter3} and \ref{fig:Trotter4}, we consider the one-dimensional quantum harmonic oscillator $H=P^2+Q^2$ with the Trotter splitting into $P^2$ and $Q^2$, which we implement using the usual creation and annihilation operator representation in the Fock basis, where
\begin{figure}[ht!]
\includegraphics[width=7cm]{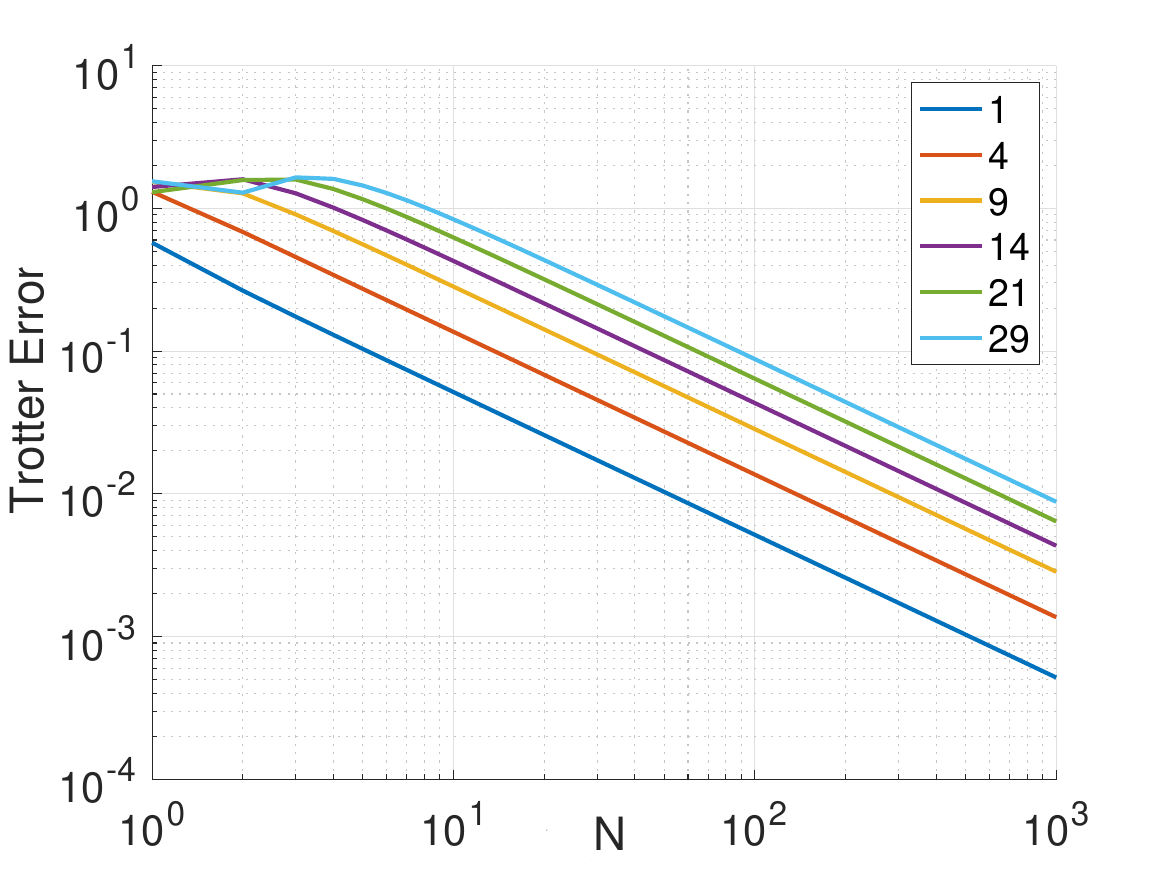}

\caption{
\label{fig:Trotter3} Trotter error for eigenstates (Fock states) $\{1,4,9,14,21,29\}$ of quantum harmonic oscillator (numerically computed with finite matrix truncation in $\mathbb C^{400}$) after time $t=1$. (log-log plot (left)). }
\end{figure}
\begin{figure}[ht!]
    \includegraphics[width=7cm]{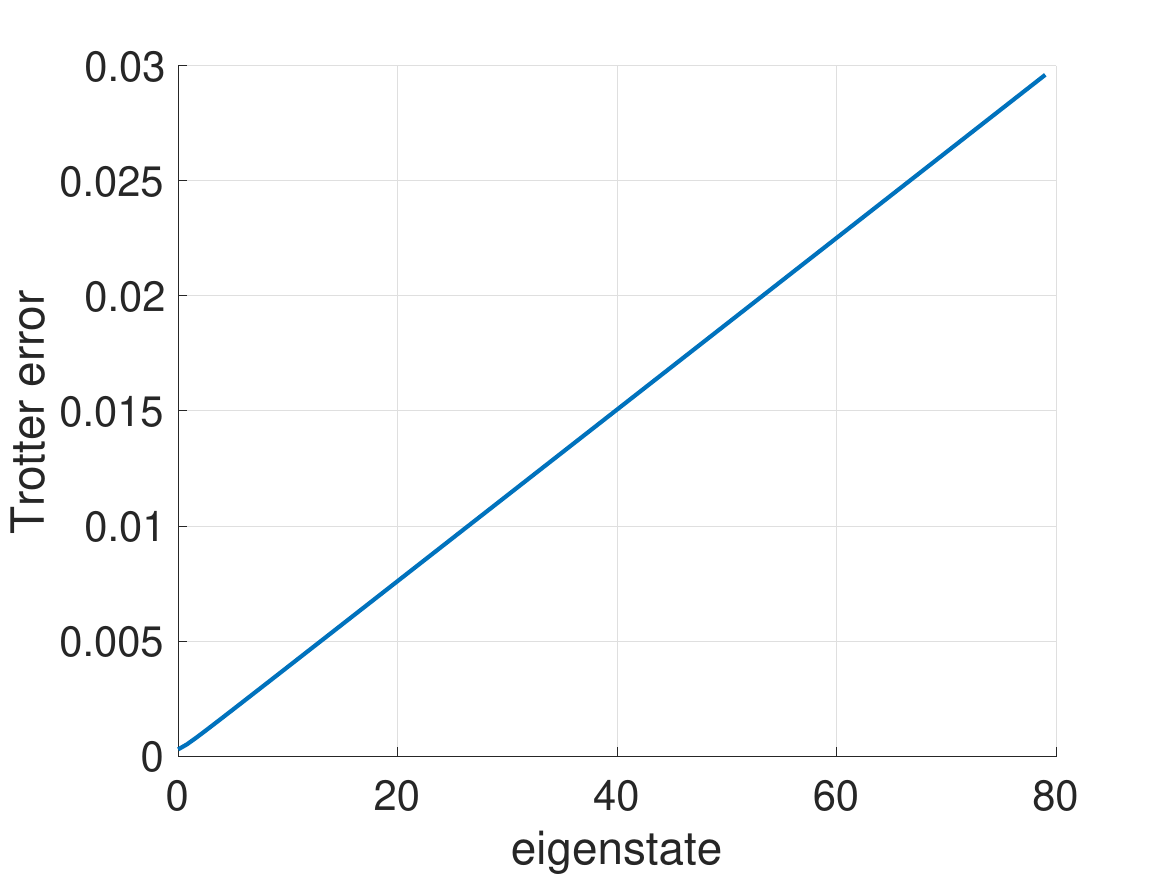}
    \includegraphics[width=7cm]{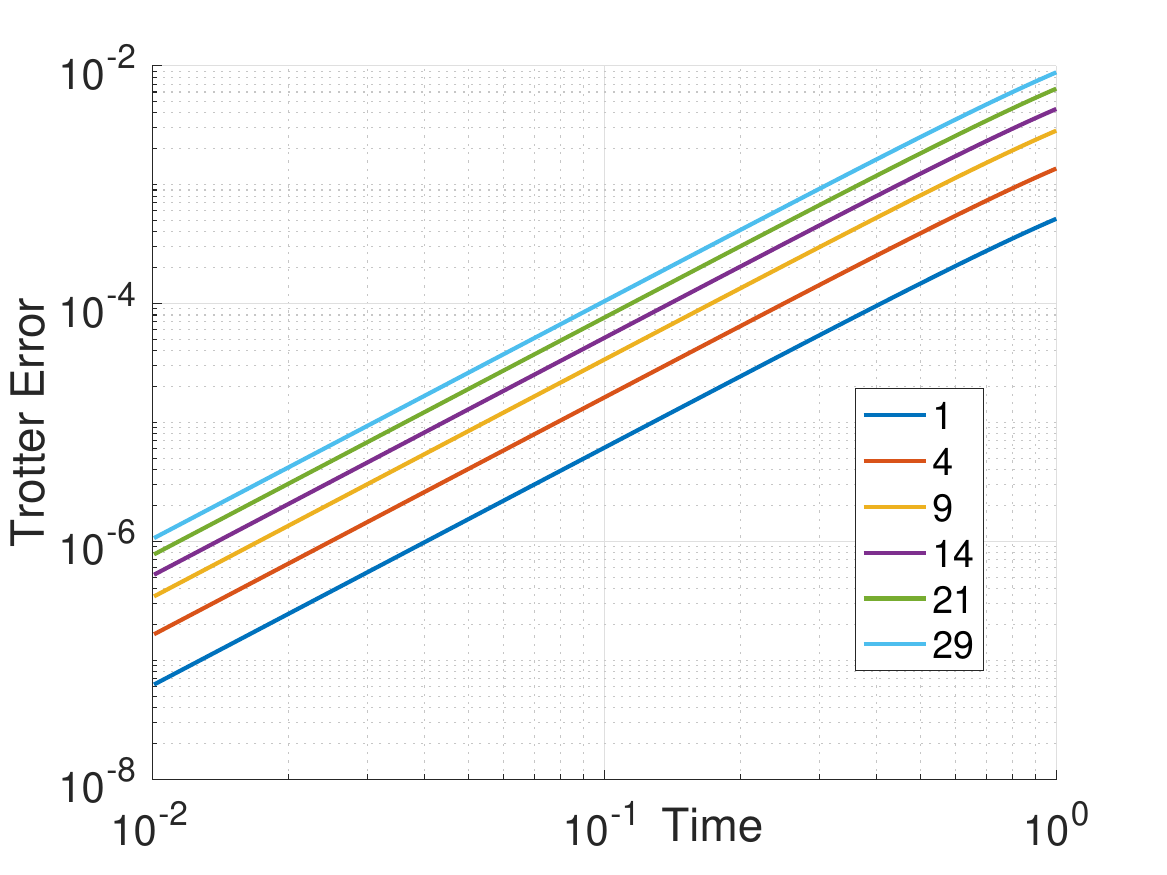}
    \caption{
\label{fig:Trotter4}
On the left, we see the Trotter error at $t=1$ for all first 81 quantum harmonic oscillator eigenstates (Fock states) after $N=1000$ Trotter iterations (linear scaling) with matrix truncation in $\mathbb C^{400}$. On the right, we see the time dependence of the Trotter error for eigenstates (Fock states) $\{1,4,9,14,21,29\}$. 
}
\end{figure}
\[a = \begin{pmatrix}
0 & \sqrt{1} & 0 & 0 & \dots & 0 & \dots \\
0 & 0 & \sqrt{2} & 0 & \dots & 0 & \dots \\
0 & 0 & 0 & \sqrt{3} & \dots & 0 & \dots \\
0 & 0 & 0 & 0 & \ddots & \vdots & \dots \\
\vdots & \vdots & \vdots & \vdots & \ddots & \sqrt{n} & \dots \\
0 & 0 & 0 & 0 & \dots & 0 & \ddots \\
\vdots & \vdots & \vdots & \vdots & \vdots & \vdots & \ddots \end{pmatrix}\]
with $Q = \frac{a+a^*}{\sqrt{2}}$ and $P=\frac{a-a^*}{\sqrt{2}i}$. The $\mathcal O(n^{-1})$ convergence rate, see Figure \ref{fig:Trotter3}  the linear error in the eigenvalues of the Fock states, see Figure \ref{fig:Trotter4} (left), and the quadratic time dependence, see Figure \ref{fig:Trotter4} (right) are correctly predicted by our estimate in \cref{thm:scheiss_oszillator}.

In practice, since computations involve finite matrices and vectors, the implementation error for approximating \( A + L \) behaves like \( \mathcal{O}(1/n) \), with a constant that depends on the truncation dimension. Numerical results confirm this decay rate, but the constant can be large, leading to outliers and a poor uniform convergence rate. Additionally, truncating to a finite dimension introduces errors due to the omission of high-frequency components, which do not appear at the full operator level. Therefore, we focused on the ground states of the respective Schrödinger operators to mitigate these finite-size effects when analyzing numerical errors.}

\appendix
\section{Proofs of \cref{thm:TrotterState} and \cref{prop:normal}
}
\label{app:ProofIntroTheo}
In this section, we prove \cref{thm:TrotterState} and \cref{prop:normal}.
\begin{proof}[Proof of Theorem~\ref{thm:TrotterState}]
First of all, note that by \eqref{eq:LRelK}, since the $A$-bound of $ L$ is strictly smaller than 1, and $( L, D( L))$ is dissipative \cite[Proposition II.3.23]{EN00}, we have that $(A+ L, D(A))$ generates a strongly continuous contraction semigroup \cite[Thm.~III.2.7]{EN00}.
    Hence, we can apply Lemma~\ref{lem:SemiTrottCommBound} which gives for $x\in D(A^2)$ 
\begin{align}
  \label{eq:proofEqComm}
			\left\|\left(\left(e^{t L/n}e^{tA/n}\right)^n - e^{t(A+ L)}\right)x\right\| \le t\sup_{s,\tau\in[0,t]}\left\| [ L,e^{sA/n} ]x_\tau\right\|,
\end{align}
  where we denoted $x_\tau =e^{\tau(A+ L)}x\in D(A^2)$ and used that since $ D(A^2) =  D((A+ L)^2)$ \cite [Lem.~IV.3]{P18}, the operator $e^{\tau(A+ L)}$ leaves $ D(A^2)$ invariant.
     Then, using again that $ L$ is relatively $A$-bounded, we have
     \begin{align*}
         \left\Vert  [ L,e^{sA/n} ]x_\tau\right\Vert
         &= \left\| [ L,e^{sA/n}- I ]x_\tau\right\|\\
         &\le \left\| L(e^{sA/n}- I)x_\tau\right\| + \left\|(e^{sA/n}- I) L x_\tau\right\|\\
         &\le a\left\|A(e^{sA/n}- I)x_\tau\right\| + b\left\|(   e^{sA/n}- I)x_\tau\right\| + \left\|(e^{sA/n}- I) L x_\tau\right\|.
     \end{align*}
     Furthermore, since for $y\in D(A)$ we have
     \begin{align*}
         \left\|(e^{sA/n}- I) y\right\| = \frac{s}{n}\left\|\int_0^{1} e^{rsA/n}A y \,dr\right\| \le \frac{s}{n}\left\|A y\right\|,
     \end{align*}
     we find, as $A$ and $e^{sA}- I$ commute, that
     \begin{align}
     \label{eq:xTauBound}
         \left\| [ L,e^{sA/n} ] x_\tau\right\|
         &\le \frac{s}{n} \left(a\left\|A^2 x_\tau\right\| + b\left\|A x_\tau\right\| + \left\|A L x_\tau\right\|\right),
     \end{align}
     where we have used that by assumption $A x_\tau, L x_\tau\in  D(A).$
     For the term in the middle, we use
       \[\begin{split} \|A x_\tau\| &\le \Vert (A+ L ) x_{\tau} \Vert + \Vert  L x_{\tau} \Vert \\
     &\le \Vert (A+ L) x \Vert + \Vert 
      L x_{\tau} \Vert \\
     &\le \Vert (A+ L) x \Vert + b\|x\| + a\Vert A x_{\tau}\Vert. \end{split}\]
     Using now $a<1$, we can bring $a\Vert A \psi_{\tau}\Vert$ to the other side and get a bound of the form 
     \begin{equation*}
     \label{eq:EasyTerm}
     \Vert A x_{\tau}\Vert \le \frac{1}{1-a}((1+a)\Vert A x \Vert+2b\|x\|).\end{equation*}
     By the same argument, we get
\begin{align}
\label{eq:H(H+G)Term}
     \norm{A(A+ L) x_\tau}
    &= \norm{A[(A+ L)x]_\tau}\\
    &\le \frac{1}{1-a}((1+a)\norm{ A (A+ L)x}+2b\|
    x\|).
\end{align}
Since $A^2 = A(A+ L) - A L$ on $ D(A^2)$ and by assumption of $A L$ being relatively $A^2$-bounded we find
\begin{align}
    \nonumber \norm{A^2 x_\tau}&\le \norm{A(A+ L) x_\tau} + \norm{A L x_\tau}\\
    &\nonumber\le \frac{1}{1-a}((1+a)\norm{ A(A+ L) x}+2b\|x\|) + \norm{A
     L x_\tau}\\
    &\le \frac{1}{1-a}((1+a+ a')\norm{ A^2 x}  +(2b +  b')\|x\|) +  a'\norm{A^2 x_\tau} +  b'\|x\|.
\end{align}
Using $ a'<1$ we get
\begin{equation}\label{eq:H2Term}
        \norm{A^2 x_\tau}\le \frac{(1+a+ a')\norm{ A^2 x} +(2b + (2-a) b')\|x\|}{(1-a)(1- a')}.
\end{equation}
Thus, using relative $A^2$-boundedness of $A L$ again in \eqref{eq:xTauBound}, we have
\begin{align*}
        &\norm{[ L, e^{sA/n}] x_\tau}
        \le \frac{s}{n} \left((a+ a')\left\|A^2 x_\tau\right\| + b\left\|A x_\tau\right\| +  b'\|x\|\right)\\
        &\le \frac{ s}{ n }\left((a+ a')\frac{(1+a+ a')\norm{ A^2 x} +2b + (2-a) b'}{(1-a)(1- a')} + \frac{b}{1-a}((1+a)\Vert A x \Vert+2b\|x\|) +  b'\|x\|
        \right)\\
        &=\frac s {n} \left(c_2\norm{ 
        A^2 x} + c_1\Vert A x \Vert + c_0\|x\|
        \right).
    \end{align*}
    where
    \begin{align*}
        c_0 &= \frac{(a+ a')[2b + (2-a) b']}{(1-a)(1- a')} + \frac{2b^2}{1-a} +  b', \quad c_1 = \frac{1+a}{1-a}b\\
        c_2 &= \frac{(a+ a')(1+a+ a')}{(1-a)(1- a')}.
    \end{align*}
    Plugging this into \eqref{eq:proofEqComm} finishes the proof of \eqref{eq:IntroTheoMain}.

\end{proof}
\begin{proof}[Proof of Prop.~\ref{prop:normal}]
From the elementary inequality 
 $\vert z \vert^{\alpha}=(\vert \Re(z) \vert + \vert \Im(z) \vert)^{\alpha} \le \vert \Re(z) \vert^{\alpha} + \vert \Im(z) \vert^{\alpha},$ we conclude upon integrating against the spectral measure that by using Lemma \ref{lemm:SA}
 \[ (F_{\alpha}(\Re(N))=D(\vert \Re(N)\vert^{\alpha})) \cap (F_{\alpha}(\Im(N))=D(\vert \Im(N)\vert^{\alpha})) \subset D(\vert N \vert^{\alpha}).\]
Similarly, 
\[ \vert \Re(z) \vert^{\alpha}\le\vert z \vert^{\alpha} \text{ and }\vert \Im(z) \vert^{\alpha}\le\vert z \vert^{\alpha} ,\]
implies that $D(\vert N \vert^{\alpha}) \subset D(\vert \Re(N) \vert^{\alpha})$ and $D(\vert N \vert^{\alpha}) \subset D(\vert \Im \vert^{\alpha})$ which shows that 
\[ D(\vert N \vert^{\alpha}) =(F_{\alpha}(\Re(N))=D(\vert \Re(N)\vert^{\alpha})) \cap (F_{\alpha}(\Im(N))=D(\vert \Im(N)\vert^{\alpha})). \]
 We conclude that if $x \in D(\vert N \vert^{\alpha})$,  then 
 \[\begin{split} \Vert (T_t -I)x \Vert &\le \Vert e^{it\Im(N)}(e^{t\Re(N)} -I)x\Vert + \Vert (e^{it\Im(N)} -I)x\Vert \\
 &=\Vert (e^{t\Re(N)} -I)x\Vert + \Vert (e^{it\Im(N)} -I)x\Vert = \mathcal O(t^{\alpha}). \end{split} \]
 Thus, $x \in F_{\alpha}(N).$ On the other hand, $F_{\alpha}(N)=F_{\alpha}(N^*)$ by normality and $e^{t\Re(N)}  = e^{tN/2}e^{tN^*/2}.$
 Thus, we conclude, since $\Vert e^{tN/2}\Vert \le 1,$ that 
 \[\begin{split} 
 \Vert (e^{t\Re(N)}-I)x\Vert &\le \Vert e^{tN/2}(e^{tN^*/2}-I)x \Vert + \Vert (e^{tN/2}-I)x \Vert \\
 &\le \Vert (e^{tN^*/2}-I)x \Vert + \Vert (e^{tN/2}-I)x \Vert=\mathcal O(t^{\alpha}),
 \end{split}\]
 which shows that $x \in F_{\alpha}(\Re(N)).$ All the previous inclusions follow very directly, yet the inclusion $F_{\alpha}(N) \subset F_{\alpha}(\Im(N))$ seems less obvious. Thus, for the final inclusion, we proceed as in Lemma \ref{lemm:SA}. Writing $t = t_1 + it_2$ for $t_i \in \mathbb R$ we find 

 \begin{equation}
 \label{eq:spectral_integral}
 \begin{split} \Vert \lambda^{r} N(\lambda-iN)^{-1}x \Vert^2 &= \lambda^{2r} \int_{\{\Re(t) \le 0\}}   \frac{t_1^2+t_2^2}{(\lambda-t_1)^2+t_2^2}d\langle E_N(t)x,x\rangle.
\end{split}
\end{equation}
Optimizing over $\lambda$ we find 
\[ \lambda_{\text{opt}} =\frac{t_1 (1 - 2 r) + \sqrt{t_1^2 + 4 t_2^2 (1 - r) r}}{2(1- r)},\]
such that we find the elementary upper bound
\[\lambda_{\text{opt}} \le \frac{-2r t_1 + 2 \vert t_2\vert \sqrt{(1-r)r}}{2(1-r)}\le  \max\{\tfrac{r}{1-r},\sqrt{\tfrac{r}{1-r}}\}\}\vert t \vert.  \]
By establishing a similar lower bound
\[ \lambda_{\text{opt}} \ge \sqrt{2}(r-r^2) \vert t \vert, \]we find that $\lambda_{\text{opt}}^{2r}$ in front of the integral in \eqref{eq:spectral_integral} behaves like $\vert t \vert^{2r}.$ 

Similar arguments show that the integrand
\[ \frac{\vert t \vert^2 }{(\lambda_{\text{opt}}-t_1)^2+t_2^2} = \frac{4(1-r)^2 \vert t \vert^2}{4(1-r)^2 t_2^2 + (\vert t_1\vert+\sqrt{t_1^2 + 4t_2^2 (1-r)r})^2}\]
is uniformly bounded from above and below, i.e., there are $K_r,k_r>0$ such that for all admissible $t \in (-\infty,0)\times i\mathbb R$
\[k_r \le \frac{\vert t \vert^2 }{(\lambda_{\text{opt}}-t_1)^2+t_2^2} \le K_r. \]
Thus, arguing as in the proof of the previous Lemma, we conclude that 
\[ F_{\alpha}(N) = D(\vert N \vert^{\alpha}).\]
\end{proof}


\begin{thebibliography}{0}
\bibitem[ACE23]{ACE23} W. Arendt, I. Chalendar, and R. Eymard, \emph{Extensions of dissipative and symmetric operators}, \textit{Semigroup Forum}, \textbf{106} (2023), 339–367.

\bibitem[AFL21]{AFL21} D. An, D. Fang, and L. Lin, \emph{Time-dependent unbounded Hamiltonian simulation with vector norm scaling}, \textit{Quantum}, \textbf{5} (2021), 459.

\bibitem[AKT14]{AKT14} W. Auzinger, O. Koch, and M. Thalhammer, \emph{Defect-based local error estimators for splitting methods, with application to Schr\"odinger equations, Part ii. Higher-order methods for linear problems}, Journal of Computational and Applied Mathematics 255,   (2014),384-403.
\bibitem[BD20]{BD20} S. Becker, N. Datta, \emph{Convergence rates for quantum evolution and entropic continuity bounds in infinite dimensions},  Communications in Mathematical Physics, 374, 2, 2020, 823-871.
\bibitem[B+21]{BDLR21}S. Becker, N. Datta, L. Lami, C. Rouzé, \emph{Energy-constrained discrimination of unitaries, quantum speed limits, and a Gaussian Solovay-Kitaev theorem}, Physical Review Letters 126 (19), 190504, 2021, DOI: https://doi.org/10.1103/PhysRevLett.126.190504.
\bibitem[BDS21]{BDS21}S. Becker, N. Datta, R. Salzmann, \emph{Quantum Zeno effect in open quantum systems}, Annales Henri Poincar\'e, 22 (11), 2021, 3795-3840.
\bibitem[B+07]{BACS07} D. W. Berry, G. Ahokas, R. Cleve, and B. C. Sanders, \emph{Efficient quantum
algorithms for simulating sparse Hamiltonians,} Communications in Mathematical Physics 270, no. 2,
2007, 359–371.
\bibitem[B+14]{B+14} D. W. Berry, A. M. Childs, R. Cleve, R. Kothari, and R. D. Somma,
\emph{Exponential improvement in precision for simulating sparse Hamiltonians}, Proceedings of the
46th Annual ACM Symposium on Theory of Computing, 2014, 283–292.
\bibitem[BCM24]{BCM24}S. Blanes, F. Casas, and  A. Murua, \emph{Splitting methods for differential equations},
Acta Numerica 33, 1-161, 2024.
\bibitem[BF19]{BF19} C. Budde,  B. Farkas, \emph{Intermediate and extrapolated spaces for bi-continuous operator semigroups.} Journal of Evolution Equations 19, 2019, 321-359. 
\bibitem[B+24a]{BFHJY23} D. Burgarth, P. Facchi, A. Hahn, M. Johnsson, K. Yuasa, \emph{Strong Error Bounds for Trotter \& Strang-Splittings and Their Implications for Quantum Chemistry}, Phys. Rev. Research 6, 043155, 2024, DOI: https://doi.org/10.1103/PhysRevResearch.6.043155.
\bibitem[B+23b]{BGHL23} D. Burgarth, N. Galke, A. Hahn, and L. van Luijk, \emph{State-dependent Trotter limits and their approximations}, Physical Review A 107, L040201, 2023, DOI: https://doi.org/10.1103/PhysRevA.107.L040201.
\bibitem[B+24]{BFHL24} D. Burgarth, P. Facchi, R. Hillier, M. Ligabò, \emph{Taming the Rotating Wave Approximation}, Quantum 8, 1262, 2024, DOI: https://doi.org/10.22331/q-2024-02-21-1262
\bibitem[BL22]{BL22}S. Bachmann and M. Lange, \emph{Trotter Product Formulae for *-Automorphisms of Quantum Lattice Systems}, Annales Henri Poincar\'e, 23, 2022, 4463–4487.
\bibitem[BM01]{BM01} H. Brezis and P. Mironescu. \emph{Gagliardo-Nirenberg, composition and products in fractional Sobolev spaces.} Journal of Evolution Equations, 4, 2001, 387-404.
\bibitem[BMW+15]{BMW+15} R. Babbush, J. McClean, D. Wecker, A. Aspuru-Guzik,
and N. Wiebe, \emph{Chemical basis of Trotter-Suzuki errors in
quantum chemistry simulation}, Physical Review A 91, 022311,
2015, DOI: https://doi.org/10.1103/PhysRevA.91.022311.
\bibitem[C68]{C68} P.R. Chernoff, \emph{Note on product formulas for operator semigroups},
Journal of Functional Analysis, 2, 1968, 238-242.
\bibitem[EN00]{EN00} K.-J. Engel and R. Nagel, \emph{One-parameter semigroups for linear evolution equations}, Graduate Texts in Mathematics, vol. 194, Springer-Verlag, New York, 2000.
\bibitem[EI05]{EI05} P. Exner, T. Ichinose, \emph{A product formula related to quantum Zeno dynamics},
Annales Henri Poincar\'e 6(2), 48, 2005, 195–215.
\bibitem[F82]{F82} R. P. Feynman, \emph{Simulating physics with computers,} International Journal of Theoretical Physics, 21,  1982, 467–488.
\bibitem[F77]{F77} J. Fröhlich, \emph{Application of commutator theorems to the integration of representations of {Lie} algebras and commutation relations}, Communications in Mathematical Physics, 54, 2, 1977, 135-150.
\bibitem[HO09]{HO09}E. Hansen and A. Ostermann, \emph{Exponential splittings for unbounded operators},
Mathematics of  Computation. 78, 2009, 1485-1496.
\bibitem[HHL09]{HHL09} A. W. Harrow, A. Hassidim, S. Lloyd, \emph{Quantum algorithm for linear systems of equations}. Physical Review Letters. 103 (15): 150502., 2009, DOI: https://doi.org/10.1103/PhysRevLett.103.150502
\bibitem[H05]{H05}N.J. Higham, \emph{The scaling and squaring method for the matrix exponential
revisited}, SIAM Journal of Matrix Analysis and Applications 26, 2005, 117-1193.
\bibitem[HNVW23]{HNVW23}
T. Hyt\"onen, J. van Neerven, M. Veraar, and L. Weis,
\emph{Analysis in Banach Spaces: Volume III. Harmonic Analysis and Spectral Theory},
Ergebnisse der Mathematik und ihrer Grenzgebiete, 3. Folge, vol.76, Springer, Cham, 2023.
\bibitem[I03]{I03} T. Ichinose, \emph{Time-Sliced Approximation to Path Integral and Lie-Trotter-Kato Product Formula}, A Garden of Quanta, World Scientific, 2003, 77-93.
\bibitem[IK24]{IK24} A. Iserles and K. Kropielnicka, \emph{An elementary approach to splittings of unbounded
operators}, Computers \& Mathematics with Applications  176, 2024, 29-34.
\bibitem[JL00]{JL00} T. Jahnke and C. Lubich, \emph{Error Bounds for Exponential Operator Splittings}, BIT Numerical Mathematics 40, 2000, 735-744. 

\bibitem[K78]{K78} T. Kato, \emph{Trotter's product formula for an arbitrary pair of self-adjoint contraction semigroups}, I. Gohberg (ed.) M. Kac (ed.), Topics in functional analysis, Acad. Press, 1978, 185-195.
\bibitem[K95]{K95} A. Y. Kitaev, \emph{Quantum measurements and the Abelian Stabilizer Problem}. arXiv:quant-ph/9511026, 1995, DOI: 
https://doi.org/10.48550/arXiv.quant-ph/9511026
\bibitem[K88]{K88} R. Kosloff, \emph{Time-dependent quantum mechanical methods for molecular dynamics}, Journal of Physical Chemistry, 92, 1988, 2087-2100.
\bibitem[LWG+96]{LWG+96}B. P. Lanyon, J. D. Whitfield, G. G. Gillett, M. E.
Goggin, M. P. Almeida, I. Kassal, J. D. Biamonte, M.
Mohseni, B. J. Powell, M. Barbieri, A. Aspuru-Guzik,
and A. G. White, \emph{Towards quantum chemistry on a quantum computer}, Nature Chemistry 2, 2010, 106-110.
\bibitem[LL20]{LL20} C. Lasser and C. Lubich, \emph{Computing quantum dynamics in the semiclassical
regime}, Acta Numerica 29, 2020, 229-401.
\bibitem[L96]{L96} S. Lloyd, \emph{Universal Quantum Simulators,} Science 273, 1996, 1073-1078.
\bibitem[vL25]{vL24} L. van Luijk, \emph{Energy-limited quantum dynamics}, Communications in Mathematical Physics, 406, 120 (2025).
\bibitem[vL+24]{LGHB24} L. van Luijk, N. Galke, A. Hahn, D. Burgarth, \emph{Error bounds for Lie Group representations in quantum mechanics},
Journal of Physics A: Mathematical and Theoretical 57 (10), 105301, 2024, DOI 10.1088/1751-8121/ad288b.
\bibitem[L09]{L09}A. Lunardi, Interpolation theory, second ed., Appunti. Scuola Normale Superiore di Pisa (Nuova Serie).Edizioni della Normale, Pisa, 2009.
\bibitem[MQ02]{MQ02}R. I. McLachlan and G.R.W. Quispel, \emph{Splitting methods}, Acta Numerica 11, 2002, 341–434.
\bibitem[M24]{M24}T. M\"obus, \emph{On Strong Bounds for Trotter and Zeno Product Formulas with Bosonic Applications}, Quantum 8, 1424, 2024, DOI https://doi.org/10.22331/q-2024-07-25-1424.
\bibitem[MR23]{MR23} T. M\"obus, C. Rouz\'e, \emph{Optimal convergence rate in the quantum Zeno effect for open quantum systems in infinite dimensions}, Annales Henri Poincar\'e 24 (5), 2023, 1617-1659. 
\bibitem[MW19]{MW19}T. Möbus, M. M. Wolf, \emph{Quantum Zeno effect generalized}, Journal of Mathematical Physics, 60, 5, 2019, DOI https://doi.org/10.1063/1.5090912.
\bibitem[NSZ18]{NSZ18}H. Neidhardt, A. Stephan, V. A. Zagrebnov,\emph{Operator-Norm Convergence of the Trotter Product Formula on Hilbert and Banach Spaces: A Short Survey}, Current Research in Nonlinear Analysis. Springer Optimization and Its Applications, vol.135, 2018, 229-247.
\bibitem[NSZ18b]{NSZ18b}H. Neidhardt, A. Stephan, V. A. Zagrebnov, \emph{Remarks on the operator-norm convergence of the Trotter product formula},  Integral Equations and Operator Theory, 90.15, 2018, DOI: https://doi.org/10.1007/s00020-018-2424-z
\bibitem[NSZ19]{NSZ19} H. Neidhardt, A. Stephan, V. A. Zagrebnov, \emph{Trotter Product Formula and Linear Evolution Equations on Hilbert Spaces}, Analysis and Operator Theory. Springer Optimization and Its Applications, vol. 146, 2019, 271-299.
\bibitem[NSZ20]{NSZ20} H. Neidhardt, A. Stephan, V. A. Zagrebnov, \emph{Convergence Rate Estimates for Trotter Product Approximations of Solution Operators for Non-autonomous Cauchy Problems}, Publications of the Research Institute for Mathematical Sciences, 56.1, 2020, 83-135.
\bibitem[N64]{N64} E. Nelson, \emph{Feynman Integrals and the Schr\"odinger Equation}. Journal of Mathematical Physics, 5 (3), 1964, 332-343.
\bibitem[P83]{Paz83} A. Pazy, \emph{Semigroups of Linear Operators and Applications to Partial
Differential Equations}, Applied Mathematical Sciences, vol. 44, Springer-
Verlag, 1983.
\bibitem[P18]{P18}M. Penz. \emph{Regularity for evolution equations with non-autonomous perturbations in Banach spaces}.
Journal of Mathematical Physics, 59(10), 2018, DOI: https://doi.org/10.1063/1.5011306
\bibitem[RS75]{RS75} M. Reed, B. Simon, \emph{Methods in modern mathematical Physics, Volume II. Fourier analysis, self-adjointness}, Elsevier, 1975.
\bibitem[S24]{S24} R. Salzmann, \emph{Quantitative Quantum Zeno and Strong Damping Limits in Strong Topology}, arXiv:2409.06469, 2024, DOI:
https://doi.org/10.48550/arXiv.2409.06469
\bibitem[S12]{S12} K. Schm\"udgen, \emph{Unbounded Self-adjoint Operators on Hilbert Space}, Springer Science \& Business Media, 432 p., 2012.
\bibitem[S94]{S94} Q. Sheng, \emph{Global error estimates for exponential splitting}, IMA Journal of Numerical Analysis, Volume 14, Issue 1, 1994, 27–56.
\bibitem[S18]{S18} M. E. Shirokov, \emph{On the Energy-Constrained Diamond Norm and Its Application in Quantum Information Theory}. Problems of Information Transmission 54.1, 2018, 20–33.
\bibitem[S19]{S19} M. E. Shirokov, \emph{On completion of the cone of completely positive linear maps with respect to the energy-constrained diamond norm}, Lobachevskii Journal of Mathematics, vol. 40 (10), 2019, 1549–1568.
\bibitem[S20]{S20} M. E. Shirokov, \emph{Operator E-norms and their use}, Matematicheskii Sbornik, Volume 211, Number 9, 2020, DOI 10.1070/SM9336.
\bibitem[T08]{Th08}  M. Thalhammer \emph{High-order exponential operator splitting methods for time-dependent Schr\"odinger equations}, SIAM Journal on Numerical Analysis 46, no. 4, 2008, 2022-2038.

\bibitem[T12]{Th12} M. Thalhammer \emph{Convergence analysis of high-order time-splitting pseudospectral methods for nonlinear Schr\"odinger equations}, SIAM Journal on Numerical Analysis 50, no. 6, 2012, 3231-3258.

\bibitem[T92]{Th92} B. Thaller, \emph{The Dirac Equation}. Texts and Monographs in Physics, Springer Verlag, Berlin, 1992.
\bibitem[T59]{T59}H. Trotter, \emph{On the product of semigroups of operators}, Proceedings of the American Mathematical Society, 10, 1959, 545-551.
\bibitem[W17]{W17} A. Winter, \emph{Energy-constrained diamond norm with applications to the uniform continuity of continuous variable channel capacities}, \arXiv{1712.10267}, 2017, DOI: https://doi.org/10.48550/arXiv.1712.10267.
\bibitem[ZNI24]{ZNI24} V. A. Zagrebnov, H. Neidhardt, T. Ichinose, \emph{Trotter-Kato Product Formulæ}, Springer, Berlin, 2024.
\end{thebibliography}
\end{document}